%% file: main.tex
\documentclass{CSML}
\pdfoutput=1

\usepackage{lastpage}
\newcommand{\dOi}{13(2:8)2017}
\lmcsheading{}{1--\pageref{LastPage}}{}{}%
{Jul.~28, 2016}{Jun.~08, 2017}{}

\usepackage{hyperref}
\hypersetup{hidelinks}

\usepackage[utf8]{inputenc}
\usepackage{url}
\usepackage{amssymb}
\usepackage{amsmath}
\usepackage{graphicx}
\usepackage{xspace}
\usepackage{color}
\usepackage{multirow}
\usepackage{tikz}
\usetikzlibrary{patterns,positioning,calc}
\usepackage{array,multirow,makecell}
\usepackage{proof}
\usepackage{prooftree}
\usepackage{thm-restate}
\usepackage{multirow}
\usepackage{float}
\usepackage{wrapfig}
\usepackage{cancel}

\input{header_comm}

\newtheorem{example}[thm]{Example}
\newtheorem{definition}[thm]{Definition}
\newtheorem{lemma}[thm]{Lemma}
\newtheorem{proposition}[thm]{Proposition}

\newtheorem{corollary}[thm]{Corollary}
\newtheorem{theorem}[thm]{Theorem}

\newcommand{\stef}[1]{#1}
\newcommand{\david}[1]{#1} 
\newcommand{\lucca}[1]{#1} 

\let\myMarginpar\marginpar
\renewcommand\marginpar[1]{\myMarginpar[\scriptsize #1]{\scriptsize #1}}

\begin{document}

\title{A reduced semantics for deciding trace equivalence}
\thanks{This work has received funding from the European Research Council
(ERC) under the European Union’s Horizon 2020 research and innovation
program (grant agreement No 714955-POPSTAR), as well as the ANR projects JCJC VIP ANR-11-JS02-006 and
    Sequoia ANR-14-CE28-0030-01.}

\author[D. Baelde, S. Delaune, and L. Hirschi]{David Baelde\rsuper a}	
\author[]{St\'ephanie Delaune\rsuper b}	
\author[]{Lucca Hirschi\rsuper c}	
\address{{\lsuper{a,c}}LSV, ENS Cachan \& CNRS, Université Paris-Saclay, France}	
\email{\{baelde,hirschi\}@lsv.ens-cachan.fr}  
\address{{\lsuper{b}}CNRS \& IRISA, France}
\email{stephanie.delaune@irisa.fr}  




\begin{abstract}
  Many privacy-type properties of security protocols can be modelled
   using trace equivalence properties in suitable process algebras. 
  It has been shown that such properties can be decided for interesting
  classes of finite processes (\ie without replication) by means of  symbolic execution
  and constraint solving.
  However, this does not suffice to obtain practical tools. Current
  prototypes suffer from a classical combinatorial explosion problem caused
  by the exploration of many interleavings in the behaviour of processes.
  M\"odersheim \emph{et al.}~\cite{ModersheimVB10} have tackled this problem
  for reachability properties using partial order reduction techniques. We revisit their work, generalize it and
  adapt it for equivalence checking. 
We obtain an optimisation in the form of a reduced symbolic semantics
that eliminates redundant interleavings on the fly. The obtained
partial order reduction technique has
been integrated in a tool called \apte.
We conducted complete benchmarks showing
dramatic improvements.
\end{abstract}

\maketitle

\input{intro}

\input{model}

\input{compression}

\input{symbolic}

\input{diff-stef}

\input{apte}

\input{related-work}

\input{conclu}

\paragraph*{Acknowledgements}
We would like to thank Vincent Cheval for interesting discussions and
comments, especially on Section~\ref{sec:apte}.

\bibliographystyle{abbrv}
\bibliography{reference}

\newpage
\appendix

\input{notations}

\input{app-compression}

\end{document}

%% file: header_comm.tex
\usepackage{centernot}
\usepackage{xcolor}
\usepackage{mathtools}
\usepackage{stmaryrd}

\usepackage{placeins}

\makeatletter
\newcommand{\rightdoublearrow}{%
  \rightarrow\mkern-10mu\protect\joinrel\rightarrow}

\newcommand{\rightdoublearrowfill@}
  {\arrowfill@\relbar\relbar\rightdoublearrow}
\newcommand{\mapstodoublearrowfill@}
  {\arrowfill@{\mapstochar\relbar}\relbar\rightdoublearrow}
\newcommand{\xrightdoublearrow}[2][]
  {\ext@arrow 3{15}59\rightdoublearrowfill@{#1}{#2}}
\newcommand{\xmapstodoublearrow}[2][]
  {\ext@arrow 3{15}59\mapstodoublearrowfill@{#1}{#2}}
\makeatother

\makeatletter
\newcommand{\xMapsto}[2][]{\ext@arrow 0599{\Mapstofill@}{#1}{#2}}
\def\Mapstofill@{\arrowfill@{\Mapstochar\Relbar}\Relbar\Rightarrow}
\makeatother

\makeatletter
\def\rightarrowfillstar@{\arrowfill@\relbar\relbar{\rightarrow\smash{^*}}}
\newcommand{\xrightarrowstar}[2][]{\ext@arrow
  0{13}{15}8\rightarrowfillstar@{#1}{#2}}
\newcommand{\lrstep}{\@ifstar{\xrightarrowstar}{\xrightarrow}}
\makeatother

\makeatletter
\def\Rightarrowfillstar@{\arrowfill@=={\Rightarrow\smash{^*}}}
\newcommand{\NewxRightarrowstar}[2][]{\ext@arrow
  0{13}{15}8\Rightarrowfillstar@{#1}{#2}}
\newcommand{\NewxRightarrow}[2][]{\ext@arrow
  0{13}{15}8\Rightarrowfill@{#1}{#2}}
\newcommand{\LRstep}{\@ifstar{\NewxRightarrowstar}{\NewxRightarrow}}
\makeatother

\let\oldlrstep\lrstep
\renewcommand{\lrstep}[1]{%
  \mathrel{\raisebox{-0.5pt}[5pt]{%
      $\oldlrstep{\raisebox{-1.2pt}[2pt][1pt]{\scalebox{0.7}{$#1$}}}$%
  }}%
}
\newcommand{\lrstepannot}[2]{%
  \mathrel{\raisebox{-0.5pt}[5pt]{%
      $\oldlrstep{\raisebox{-1.2pt}[2pt][1pt]{\scalebox{0.7}{$#1$}}}$%
  \raisebox{-1.5pt}{\scalebox{0.7}{$#2$}}%
  }}%
}
\renewcommand{\lrstepannot}[2]{\lrstep{#1}_{#2}}

\let\oldxmapsto\xmapsto
\renewcommand{\xmapsto}[1]{%
  \mathrel{\raisebox{-0.5pt}[5pt]{%
      $\oldxmapsto{\raisebox{-1.2pt}[2pt][1pt]{\scalebox{0.7}{$#1$}}}$%
  }}%
}
\newcommand{\xmapstoannot}[2]{\xmapsto{#1}_{#2}}
\newcommand{\xmapstoannotwo}[3]{\xmapsto{#1}_{#2}^{#3}}

\newcommand{\Bio}{\mathcal{B}}
\newcommand{\cb}{c_B}

\newcommand{\I}{{\mathrel{\mathcal{I}}}}
\newcommand{\Ind}{\mathrel{\mathcal{I}}}
\newcommand{\act}{\mathsf{act}}

\newcommand{\init}{\mathsf{init}}

\newcommand{\obse}{\mathrm{obs}}
\newcommand{\myio}{\mathsf{io}}

\newcommand{\C}{\mathcal{C}}
\newcommand{\cs}[2]{(#1;#2)}
\newcommand{\Set}{\mathcal{S}}
\newcommand{\eqi}{\mathop{{=}^{?}}}
\newcommand{\neqi}{\mathop{{\neq}^{?}}}
\newcommand{\dedi}[3]{#1 \mathop{\vdash^{?}_{#2}} #3} 
\newcommand{\Sol}{\mathsf{Sol}}

\newcommand{\eqdef}{\stackrel{\mathsf{def}}{=}}

\newcommand{\test}[3]{\mathtt{if}\ #1\ \mathtt{then}\ #2\ \mathtt{else}\ #3}
\newcommand{\testt}[2]{\mathtt{if}\ #1\ \mathtt{then}\ #2}
\newcommand{\Out}{\mathtt{out}}
\newcommand{\In} {\mathtt{in}}
\newcommand{\triple}[3]{(#1;#2;#3)}
\newcommand{\q}{\mathcal{Q}}
\newcommand{\p}{\mathcal{P}}
\newcommand{\E}{\mathsf{E}}
\newcommand{\proc}[2]{(#1;#2)}
\newcommand{\refer}{\triangleright}
\newcommand{\Ch}{\mathcal{C}}

\newcommand{\X}{\mathcal{X}}
\newcommand{\W}{\mathcal{W}}
\newcommand{\T}{\mathcal{T}}
\newcommand{\N}{\mathcal{N}}
\newcommand{\inp}[2]{\In(#1,#2)} 
\newcommand{\out}[2]{\Out(#1,#2)} 
\newcommand{\tr}{\mathsf{tr}}

\newcommand{\pair}[2]{\langle #1,#2 \rangle}
\newcommand{\ok}{\mathsf{ok}}
\newcommand{\ska}{\mathit{ska}}
\newcommand{\skb}{\mathit{skb}}
\newcommand{\start}{\mathsf{start}}
\newcommand{\aenc}[2]{\mathsf{aenc}(#1,#2)}
\newcommand{\adec}[2]{\mathsf{adec}(#1,#2)}
\newcommand{\pk}[1]{\mathsf{pk}(#1)}
\newcommand{\pub}{\mathsf{pub}}
\newcommand{\enc}[2]{\mathsf{enc}(#1,#2)}
\newcommand{\dec}[2]{\mathsf{dec}(#1,#2)}
\newcommand{\projl}[1]{\pi_1(#1)}
\newcommand{\projr}[1]{\pi_2(#1)}

\newcommand{\dom}{\mathrm{dom}}
\newcommand{\img}{\mathrm{img}}
\newcommand{\fv}{\mathit{fv}} 
\newcommand{\fvs}{\mathit{fv}^2} 
\newcommand{\fvp}{\mathit{fv}^1} 

\newcommand{\sint}[1]{\lrstep{#1}}  
\newcommand{\sintw}[1]{\LRstep{#1}} 
\newcommand{\sintf}[2]{\xrightdoublearrow{#1}_{#2}} 
\newcommand{\sintc}[1]{\lrstepannot{#1}{c}} 
\newcommand{\ssym}[1]{\xmapsto{#1}}         
\newcommand{\ssymf}[2]{\xmapstodoublearrow{#1}_{#2}} 
\newcommand{\ssymc}[1]{\xmapstoannot{#1}{c}}         
\newcommand{\ssymd}[1]{\xmapstoannot{#1}{r}}         

\newcommand{\eint}{\approx}    
\newcommand{\eintc}{\approx_c} 
\newcommand{\esym}{\approx^s}  
\newcommand{\esymc}{\approx^s_c}   
\newcommand{\esymdff}{\approx^s_r}

\newcommand{\inceint}{\sqsubseteq}
\newcommand{\inceintc}{\sqsubseteq_c}
\newcommand{\incesym}{\sqsubseteq^s}
\newcommand{\incesymc}{\sqsubseteq_c^s}
\newcommand{\incesymdff}{\sqsubseteq_r^s}

\newcommand{\statequiv}{\sim}

\newcommand{\och}{\prec}                   
\newcommand{\vect}[1]{\overrightarrow{#1}} 
\newcommand{\io}[3]{\mathtt{io}_{\mathit{#1}}(\vect{#2},\vect{#3})}
\newcommand{\iossvect}[3]{\mathtt{io}_{\mathrm{#1}}({#2},{#3})}

\newcommand{\Dep}[2]{\mathrm{dep}\left(#1,#2\right)} 
\newcommand{\AllDep}[1]{\mathrm{Deps}\left(#1\right)}
\newcommand{\depSym}{\mathbb{n}}                     
\newcommand{\constrd}[2]{#1 \depSym #2}       

\newcommand{\mysf}[1]{\ensuremath{\mathsf{#1}}}

\newcommand{\spec}{\mysf{Spec}\xspace}
\newcommand{\apte}{\mysf{Apte}\xspace}
\newcommand{\akiss}{\mysf{Akiss}\xspace}
\newcommand{\quadr}[4]{\left(#1;#2;#3;#4\right)}
\newcommand{\ps}{\mathsf{ps}}
\newcommand{\extof}[1]{\lceil {#1} \rceil}           
\newcommand{\ofext}[1]{\lfloor {#1} \rfloor}           
\newcommand{\csext}[3]{(#1;#2;#3)}
\newcommand{\as}{\mathbf{A}}
\newcommand{\bs}{\mathbf{B}}
\newcommand{\symbinclset}{\prec^{+}}
\newcommand{\symbeqset}{\sim^{+}}
\newcommand{\ssyma}[1]{\xmapstoannotwo{#1}{}{\mathsf{A}}}     
\newcommand{\ssymaone}[1]{\xmapstoannotwo{#1}{}{\mathsf{A1}}} 
\newcommand{\ssymatwo}{\xmapstoannotwo{}{}{\mathsf{A2}}}      
\newcommand{\ssymac}[1]{\xmapstoannotwo{#1}{c}{\mathsf{A}}}
\newcommand{\ssymaonec}[1]{\xmapstoannotwo{#1}{c}{\mathsf{A1}}}
\newcommand{\ssymad}[1]{\xmapstoannotwo{#1}{r}{\mathsf{A}}}
\newcommand{\eqapte}{\approx^{\mathsf{A}}}
\newcommand{\eqaptec}{\approx^{\mathsf{A}}_c}
\newcommand{\eqapted}{\approx^{\mathsf{A}}_r}

\newcommand{\refS}[1]{Axiom~\ref{#1}}
\newtheorem{axiom}{Axiom}

\newcommand{\ie}{\emph{i.e.}\xspace}
\newcommand{\eg}{\emph{e.g.}\xspace}

\newcommand{\ignore}[1]{}

\definecolor{grey}{rgb}{0.30,0.30,0.30} 
\definecolor{darkgrey}{rgb}{0.98,0.98,0.98} 
\definecolor{rouge_sombre}{rgb}{0.8,0,0}
\definecolor{rouge}{rgb}{0.8,0,0}
\definecolor{orange}{rgb}{1,0.5,0}
\definecolor{vert}{rgb}{0.2,0.7,0.2}

\definecolor{oran}{RGB}{255,50,0}

%% file: intro.tex
\section{Introduction}
\label{sec:intro}

Security protocols are widely used today to secure transactions that rely on 
public channels like the Internet, where malicious agents may listen to 
communications and interfere with them. 
Security has a different meaning depending on the underlying 
application. It ranges from the confidentiality of data (medical files, 
secret keys, etc.) to, \eg verifiability in electronic voting systems. Another 
example is the notion of privacy that appears in many contexts such as 
vote-privacy in electronic voting or untraceability in RFID technologies.

To achieve their security goals, security protocols rely on various
cryptographic primitives such as symmetric and asymmetric encryptions,
signatures, and hashes. Protocols also involve a high level of concurrency and
are difficult to analyse by hand. 
Actually, 
many protocols have been shown to be flawed several years after their
publication (and deployment). 
 For example, a flaw has been discovered in the Single-Sign-On
 protocol used, \eg by Google Apps. It has been shown that a malicious
 application could very easily get access to any other application
 (\eg Gmail or Google Calendar) of their users~\cite{BreakingGoogle-FMSE2008}. This flaw has
 been found when analysing the protocol using formal methods,
 abstracting messages by a term algebra and using the Avantssar
 validation platform~\cite{avantssar-2012}. Another example is a flaw on vote-privacy
 discovered during the formal and manual analysis of an electronic
 voting protocol~\cite{JCS2012-Ben}.

\smallskip{}

Formal symbolic methods have proved their usefulness for precisely analysing the 
security of protocols. Moreover, it
allows one to benefit
from machine 
support through the use of various existing techniques, ranging from
model-checking to resolution and rewriting techniques. Nowdays,
several verification tools are available, \stef{\eg
  \cite{BlanchetCSFW01,Cr2008Scyther,sysdesc-CAV05,meier2013tamarin,santiago2014formal}}.
\stef{A synthesis of decidability and undecidability results for
equivalence-based security properties, and an overview of existing
verification tools 
that may be used to verify equivalence-based security properties can
be found in~\cite{DH-survey16}.}

In order to design decision procedures,
a reasonable assumption is to bound the number of protocol sessions, thereby limiting the length of execution traces. 
Under such an hypothesis, a wide variety of model-checking approaches have been 
developed
 (\eg \cite{MS02,Tiu-csf10}), and
several tools are now available to automatically verify cryptographic 
protocols, {\eg \cite{spec,sysdesc-CAV05}}.
A major challenge faced here is that one has to account for infinitely
many behaviours of the attacker, who can generate arbitrary messages.
In order to cope with this prolific attacker problem and obtain decision
procedures, approaches based on symbolic semantics and constraint resolution
have been proposed~\cite{MS02,RT01}.
This has lead to tools for verifying reachability-based security properties 
such as confidentiality~\cite{MS02} or, more recently, equivalence-based
properties such as privacy~\cite{Tiu-csf10,cheval-ccs2011,CCK-esop12}.
Unfortunately, the resulting tools, especially those for checking
equivalence (\eg \apte~\cite{Cheval-tacas14}, \spec~\cite{Tiu-csf10}, \akiss~\cite{chadha2012automated}) have a very limited
practical impact because they scale badly. 
This is not surprising since they treat concurrency in a very naive
way, exploring all possible symbolic interleavings of concurrent
actions.

\medskip{}

\paragraph*{\bf Related work.}
In standard model-checking approaches for concurrent systems, the 
interleaving problem is handled using partial order reduction (POR)
techniques~\cite{Peled98}.
In a nutshell,
these techniques aim to effectively exploit the fact that
the order of execution of two 
independent (parallel) actions is irrelevant when checking
reachability.
The theory of partial order reduction is well developed in the
context of reactive systems verification
(\eg~\cite{Peled98,ModelCheckingBook,godefroid1996partial}). However,
as pointed out by Clarke \emph{et al.} in~\cite{ClarkeJM03},
POR techniques from traditional model-checking cannot be directly
applied in the context of security protocol verification.
Indeed, the application to security requires one to keep track of the
knowledge of the attacker, and to refer to this
knowledge in a meaningful way (in particular to know which messages
can be forged at some point to feed some input).
Furthermore, security protocol analysis does not rely on the internal
reduction of a protocol, but has to consider arbitrary execution
contexts (representing interactions with arbitrary, active attackers).
Thus, any input may depend on any output, since the attacker has
the liberty of constructing arbitrary messages from past outputs.
This results in a dependency relation which is \emph{a priori} very
large, rendering traditional POR arguments suboptimal, and calling for
domain-specific techniques.

In order to improve existing verification tools for security protocols,
one has to design POR techniques that integrate nicely with symbolic execution.
This is necessary to precisely deal with infinite, structured data.
In this task, we get some inspiration from Mödersheim {\it et
al.}~\cite{ModersheimVB10}, who design a partial order reduction
technique that blends well with symbolic execution in the context of
security protocols verification. However, we shall see that their
key insight is not fully exploited, and yields only a quite limited
partial order reduction. Moreover, they only consider reachability
properties (like all previous work on POR for security protocol verification)
while we seek an approach that is adequate for
model-checking equivalence properties.

\medskip{}

\paragraph*{\bf Contributions.}

In this paper, we revisit the work of~\cite{ModersheimVB10} to obtain a 
partial order reduction technique for the verification of equivalence 
properties.
Among the several definitions of equivalence that have been proposed,
we consider \emph{trace equivalence} in this paper:
two processes are trace equivalent
when they have the same sets of observable traces and, for each such trace,
sequences of messages outputted by the two processes
 are \emph{statically equivalent},
\ie indistinguishable for the attacker.
This notion is well-studied and several algorithms and tools support
it~\cite{BlanchetAbadiFournetJLAP08,chevalier10,Tiu-csf10,cheval-ccs2011,CCK-esop12}.
Contrary to what happens for reachability-based properties, trace equivalence 
cannot be decided relying only on the reachable states. The sequence of 
actions that leads to this state plays a role. Hence, extra precautions
have to be taken before discarding a particular interleaving: we have to 
ensure that this is done in both sides of the equivalence in a similar
fashion.
Our main contribution is an optimised form of equivalence that discards
a lot of interleavings, and
a proof that this reduced equivalence coincides with
trace equivalence.  Furthermore, our study brings
an improvement of the original technique~\cite{ModersheimVB10} that
would apply equally well for reachability checking.
On the practical side,
we explain how we integrated our partial order reduction into the state-of-the art 
tool \apte~\cite{cheval-ccs2011}, prove the correctness of this integration,
and provide experimental results showing dramatic improvements.
We believe that our presentation is generic enough to be easily adapted
for  other tools (provided that they are based on a forward symbolic exploration of traces
combined with a constraint solving procedure).
A big picture of the whole approach along with the new results is given in
Figure~\ref{fig:bp}.
Vertically, it goes from the regular semantics, to symbolic semantics and
\apte's semantics.
Those semantics have variants when our optimisations are applied or not: no optimisation, only compression
or compression plus reduction.

This paper essentially subsumes the conference paper that has been
published in 2014~\cite{BDH-post14}. However, we consider here
  a generalization of the semantics used
in~\cite{BDH-post14}. This generalization notably allows us to capture
the semantics used 
in~\apte, which allows us to formally prove the integration of our optimisations
in that tool.
In addition, this paper incorporates proofs of all the
results, additional examples, 
and an extensive related work section. Finally, it comes with
a solid implementation in the tool \apte~\cite{Cheval-tacas14}.

\medskip{}

\paragraph*{\bf Outline.}
In Section~\ref{sec:model}, we introduce our model for security
processes. 
We then consider the class of simple processes introduced in~\cite{CCD-tcs13},
with else branches and no replication.
Then we present two successive optimisations in the form of refined
semantics and associated trace equivalences.
Section~\ref{sec:compression} presents 
a \emph{compressed} semantics that limits interleavings by
executing blocks of actions. 
Then, by adapting well-known argument, this is lifted to a symbolic semantics in 
Section~\ref{sec:constraint-solving}.
Section~\ref{sec:diff} presents the \emph{reduced} semantics
which makes use of \emph{dependency constraints} to remove more
interleavings.
In Section~\ref{sec:apte}, we explain how
this reduced semantics has been integrated in the tool \apte, prove
its correcteness, and give some benchmarks obtained on several case
studies. Finally, Section~\ref{sec:relWork} is devoted to related work,
and concluding remarks are given in Section~\ref{sec:conclu}.
An overview of the different semantics we will define and the results relating them
is depicted in Figure~\ref{fig:bp}.
A table of symbols can be found in Appendix~\ref{sec:not}.

\input{bpFigure}

%% file: bpFigure.tex
\begin{figure}[ht]
  \centering
\begin{tikzpicture}
[node distance = 1cm, auto,font=\normalsize,
every node/.style={node distance=2.5cm},
sem/.style={rectangle, draw, 
  fill=black!10,
  inner sep=5pt, 
  text width=2cm,
  text badly centered, 
  minimum height=1.2cm, 
  font=\bfseries\footnotesize\sffamily},
sem-red/.style={rectangle, draw, 
  fill=red!10,
  inner sep=5pt, 
  text width=2cm,
  text badly centered, 
  minimum height=1.2cm, 
  font=\bfseries\footnotesize\sffamily},
tab/.style = {inner sep = 0pt, 
  rectangle,
  fill=white,
  inner sep=0pt, 
  text width=4cm,
  text badly centered, 
},
tab2/.style = {inner sep = 0pt, 
  rectangle,
  fill=white,
  inner sep=0pt, 
  anchor=west,
},
new/.style = {
  fill=black,
  color=black,
},
known/.style = {
  fill=black,
  color=black,
},
]

\node [sem] (ssymc) {$\ssymc{}$};
\node [sem, above of=ssymc] (sintc) {$\sintc{}$};
\node [sem, left=2.5cm of sintc,
] (sint) {$\sint{}$};
\node [sem, left=2.5cm of ssymc] (ssym) {$\ssym{}$};
\node [sem, right=2.5cm of ssymc] (ssymd) {$\ssymd{}$};
\node [sem, below=2cm of ssymc] (ssymac) {$\ssymac{}$};
\node [sem, left=2.5cm of ssymac] (ssyma) {$\ssyma{}$};
\node [sem, right=2.5cm of ssymac] (ssymad) {$\ssymad{}$};

\node [tab, above=2.9cm of ssym] (noOptim-tab) {No Optimization};
\node [tab, above=2.9cm of ssymc] (comp-tab) {Compression};
\node [tab, above=2.9cm of ssymd] (diff-tab) {Reduction};

\node [tab2, left=0.5cm of sint] (concrete-tab) {Concrete};
\node [tab2, left=0.5cm of ssym] (symbolic-tab) {Symbolic};
\node [tab2, left=0.5cm of ssyma] (apte-tab) {Apte};

\node [below=1cm of symbolic-tab] (dotted-left-right) {};
\node [left=0.75cm of dotted-left-right] (dotted-left) {};
\node [right=14.75cm of dotted-left](dotted-right) {};

\draw[new] (sint) -- node[below]
{$\eint=\eintc$} 
node[above]
{Theorem~\ref{theo:comp-soundness-completeness}}
++(sintc);

\draw[new] (sintc) -- node[left]
{Theorem~2}
node[right]
{$\eintc=\esymc$} 
++(ssymc);

\draw[new] (ssymc) -- node[left]
[pos=0.7]{Theorem~\ref{thm:eqaptec-esymc}}
node[right]
[pos=0.7]{$\esymc=\eqaptec$} 
++(ssymac);

\draw[new] (sintc) -- node[below]
[pos=0.5, sloped, anchor=center, below]{$\eintc=\esymdff$} 
node[above]
[pos=0.5, above=-0.35cm]
{\rotatebox{332}{Theorem~\ref{theo:reduced}}}
++(ssymd);

\draw[new] (ssymd) -- node[left]
[pos=0.7]{Theorem~\ref{thm:cor-com-apte-red}}
node[right]
[pos=0.7]{$\esymdff=\eqapted$} 
++(ssymad);

\draw[known] (sint) -- node[left]
{\cite{baudet-ccs2005,CCD-tcs13}}
node[right]
{$\eint=\esym$} 
++(ssym);

\draw[known] (ssym) -- node[left]
[pos=0.7]{\cite{CCD-tcs13}
}
node[right]
[pos=0.7]{$\esym=\eqapte$} 
++(ssyma);

\draw [dashed] (dotted-left) -- (dotted-right);

\end{tikzpicture} 
  \caption{Overview of the paper}
  \label{fig:bp}
\end{figure}

%% file: model.tex
\newcommand{\valid}{\mathit{valid}}
\newcommand{\M}{\mathcal{M}}
\newcommand{\St}{\mathit{st}}
\section{Model for security protocols}
\label{sec:model}

In this section, we introduce the cryptographic process calculus that
we will use to describe security protocols. This calculus is close to
the applied pi calculus~\cite{AbadiFournet2001}. 
We consider a semantics in the spirit of the one used
  in~\cite{BDH-post14} but we also allow to block some actions
depending on a validity predicate.  
This predicate can be chosen in such a way that no action is blocked, making the semantics 
as in~\cite{BDH-post14}.
It can also be chosen as in \apte as we eventually do in order to prove the integration of
our optimisations into this tool.

\subsection{Syntax}
\label{subsec:syntax}

A protocol consists of some agents communicating on a network.
Messages sent by agents are modeled using a term algebra.
We assume two infinite and disjoint sets of variables, $\X$ and $\W$.
Members of $\X$ are denoted $x$, $y$, $z$, whereas members of~$\W$ are
denoted~$w$ and used 
as \emph{handles} for previously
output terms.
We also assume a set~$\N$ of \emph{names}, which are used
for representing keys or nonces\footnote{
 Note that we
 do not have an explicit set of restricted (private) names.
 Actually, all
 names are restricted and public ones will be explicitly given to
 the attacker.},
and a signature $\Sigma$ consisting of a finite set of function symbols.
Terms are generated inductively from names, variables, and function 
symbols applied to other terms. For $S \subseteq \X\cup\W\cup\N$,
the set of terms built from~$S$ by applying function symbols in
$\Sigma$ is denoted by $\T(\Sigma, S)$. 
We write $\St(t)$ for the set of syntactic subterms of a term~$t$.
Terms in $\T(\Sigma, \N\cup\X)$ are denoted by
$u$, $v$, etc. while terms in $\T(\Sigma,\W)$ represent \emph{recipes}
(describing how the attacker built a term from the available outputs)
and are written $M$, $N$, $R$. 
We write $\fv(t)$ for the set of variables (from $\X$ or $\W$)
occurring in a term $t$. A term is \emph{ground} if it does not contain any
variable, \ie it belongs to $\T(\Sigma,\N)$.
One may rely on a sort system for terms, but its details
are unimportant for this paper. 

To model algebraic properties of cryptographic primitives, we consider
an equational theory $\E$. The theory will usually be generated from a
finite set of axioms enjoying nice properties (\eg convergence)
but these aspects are irrelevant for the present work.

\begin{example}
\label{ex:signature}
In order to model asymmetric encryption and pairing, we consider:
\[
\Sigma = \{ \aenc{\cdot}{\cdot},\; \adec{\cdot}{\cdot}, \;\pk{\cdot},
\;\pair{\cdot}{\cdot}, \;\projl{\cdot},\; \projr{\cdot}\}.
\]
To take into account the properties of
these operators, we consider the equational
theory~$\E_{\mathsf{aenc}}$ generated by the three following
equations:
\[
\adec{\aenc{x}{\pk{y}}}{y} =x, \;\;\;\;
\projl{\pair{x_1}{x_2}} = x_1, \mbox{ and } \;\;
\projr{\pair{x_1}{x_2}} = x_2. 
\]
For instance, we have
$\projr{\adec{\aenc{\pair{n}{\pk{ska}}}{\pk{skb}}}{skb}}
=_{\E_\mathsf{aenc}} \pk{ska}$.
\end{example}

  Our model is parameterized by a notion of \emph{message},
  intuitively meant to represent terms that can actually be communicated by
  processes. Formally, we assume a special subset of ground terms $\M$,
  only requiring that it contains at least one public constant.
Then, we say that a ground term $u$ is
  \emph{valid}, denoted $\valid_\M(u)$, whenever for any $v \in \St(u)$,
  we have that there exists $v' \in \M$ such that $v =_\E v'$.
This notion of validity will be imposed \stef{on communicated terms}.
As we shall see, $\M$ can be chosen in such a way that the validity
constraint allows us to discard some terms for which the computation of some parts fail.
Note that $\M$ can also be chosen to be the set of all ground terms, yielding a trivial validity
predicate that holds for all ground terms. The following developments are
parametrized by~$\M$.

\begin{example}
\label{ex:signature-a}
The signature used in \apte is $\Sigma=\Sigma_c\cup\Sigma_d$ where:
\[
  \begin{array}{rcl}
\Sigma_c &=& \Sigma_0\cup\{\mathsf{aenc}(\cdot,\cdot),\mathsf{pk}(\cdot),\mathsf{enc}(\cdot,\cdot), \mathsf{hash}(\cdot), \mathsf{sign}(\cdot,\cdot),\mathsf{vk}(\cdot), \langle \cdot,\cdot\rangle\}\\[2mm]
\Sigma_d &=& \{\mathsf{adec}(\cdot,\cdot), \mathsf{dec}(\cdot,\cdot), \mathsf{check}(\cdot,\cdot), \pi_1(\cdot),\pi_2(\cdot)\}
  \end{array}
\]
\stef{where $\Sigma_0$ may contain some additional user-defined function symbols.} 
The equational theory $\E_\apte$ of \apte is an \stef{extension} of the
theory~$\E_{\mathsf{aenc}}$ generated by adding the following
equations:
\[
\dec{\enc{x}{y}}{y} =x \;\; \quad\quad \;\;
\mathsf{check}(\mathsf{sign}(x,y),\mathsf{vk}(y)) = x
\]
The validity predicate used in the semantics of~\apte is
obtained by taking $\M = \T(\Sigma_c,\N)$, \stef{\ie the ground terms built
using constructor symbols.}
This choice allows us to discard terms for which a failure will happen during the
computation and which therefore do not correspond to a message:
\eg $\projl{\langle \ok, \dec{\enc{a}{k}}{k'}\rangle}$ is not valid
since $\dec{\enc{a}{k}}{k'}$ is not equal modulo $\E_\apte$ to a term in
$\T(\Sigma_c,\N)$.
\end{example}

We do not need the full applied pi calculus \cite{AbadiFournet2001}
to represent security protocols.  Here, 
we only consider public channels and we assume that each process communicates on
a dedicated channel.
Formally, we assume a set $\Ch$ of \emph{channels} and we consider the fragment of
\emph{simple processes} without replication built on \emph{basic processes} as defined in~\cite{CCD-tcs13}.
A basic process represents a party in a protocol, which may sequentially
perform  actions such as waiting for a message, checking that a message has a 
certain form, or outputting a message.
Then, a simple process is a parallel composition of such basic processes
playing on distinct channels.

\begin{definition}[basic/simple process]
\label{def : basic process}
The set  of \emph{basic processes} 
on ${c \in \Ch}$
is defined using the following grammar (where $u,v\in\T(\Sigma,\N\cup\X)$ and
$x \in \X$):
\[
\begin{array}{lcll}
  P,Q &:=& 0                            &\mbox{null} \\
      &|& \test{u=v}{P}{Q}  \quad\quad &\mbox{conditional}\\
      &|& \In(c,x).P                   &\mbox{input}\\
      &|& \Out(c,u).P                  &\mbox{output}\\
\end{array}
\]
 A \emph{simple process} $\p = \{ P_1, \ldots, P_n \}$
 is a {multiset} of basic processes $P_i$ on pairwise
 distinct channels $c_i$. {We assume that null processes are removed.}
\end{definition}

Intuitively, a multiset of basic processes denotes a parallel composition.
For conciseness, we often omit brackets, null processes, and even
``\texttt{else} 0''.
Basic processes are denoted by the letters~$P$ and~$Q$, whereas simple
processes are denoted using~$\p$ and~$\q$.

\smallskip{}

During an execution, the attacker learns the messages that have been
sent on the different public channels. Those messages are organized
into a \emph{frame}. 

\begin{definition}[frame]
  A \emph{frame} $\Phi$ is a substitution whose domain is included in
  $\W$ and image is included in $\T(\Sigma, \N\cup\X)$. It is written
  $\{ w \refer u, \ldots \}$. A frame is \emph{ground} when
  its image only contains ground terms.
\end{definition}
In the \stef{remainder} of this paper, we will actually only consider ground
frames that are made of valid terms.

An \emph{extended simple process} (denoted $A$ or $B$) is a pair made of a simple process
and a frame.  Similarly, we define \emph{extended basic processes}.
When the context makes it clear, we may omit ``extended'' and simply call them {\em simple processes}
and {\em basic processes}.

\begin{example}
\label{ex:private}
We consider the protocol given
in~\cite{AbadiF04} designed for authenticating an agent with another one without revealing their
identities to other participants.
In this protocol, $A$ is willing to engage in
communication with~$B$ and wants to be sure that she is indeed talking to $B$ and not
to an attacker who is trying to impersonate $B$.
However, $A$ does not want to compromise her privacy by revealing her identity
or the identity of~$B$ more broadly. The participants~$A$ and~$B$
proceed as follows:
\[
\begin{array}{rcl}
 A \rightarrow B& \; :\; & \{N_a,\pub_A\}_{\pub_B}\\
 B \rightarrow A& \;: \;& \{N_a,N_b,\pub_B\}_{\pub_A}
 \end{array}
\]
First~$A$ sends to~$B$ a nonce~$N_a$ and her public key encrypted with
the public key of~$B$. If the message is of the expected form then~$B$
sends to~$A$ the nonce~$N_a$, 
a freshly generated nonce~$N_b$ and his public key, all of this being
encrypted with the public key of~$A$. 
Moreover, if the message received by $B$ is not of the expected form
then~$B$ sends out a ``decoy'' message: $\{N_b\}_{\pub_B}$. This message should basically look like~$B$'s other
 message from the point of view of an outsider.

Relying on the signature and equational theory introduced in
Example~\ref{ex:signature}, a session of role~$A$ played by agent~$a$
(with private key $ska$)
with~$b$ (with public key $pkb$) can be modeled as follows:
\[
\begin{array}{lcl}
P(\ska, pkb) &\eqdef&
\Out(c_A, \aenc{\pair{n_a}{\pk{ska}}}{pkb}). \\
&& \In(c_A,x). \\ 
&& \testt{\pair{\projl{\adec{x}{ska}}}{\projr{\projr{\adec{x}{ska}}}} =
  \pair{n_a}{pkb}}{0}
\end{array}
\]
Here, we are only 
  considering the authentication protocol.
A more comprehensive model should include the access to an
application in case of a success.
Similarly, a session of role~$B$ played by agent~$b$ with~$a$ can be
modeled by the following basic process, where $N = \adec{y}{skb}$.
\[
\begin{array}{lcl}
Q(skb, pka) &\eqdef & \In(\cb,y) . \\
&& \mathtt{if}\ \projr{N} = pka\;  \mathtt{then} \; \Out(\cb, \aenc{\pair{\projl{N}}{\pair{n_b}{\pk{skb}}}}{pka}) \\
&&\phantom{\mathtt{if}\ \projr{N} = pka}\; {\mathtt{else} \; \Out(\cb, \aenc{n_b}{\pk{skb}})} 
\end{array}
\]

To model a scenario with one session of each role (played by the
agents~$a$ and~$b$), we may consider the extended process
$\proc{\p}{\Phi_0}$ where:
\begin{itemize}
\item $\p
\eqdef \{P(ska, \pk{skb}), Q(skb, \pk{ska})\}$, and
\item 
 $\Phi_0 \eqdef \{w_0 \refer \pk{ska'}, \; w_1 \refer \pk{ska}, \;  w_2
\refer \pk{skb}\}$. 
\end{itemize}
The purpose of $\pk{\ska'}$ will be clear later on. It allows us to 
consider
the existence of
another agent $a'$ whose public key $\pk{ska'}$ is known by the attacker.
\end{example}

\subsection{Semantics}
\label{subsec:semantics}

We first define a standard concrete semantics using a relation over 
\emph{ground} extended simple processes, \ie extended simple processes
$\proc{\p}{\Phi}$ such that 
$\fv(\p) = \varnothing$ \david{(as said above, we also assume
that $\Phi$ contains only valid ground terms)}.
The semantics of a ground extended simple process
$\proc{\p}{\Phi}$ is induced by the relation $\sint{a}$ over ground
extended simple processes as defined in Figure~\ref{fig:concrete-semantics}.

\begin{figure}[h]
\[
\begin{array}{lrcl}
\mbox{\sc In} &
\proc{\{\In(c,x).Q\}\uplus\p}{\Phi} 
& \sint{\In(c,M)} & 
\proc{\{Q\{x\mapsto u\}\}\uplus\p}{\Phi}
\\ \multicolumn{4}{r}{\hspace{5cm} \mbox{if $M \in \T(\Sigma,\dom(\Phi))$,
 $\valid(M\Phi),\valid(u)$ and $M\Phi=_\E u$}}
\\[2mm]
 \mbox{\sc Out}
&
 \proc{\{\Out(c,u).Q\}\uplus\p}{\Phi}
 & \sint{\Out (c,w)} & \proc{\{Q\}\uplus\p}{\Phi\cup\{w\refer u\}}\\
\multicolumn{4}{r}{\mbox{if $w$ is a fresh variable, and $\valid(u)$}}
\\[2mm]
 \mbox{\sc Then} \;\;\;\;\;&
  \multicolumn{3}{l}{\proc{\{\test{u=v}{Q_1}{Q_2}\}\uplus\p}{\Phi} 
\; \sint{\;\tau\;}\; \proc{\{Q_1\}\uplus\p}{\Phi}} \\
\multicolumn{4}{r}{\mbox{if $u =_{\E} v$, $\valid(u)$, and $\valid(v)$}}\\[2mm]
\mbox{\sc Else} &
 \multicolumn{3}{l}{\proc{\{\test{u=v}{Q_1}{Q_2}\}\uplus\p}{\Phi}
 \; \sint{\;\tau\;} \; 
 \proc{\{Q_2\}\uplus\p}{\Phi}} \\
\multicolumn{4}{r}{\mbox{if $u \neq_{\E} v$ or $\neg \valid(u)$
  or $\neg \valid(v)$}}\\[2mm]
\multicolumn{4}{l}{\text{where } c\in\mathcal{C}, w\in\mathcal{W}\text{ and }x\in\mathcal{X}.}\\
\end{array}
\]
\caption{Concrete semantics}
\label{fig:concrete-semantics}
\end{figure}

A process may input any valid term that an attacker can build (rule {\sc
  In}): $\{x \mapsto u\}$ is a substitution that replaces any occurrence
of $x$ with $u$. 
Once a recipe $M$ is fixed, we may note that there are still
different instances of the rule, but only in the sense that $u$ is
chosen modulo the equational theory~$\E$.
In practice, of course, not all such $u$ are enumerated. How this
is achieved in practice is orthogonal to the theoretical development
carried out here.
In the {\sc Out} rule, we enrich the attacker's knowledge by adding the newly 
output term~$u$, with a fresh handle~$w$, to the frame.
The two remaining rules are unobservable ($\tau$ action) from the point of 
view of the attacker. 
When $\mathcal{M}$ contains all the ground terms, $\valid(u)$
  is true for any term $u$ and this semantics coincides with the one
  defined in~\cite{BDH-post14}. However, this parameter gives us enough
  flexibility to obtain a semantics similar to the one used in \apte,
  and therefore formally prove in Section~\ref{sec:apte} how to
  integrate our techniques in \apte.

 The relation $A \sint{a_1 \cdot \ldots \cdot a_k} B$
 between extended simple processes,
 where $k\geq 0$ and
 each $a_i$ is an observable or a $\tau$ action,
is defined in the usual way. We also consider the relation $\sintw{\;\tr}$
defined as follows: $A \sintw{\;\tr} B$ if, and only if, there exists
$a_1 \cdot \ldots \cdot a_k$ such that $A \sint{a_1 \cdot \ldots \cdot a_k} B$, and $\tr$ is
obtained from $a_1 \cdot \ldots \cdot a_k$ by erasing all occurrences of
$\tau$.

\begin{example}
\label{ex:semantics}
Consider the simple process $\proc{\p}{\Phi_0}$
introduced in Example~\ref{ex:private} (with $\M$ equal to
  $\T(\Sigma_c,\N)$ as in \apte). We have:
\[
\proc{\p}{\Phi_0}
 \sint{\Out(c_A,w_3) \cdot \In(\cb,w_3) \cdot \tau \cdot
   \Out(\cb,w_4) \cdot \In(c_A,w_4) \cdot \tau}
 \proc{\varnothing}{\Phi}.
\]

This trace corresponds to the normal execution of one instance of the
protocol. The two silent actions have been triggered using the {\sc
  Then} rule. The resulting frame $\Phi$ is as follows:
\[
\Phi_0 \uplus \{w_3 \refer
\aenc{\pair{n_a}{\pk{ska}}}{\pk{skb}}, \;w_4 \refer
\aenc{\pair{n_a}{\pair{n_b}{\pk{skb}}}}{\pk{ska}}\}.
\]
\end{example}


\subsection{Trace equivalence}
\label{subsec:equivalence}

Many interesting security properties, such as privacy-type
properties studied, \eg in~\cite{arapinis-csf10}, are formalized
using the notion of \emph{trace equivalence}. 
Before defining trace equivalence, 
we first introduce the notion of \emph{static
equivalence}  that compares sequences of messages. 

\begin{definition}[static equivalence] \label{def:statequiv}
Two frames $\Phi$ and $\Phi'$ are in \emph{static equivalence}, $\Phi
\statequiv \Phi'$, when we have that $\dom(\Phi) = \dom(\Phi')$,
and:
\begin{itemize}
\item $\valid(M\Phi)\;\;\Leftrightarrow \;\; \valid(M\Phi')$ for any
  term $M \in \T(\Sigma,\dom(\Phi))$; and
\item $M\Phi =_\E N\Phi \;\;\Leftrightarrow \;\; M\Phi' =_\E N\Phi'$
for any terms $M, N \in \T(\Sigma,\dom(\Phi))$ such that
$\valid(M\Phi)$ and $\valid(N\Phi)$.
\end{itemize}
\end{definition}

Intuitively, two frames are equivalent if an attacker cannot see the
difference between the two situations they represent, \ie   they satisfy the same equalities
and failures.

\begin{example}
\label{ex:static}
Consider the frame $\Phi$ given in Example~\ref{ex:semantics} and the
frame $\Phi'$ below:
\[
\Phi' \eqdef \Phi_0 \uplus \{\; w_3 \refer
\aenc{\pair{n_a}{\pk{ska'}}}{\pk{skb}}, \;\;w_4 \refer
\aenc{n_b}{\pk{skb}} \;\}.
\]
We have that $\Phi \statequiv \Phi'$.  
This is a non-trivial equivalence. 
Intuitively, it
holds since the attacker is not able to decrypt any of the
ciphertexts, and each ciphertext contains a nonce that prevents him to
build it from its components. 

Now, if we decide to give access to $n_a$ to the attacker, \ie considering 
$\Phi_+ = \Phi \uplus \{w_5 \refer n_a\}$ and $\Phi'_+ = \Phi' \uplus \{w_5
\refer n_a\}$, then the two frames $\Phi_+$ and $\Phi'_+$ are not in
static equivalence anymore as witnessed by $M = \aenc{\pair{w_5}{w_1}}{w_2}$ and
$N = w_3$. Indeed, we have that $M\Phi_+ =_{\E_\mathsf{aenc}} N\Phi_+$ whereas
${M\Phi'_+ \neq_{\E_{\mathsf{aenc}}} N\Phi'_+}$, and all these 
  witnesses are
valid.
\end{example}

\begin{definition}[trace equivalence] \label{def:concrete-equivalence}
Let~$A$ and~$B$ be two extended simple processes. We have that $A \inceint B$ if,
for every
sequence of actions $\tr$ such that $A \sintw{\;\tr} \proc{\p}{\Phi}$, there
exists $\proc{\p'}{\Phi'}$ such that $B \sintw{\;\tr}
\proc{\p'}{\Phi'}$ and $\Phi \statequiv \Phi'$. 
The processes~$A$ and~$B$ are 
\emph{trace equivalent}, denoted by $A \eint B$, if $A \inceint B$
and $B \inceint A$.
\end{definition}

\begin{example}
\label{ex:trace-equiv}
 Intuitively, the private authentication protocol presented in
 Example~\ref{ex:private} 
preserves anonymity if an attacker cannot
 distinguish whether $b$ is willing to talk to $a$ (represented by the
 process $Q(\skb,\pk{\ska})$) or willing to talk to $a'$  (represented by the
 process $Q(\skb,\pk{\ska'})$), provided $a$, $a'$ and~$b$ are honest participants.
This can be expressed relying on the following equivalence:
\[
\proc{ Q(\skb,\pk{\ska})}{\Phi_0} \stackrel{?}{\eint}
\proc{Q(\skb,\pk{\ska'})}{\Phi_0}.
\]

For illustration purposes,  we also consider a variant of the process~$Q$, 
denoted~$Q_0$, where its $\mathtt{else}$ branch has been
  replaced by $0$ (\ie the null process). We will see that the ``decoy''
 message plays a crucial role to ensure privacy.
We have that:
\[
\proc{Q_0(\skb,\pk{\ska})}{\Phi_0} \sint{\In(\cb,\aenc{\pair{w_1}{w_1}}{w_2}) \cdot \tau
  \cdot  \Out(\cb,w_3)} \proc{\varnothing}{\Phi}
\]
where $\Phi = \Phi_0 \uplus \{w_3 \refer
\aenc{\pair{\pk{ska}}{\pair{n_b}{\pk{skb}}}}{\pk{ska}}\}$.
\stef{We may note that this trace does not correspond to a normal execution
of the protocol. Still, the first input is fed with the message
$\aenc{\pair{\pk{ska}}{\pk{ska}}}{\pk{skb}}$ which is a message of
the expected format from the point of view of the
process $Q_0(skb,\pk{ska})$. Therefore, \lucca{once conditionals are positively evaluated,
the output $\Out(\cb,w_3)$ can be triggered.}}

This trace has no counterpart in $\proc{Q_0(\skb,\pk{\ska'})}{\Phi_0}$. Indeed, we have
that:
\[
\proc{Q_0(\skb,\pk{\ska'})}{\Phi_0} \sint{\In(\cb,\aenc{\pair{w_1}{w_1}}{w_2}) \cdot \tau}
\proc{\varnothing}{\Phi_0}.
\]
Hence, we
have that $\proc{Q_0(\skb,\pk{\ska})}{\Phi_0} \not\eint
\proc{Q_0(\skb,\pk{\ska'})}{\Phi_0}$.

However, it is the case that $\proc{Q(\skb,\pk{\ska})}{\Phi_0} \eint
\proc{Q(\skb,\pk{\ska'})}{\Phi_0}$.
This equivalence
can be checked using the tool \apte~\cite{APTE} within few seconds for a
simple scenario as the one considered here, and that takes few
minutes/days as soon as we want to consider~2/3 sessions of each
role.
\end{example}

%% file: compression.tex
\section{Compression based on grouping actions}
\label{sec:compression}

Our first refinement of the semantics,
which we call compression, is closely related to
focusing from proof theory~\cite{Andreoli92}: we will assign a polarity to
processes and constrain the shape of executed traces based on those 
polarities. This will provide a first significant reduction
of the number of traces to consider when checking equivalence-based
properties between simple processes.
Moreover, 
compression can easily be used as a replacement for the usual semantics in verification algorithms. 

\smallskip{}

The key idea is to force processes
to perform all enabled output actions as soon
as possible. In our setting, we can even safely force them
to perform a complete \emph{block} of input actions
followed by ouput actions.

\begin{example}
Consider the process $\proc{\p}{\Phi}$ with 
$\mathcal{P} = \{\In(c_1,x).P_1,\ \Out(c_2,b).P_2\}$.
In order to reach $\proc{\{P_1\{x \mapsto u\},\ P_2\}}{\Phi \cup \{w
  \refer b\}}$, we have to execute the action
$\In(c_1,x)$ (using a recipe $M$ that allows one to deduce $u$)
and the action $\Out(c_2,b)$ (giving us a label of the form $\Out(c_2,w)$). 
In case of reachability properties, the execution order of these
actions only 
matters if $M$ uses $w$.  Thus we can safely perform the outputs \stef{in priority}.

The situation is more complex when considering trace equivalence.
In that case, we are concerned not only with reachable states, but also
with \emph{how} those states are reached. Quite simply, traces matter.
Thus, if we want to discard the trace $\In(c_1,M).\Out(c_2,w)$ when studying
process $\mathcal{P}$ and consider only its permutation
$\Out(c_2,w). \In(c_1,M)$, we have to 
make sure that the same permutation is available on the other process.
The key to ensure that identical permutations will be available on both
sides of the equivalence is our restriction to the class of simple processes.
\end{example}


\subsection{Compressed semantics}
\label{subsec:comp-semantics}

We now introduce the compressed semantics. Compression is an 
optimisation, since it removes some interleavings. But it also gives rise to 
convenient ``macro-actions'', called \emph{blocks}, that combine a sequence of 
inputs followed by some outputs, potentially hiding silent actions.
Manipulating those blocks rather than indiviual actions
makes it easier to define our second optimisation.

\smallskip{}

For sake of simplicity, we consider \emph{initial} simple
processes. A simple
process $A = \proc{\p}{\Phi}$ is \emph{initial} if for any $P \in \p$,
we have that $P = 0$,
$P = \In(c,x).P'$ or $P = \Out(c,u).P'$ for some term $u$ such
that $\neg \valid(u)$.
In other words, each basic process composing $A$ starts with an input
unless it is blocked due to an unfeasible output.

\begin{example}
\label{ex:initial}
Continuing Example~\ref{ex:private}, $\proc{\{P(\ska, \pk{\skb}),Q(\skb,\pk{\ska})\}}{\Phi_0}$ is not initial. 
Instead, we may consider $\proc{\{P_\init,
  Q(\skb,\pk{\ska})\}}{\Phi_0}$ where
\[
P_\init \eqdef \In(c_A,z). \testt{z = \start}{P(\ska,\pk{\skb})}
\]
assuming that $\start$ is a (public) constant in our signature.
\end{example}

\begin{figure}[h]
  \[
  \begin{array}{lc}
\mbox{\sc In} &  \infer[\mbox{with }\ell\in\{i^*;i^+\}]
  {\proc{P}{\Phi}\sintf{\;\In(c,M) \cdot \tr\;}{\ell} \proc{P''}{\Phi''}}
  {\proc{P}{\Phi}\sint{\In(c,M)} \proc{P'}{\Phi'}
  & \proc{P'}{\Phi'}\sintf{\;\;\tr\;\;}{i^*} \proc{P''}{\Phi''}}
   \\ [2mm]
\mbox{\sc Out} & 
 \infer[\mbox{with } \ell\in\{i^*;o^*\}]
   {\proc{P}{\Phi}\sintf{\;\Out(c,w) \cdot \tr\;}{\ell} \proc{P''}{\Phi''}}
   {\proc{P}{\Phi}\sint{\Out(c,w)} \proc{P'}{\Phi'}
   & \proc{P'}{\Phi'}\sintf{\;\;\tr\;\;}{o^*} \proc{P''}{\Phi''}}
 \\ [2mm]
\mbox{\sc Tau} &
 \infer[\mbox{with }\ell\in\{o^*;i^+;i^*\}]
   {\proc{P}{\Phi} \sintf{\;\;\tr\;\;}{\ell} \proc{P''}{\Phi''}}
   {\proc{P}{\Phi} \sint{\;\;\tau\;\;} \proc{P'}{\Phi'} &
    \proc{P'}{\Phi'} \sintf{\;\;\tr\;\;}{\ell} \proc{P''}{\Phi''}}
 \\ [4mm]
\mbox{\sc Proper} &
   \infer{{\proc{0}{\Phi}\sintf{\;\;\epsilon\;\;}{o^*}
   \proc{0}{\Phi}}}{}
   \quad
   \infer{
     {\proc{\In(c,x).P}{\Phi}\sintf{\;\;\epsilon\;\;}{o^*}
     \proc{\In(c,x).P}{\Phi}}
   }{}
 \\ [2mm]
& \quad
   \infer{
     {\proc{\Out(c,u).P}{\Phi}\sintf{\;\;\epsilon\;\;}{o^*}
     \proc{\Out(c,u).P}{\Phi}}
   }{\neg \valid(u)} \\[4mm]
 \mbox{\sc Improper} &
 \infer{
   {\proc{0}{\Phi}\sintf{\;\;\epsilon\;\;}{i^*} \proc{\bot}{\Phi}}
 }{}
\quad
\infer{
{\proc{\Out(c,u).P}{\Phi}\sintf{\;\;\epsilon\;\;}{i^*} \proc{\bot}{\Phi}}
 }{\neg \valid(u)}
\end{array}
\]
  \caption{Focused semantics on extended basic processes}
  \label{fig:sintf}
\end{figure}

The main idea of the compressed semantics is to ensure
that when a basic process starts executing some actions, it actually
executes a maximal block of actions. In analogy with focusing in sequent
calculus, we say that the basic process takes the focus, and can only
release it under particular conditions.
We define in Figure~\ref{fig:sintf} how blocks can be executed by
extended basic processes. In that semantics,
the label~$\ell$ denotes the stage of the execution, starting with $i^+$, then
$i^*$ after the first input and $o^*$ after the first output.

\begin{example}
\label{ex:compressed-semantics}
Going back to Example~\ref{ex:trace-equiv}, we have that:
\[
\proc{Q_0(\skb,\pk{\ska})}{\Phi_0}
\sintf{\;\;
  \In(\cb,\aenc{\pair{w_1}{w_1}}{w_2})\cdot\Out(\cb,w_3) \;\;}{i^+}
\proc{0}{\Phi}
\]
where $\Phi$ is as given in Example~\ref{ex:trace-equiv}.
As illustrated by the proof tree below, we also have $\proc{Q_0(\skb,\pk{\ska'})}{\Phi_0}
\sintf{\;\tr\;}{i^+} \proc{\bot}{\Phi_0}$ with $\tr =
\In(\cb,\aenc{\pair{w_1}{w_1}}{w_2})$. 
\[
\begin{prooftree}
\proc{Q_0(\skb,\pk{\stef{\ska'}})}{\Phi_0} \sint{\tr}
\proc{Q'}{\Phi_0}
\begin{prooftree}
\proc{Q'}{\Phi_0} \sint{\tau} \proc{0}{\Phi_0} 
\;\;
\begin{prooftree}
\justifies
\proc{0}{\Phi_0} \sintf{\epsilon}{i^*} \proc{\bot}{\Phi_0}
\using{\mbox{\sc {Improper}} \hspace{-0.7cm}}
\end{prooftree}
\justifies
\proc{Q'}{\Phi_0} \sintf{\epsilon}{i^*} \proc{\bot}{\Phi_0}
\using{\mbox{\sc {Tau}} \hspace{-0.4cm}} 
\end{prooftree}
\justifies
\proc{Q_0(\skb,\pk{\stef{\ska'}})}{\Phi_0}
\sintf{\tr}{i^+} \proc{\bot}{\Phi_0}
\using{\mbox{\sc In}}
\end{prooftree}
\]
where $Q' \eqdef \testt{\pk{\ska} = \pk{\ska'}}{\Out(c_B,u)}$ for some
message $u$. 
\end{example}

Then we define the relation $\sintc{}$
between extended simple processes
{as the least reflexive transitive relation satisfying the
rules given in Figure~\ref{fig:sintc}}.

\begin{figure}[h]
\[
\begin{array}{llcll}
 \mbox{\sc Block} &
  \infer
  {\proc{\{Q\}\uplus \p}{\Phi}\sintc{\;\;\tr\;\;}
    \proc{\{Q'\}\uplus\p}{\Phi'}}
  {\proc{Q}{\Phi}\sintf{\;\;\tr\;\;}{i^+} \proc{Q'}{\Phi'}
    & Q'\neq\bot} 
&&
\mbox{\sc Failure} &
  \infer
  {\proc{\{Q\}\uplus\p}{\Phi}\sintc{\;\;\tr\;\;} \proc{\varnothing}{\Phi'}}
  {\proc{Q}{\Phi}\sintf{\;\;\tr\;\;}{i^+} \proc{Q'}{\Phi'} & Q' = \bot}
  \end{array}
\]

 \caption{Compressed semantics on extended simple processes}
  \label{fig:sintc}
\end{figure}

A basic process is allowed to \emph{properly} end a block execution when it 
has performed outputs and it cannot perform any more output
or unobservable action ($\tau$).
Accordingly, we call \emph{proper block} a non-empty 
sequence of inputs followed by a non-empty sequence of outputs, all on the 
same channel.
For completeness, we also allow blocks to be terminated \emph{improperly},
when the process that is executing has performed inputs but no output,
and has reached the null process~$0$ or an output which is blocked.
Accordingly, we call \emph{improper block} a non-empty
sequence of inputs on the same channel.

\begin{example}
\label{ex:witness}
Continuing Example~\ref{ex:compressed-semantics}, 
using the rule {\sc block}, we can derive
\[
\proc{\{P_\init, Q_0(\skb,\pk{\ska})\}}{\Phi_0}
\sintc{\;\;
  \In(\cb,\aenc{\pair{w_1}{w_1}}{w_2})\cdot\Out(\cb,w_3) \;\;}
\proc{P_\init}{\Phi}
\]
\stef{where $P_\init$ is defined in Example~\ref{ex:initial}.}
We can also derive 
\[
\proc{\{P_\init, Q_0(\skb,\pk{\ska'})\}}{\Phi_0}
\sintc{\;\;
  \In(\cb,\aenc{\pair{w_1}{w_1}}{w_2})\;\;}
\proc{\varnothing}{\Phi_0}
\]
 using the rule {\sc Failure}.
Note that the resulting simple process is reduced to~$\varnothing$ even though
$P_\init$ has never been executed.
\label{ex:running-compre}
\end{example}

At first sight, killing the whole process when applying the rule {\sc Failure} may seem too strong. 
However, even if this kind of scenario is observable by the attacker, it does not bring
him any new knowledge, hence it plays only a limited role
in trace  equivalence: 
it is in fact sufficient to consider such improper blocks
only at the end of traces.

 \begin{example}
 \label{ex:dirty-block}
 Consider $\p = \{\; \In(c,x).\In(c,y),\; \In(c',x') \;\}$.
 Its compressed traces are of the form $\In(c,M).\In(c,N)$ and $\In(c',M')$.
 The concatenation of those two improper traces cannot be executed in 
 the compressed semantics. Intuitively, we do not loose anything for
 trace equivalence, because if a process can exhibit those two improper
 blocks they must be in parallel and hence considering their combination is 
 redundant.
 \end{example}

We now define the notions of \emph{compressed trace equivalence} (denoted
$\eintc$) and \emph{compressed trace inclusion} (denoted $\inceintc$),
similarly to $\eint$ and $\sqsubseteq$ but
relying on $\sintc{}$ instead of $\sintw{}$.

\begin{definition}[compressed trace equivalence]
\label{def:compressed-equivalence}
Let~$A$ and~$B$ be two extended simple processes. We have that $A \inceintc B$ if,
for every
sequence of actions $\tr$ such that $A \sintc{\;\tr} \proc{\p}{\Phi}$, there
exists $\proc{\p'}{\Phi'}$ such that $B \sintc{\;\tr}
\proc{\p'}{\Phi'}$ and $\Phi \statequiv \Phi'$. 
The processes~$A$ and~$B$ are 
\emph{compressed trace equivalent}, denoted by $A \eintc B$, if $A 
\inceintc B$
and $B \inceintc A$.
\end{definition}

\begin{example}
We have that $\proc{\{P_\init,Q_0(\skb,\pk{\ska})\}}{\Phi_0} \not\eintc
\proc{\{P_\init,Q_0(\skb,\pk{\ska'})\}}{\Phi_0}$.
The trace $\In(\cb,\aenc{\pair{w_1}{w_1}}{w_2})\cdot\Out(\cb,w_3)$ exhibited in
Example~\ref{ex:witness} is executable from
$\proc{\{P_\init,Q_0(\skb,\pk{\ska})\}}{\Phi_0}$. However, this trace has no
counterpart when starting with
$\proc{\{P_\init,Q_0(\skb,\pk{\ska'})\}}{\Phi_0}$.
\end{example}

\subsection{Soundness and completeness}
\label{subsec:soundness-completeness}

 We shall now establish soundness and
 completeness of the compressed semantics. More precisely, we show that
 the two relations~$\eint$ and~$\eintc$ coincide on initial simple
 processes (Theorem~\ref{theo:comp-soundness-completeness}). \stef{All the
 proofs of this section are given in Appendix~\ref{app:compression}.}

Intuitively, we can always permute output  (resp.~input) actions occurring
on distinct channels, and we can also permute an output
with an input if the outputted message is not used to build the inputted
term.
More formally, 
we define an \emph{independence relation~$\Ind_a$ over actions} as
the least symmetric relation satisfying:
\begin{itemize}
\item $\Out(c_i,w_i) \Ind_a \Out(c_j,w_j)$ and $\In(c_i,M_i) \Ind_a
  \In(c_j,M_j)$ as soon as $c_i \neq c_j$, 
\item $\Out(c_i,w_i) \Ind_a \In(c_j,M_j)$ when in addition $w_i
  \not\in \fv(M_j)$.
\end{itemize}
Then, we consider ${=}_{\Ind_a}$ to be the least
congruence (w.r.t. concatenation) satisfying:
\begin{center}
 $\act\cdot\act' =_{\Ind_a} \act'\cdot \act$
for all $\act$ and $\act'$ with $\act \Ind_a \act'$, 
\end{center}
and we show that
processes are equally able to execute equivalent (w.r.t. $=_{\Ind_a}$) traces.

\begin{restatable}{lemma}{lempermute}
\label{lem:permute}
Let $A$, $A'$ be two extended simple processes and $\tr$, $\tr'$ be such that
$\tr =_{\Ind_a} \tr'$. We have that $A \sintw{\;\tr} A'$ if, and only if, $A
\sintw{\;\tr'} A'$. 
\end{restatable}

Now, considering traces that are only made of proper blocks, a
strong relationship can be established between the two semantics.

\begin{restatable}{proposition}{procompsoundness}
\label{pro:comp-soundness}
Let $A$, $A'$ be two simple extended processes, and $\tr$ be a trace
made of proper blocks  such that $A \sintc{\;\tr\;} A'$. Then we have
that  $A \sintw{\;\tr} A'$.
\end{restatable}

Actually, the result stated in Proposition~\ref{pro:comp-soundness} 
immediately follows from
  the observation that $\sintc{}$ is included in $\sintw{}$
  for traces made of proper blocks since for them
  {\sc Failure} cannot arise.

\begin{restatable}{proposition}{procompcompleteness}
\label{pro:comp-completeness}
Let $A$, $A'$ be two initial simple processes, and $\tr$ be a trace
made of proper blocks such that
 $A \sintw{\;\tr} A'$.
Then, we have that $A \sintc{\;\tr\;} A'$.
\end{restatable}

This result is more involved and relies on the additional hypothesis
that $A$ and $A'$ have to be initial to ensure that no {\sc Failure} will
arise.

\begin{restatable}{theorem}{theocompsoundnesscompleteness}
\label{theo:comp-soundness-completeness}
Let $A$ and $B$ be two initial simple processes. We have that
\[
A \eint B \;\; \Longleftrightarrow\;\; A \eintc B.
\]
\end{restatable}

\begin{proof}[Proof sketch, details in Appendix~\ref{app:compression}]
The main difficulty 
is that Proposition~\ref{pro:comp-completeness} only considers traces
composed of proper blocks whereas we have to consider all traces.
To prove the $\Rightarrow$ implication, we
have to pay attention to the last block of the compressed trace that
can be an improper one (composed of several inputs on a channel $c$). 
The $\Leftarrow$ implication is more difficult since we have to
consider a trace $\tr$ of a process~$A$
that is an interleaving of some prefix of proper and
improper blocks. We will first complete it with $\tr^+$ to obtain an interleaving of
proper and improper blocks. We then reorder the 
actions to obtain a trace $\tr'$ such that $\tr \cdot \tr^+=_{\Ind_a} \tr'$ and
$\tr' = \tr_{\mathrm{io}} \cdot \tr_{\mathrm{in}}$ where $\tr_{\mathrm{io}}$
is made of proper blocks while $\tr_{\mathrm{in}}$ is made of improper blocks.
For each improper block~$b$ of $\tr_{\mathrm{in}}$, we show by applying
Lemma~\ref{lem:permute} and Proposition~\ref{pro:comp-completeness}
that $A$ is able to perform $\tr_{\mathrm{io}}$ in the compressed
semantics and the resulting extended process can execute the improper block $b$.
We thus have that $A$ is able to perform 
$\tr_{\mathrm{io}} \cdot b$ in the compressed semantics and thus $B$ as well.
Finally, we show that the executions of all those (concurrent) blocks~$b$
can be put together, obtaining that~$B$ can perform~$\tr'$, and thus
$\tr$ as well.
\end{proof}

Note that, as illustrated by the following example, the two underlying
notions of trace 
inclusion do \emph{not} coincide.

\begin{example}
Let $\p = \{\In(c,x)\}$ and $\q = \{\In(c,x).\Out(c,n)\}$ accompanied
with an arbitrary frame $\Phi$. 
We have
${\proc{\p}{\Phi} \inceint \proc{\q}{\Phi}}$ but
$\proc{\p}{\Phi} \not\inceintc \proc{\q}{\Phi}$
since in the compressed semantics~$\proc{\q}{\Phi}$ is not
allowed to stop its execution after its first input. 
\end{example}

%% file: symbolic.tex

\section{Deciding trace equivalence via constraint solving}
\label{sec:constraint-solving}

In this section, we propose a symbolic semantics for our compressed
semantics following, \eg \cite{MS02,baudet-ccs2005}. 
Such a semantics avoids
potentially infinite branching of our compressed semantics
due to inputs from the environment. 
Correctness is maintained by associating with each process a set of constraints on terms.

\subsection{Constraint systems}
\label{subsec:constraint-systems}

Following the notations of~\cite{baudet-ccs2005}, we consider a new set $\X^2$ of 
\emph{second-order variables}, denoted by $X$, $Y$, etc.
 We shall use those variables to abstract over recipes.
We denote by $\fvs(o)$ the set of free second-order variables of an object
$o$, typically a constraint system. To prevent ambiguities, we shall
use $\fvp$ instead of $\fv$ for free first-order variables.

\begin{definition}[constraint system]\label{def:cs}
 A \emph{constraint system} $\C = \cs{\Phi}{\Set}$ consists of a frame
 $\Phi$, and a set of
  constraints $\Set$.
 We consider three kinds of \emph{constraints}:
\[
\dedi{D}{X}{x} \quad\quad u \eqi v \quad\quad u \neqi v
\]

\noindent  where $D\subseteq\W$, $X\in\X^2$, $x\in\X$ and
$u,v\in\T(\Sigma, \N\cup\X)$.
\end{definition}

The first kind of constraint expresses that a second-order variable~$X$
has to be instantiated by a recipe that uses only 
variables from a certain set~$D$, and
that the obtained term should be $x$. The handles in~$D$
 represent terms that have been previously outputted by the process.

\medskip{}

We are not interested in general constraint systems, but only
consider constraint systems that are \emph{well-formed}.
Given a constraint system~$\C$, we define a dependency order on
first-order variables in ${\fvp(\C) \cap \X}$
  by declaring that $x$ depends on $y$ if, and only if, $\Set$ contains a
  deduction constraint $\dedi{D}{X}{x}$ with $y\in\fvp(\Phi(D))$.
  A constraint system $\C$ is \emph{well-formed} if:
\begin{itemize} 
\item the
  dependency relationship is acyclic, and
\item  for every
  $x\in\fvp(\C) \cap \X$ (resp.~$X\in\fvs(\C)$) there is a unique constraint
  $\dedi{D}{X}{x}$ in~$\Set$.
\end{itemize}
  For $X\in\fvs(\C)$, we write $D_\C(X)$ for the domain $D\subseteq\W$
  of the deduction constraint $\dedi{D}{X}{x}$ associated to~$X$
  in~$\C$.

\begin{example}
\label{ex:cs}
\stef{Continuing Example~\ref{ex:private}}, let $\Phi = \Phi_0 \uplus \{w_3 \refer
\aenc{\pair{\projr{N}}{\pair{n_b}{\pk{skb}}}}{\pk{ska}}\}$ with $N =
\adec{y}{skb}$, and $\Set$ be a set containing two constraints:
\[
\dedi{\{w_0,w_1,w_2\}}{Y}{y} \mbox{ and } \projr{N} \eqi
\pk{ska}.
\]

\noindent We have that $\C = \cs{\Phi}{\Set}$ is a well-formed constraint system.
There is only one first-order variable $y \in \fvp(\C) \cap \X$, and it does
not occur in $\fvp(\Phi(\{w_0,w_1,w_2\}))$, which is empty. Moreover,
there is indeed a unique constraint that introduces~$y$.
\end{example}

Our notion of well-formed constraint systems is in line with what is used,
\eg in~\cite{MS02,baudet-ccs2005}.
We use a simpler variant here that is sufficient for our purpose. 

\begin{definition}[solution] \label{def:cssol}
A \emph{solution} of a constraint system $\C = \cs{\Phi}{\Set}$ is a
substitution~$\theta$ such that
$\dom(\theta) = \fvs(\C)$, and
$X\theta \in \T(\Sigma, D_\C(X))$ for any $X \in \dom(\theta)$.
Moreover, we require that there exists a ground substitution $\lambda$
with $\dom(\lambda) = \fvp(\C)$ such that:
\begin{itemize}
  \item for every $\dedi{D}{X}{x}$ in $\Set$, we have
  $(X\theta)(\Phi\lambda) =_\E x\lambda$,
 $\valid((X\theta)(\Phi\lambda))$, and $\valid(x\lambda)$;
\item for every $u \eqi v$ in $\Set$, we have $u\lambda =_\E
  v\lambda$, $\valid(u\lambda)$, and $\valid(v\lambda)$; and
\item for every $u \neqi v$ in $\Set$, we have $u\lambda \neq_\E
  v\lambda$,  or $\neg \valid(u\lambda)$, or $\neg \valid(v\lambda)$.
\end{itemize}
Moreover, we require that all the terms occurring in $\Phi\lambda$ are
valid.
The set of solutions of a constraint system $\C$ is denoted $\Sol(\C$).
Since we consider constraint systems that are well-formed,
the substitution $\lambda$ is unique modulo $\E$ given 
$\theta\in\Sol(\C)$. We denote it by $\lambda_\theta$ when $\C$ is
clear from the context.
\end{definition}

\lucca{Note that the validity constraints in the notion of solution of symbolic processes
reflect the validity constraints of the concrete semantics (\ie outputted and inputted terms must be valid and the equality between
terms requires the two terms to be valid).}
Since we consider well-formed constraint systems, we may note that the
substitution $\lambda$ above is not obtained through unification. This
substitution is entirely determined
(modulo $\E$) from~$\theta$ by considering the deducibility
constraints only.

\begin{example}
\label{ex:solution}
Consider again the constraint system $\C$ given in
Example~\ref{ex:cs}. We have that $\theta = \{Y \mapsto
\aenc{\pair{w_1}{w_1}}{w_2}\}$ is a solution of $\C$. Its
associated first-order solution is $\lambda_\theta = \{y \mapsto
\aenc{\pair{\pk{ska}}{\pk{ska}}}{\pk{skb}}\}$. 
\end{example}


\subsection{Symbolic processes: syntax and semantics}
\label{subsec:symbolic-calculus}

Given an extended simple process $\proc{\p}{\Phi}$, 
we compute the constraint systems capturing its possible executions,
starting from the symbolic process $\triple{\p}{\Phi}{\varnothing}$.
Note that we are now manipulating processes that are not ground
anymore, but may contain free variables.

\begin{definition}[symbolic process] \label{def:symbolic-process}
A symbolic process is a tuple $\triple{\p}{\Phi}{\Set}$ where
$\cs{\Phi}{\Set}$ is a constraint system and
$\fvp(\p)\subseteq(\fvp(\Set) \cap \X)$.
\end{definition}

\begin{figure}[h]

\[\begin{array}{lcl}
\mbox{\sc In} &\;\;\;\;&
\triple{\In(c,y).P}{\Phi}{\Set} \ssym{\In(c,X)}
\triple{P\{y\mapsto x\}}{\Phi}{\Set \cup \{ \dedi{\dom(\Phi)}{X}{x} \}}\\
&& \hspace{1cm}\hfill\mbox{where $X$ (resp. $x$) is a  fresh second-order
  (resp. first-order) variable} 
\\[1mm]
\mbox{\sc Out} &&
\triple{\Out(c,u).P}{\Phi}{\Set} \ssym{\Out(c,w)}
  \triple{P}{\Phi\cup\{w\refer u\}}{\Set}\\
&&\hfill\mbox{where $w$ is a fresh first-order variable}
\\[1mm]
\mbox{\sc Then} &&
\triple{\test{u=v}{P}{Q}}{\Phi}{\Set} \ssym{\;\;\tau\;\;}
\triple{P}{\Phi}{\Set \cup \{ u \eqi v \}}
\\[1mm]
\mbox{\sc Else} &&
\triple{\test{u=v}{P}{Q}}{\Phi}{\Set} \ssym{\;\;\tau\;\;}
\triple{Q}{\Phi}{\Set \cup \{ u \neqi v \}}
\end{array}
\]

 \caption{Symbolic semantics for symbolic basic processes}
\label{figure:symbolic-sem}
\end{figure}

We give in Figure~\ref{figure:symbolic-sem} 
a standard symbolic semantics for 
 symbolic basic processes.
From this semantics given on symbolic basic processes only,
  we derive a semantics on simple symbolic processes in a natural
  way:
\[
 \infer[]
  {\triple{\{P\} \uplus \p}{\Phi}{\Set}\ssym{\alpha} \triple{\{P'\} \uplus\p}{\Phi'}{\Set'}}
  {\triple{P}{\Phi}{\Set}\ssym{\alpha} \triple{P'}{\Phi'}{\Set'}}
\]

We can also derive 
our compressed symbolic 
semantics $\ssymc{\;\tr\;}$ following 
the same pattern as for the concrete semantics (see Figure~\ref{figure:compressed-symbolic-sem}).
We consider interleavings that execute maximal blocks of
actions, and we allow improper termination of a block only at
the end of a trace.
\stef{Note that the $\lnot valid(u)$ conditions of the third {\sc Proper} rule and the
  second {\sc Improper} rule are replaced by $u\neq^? u$ constraints in their 
symbolic counterparts.}

\begin{figure}[h]
\[
  \begin{array}{lc}
\mbox{\sc In} &  \infer[\mbox{with }\ell\in\{i^*;i^+\}]
  {\triple{P}{\Phi}{\Set}\ssymf{\In(c,X). \tr\;}{\ell} \triple{P''}{\Phi''}{\Set''}}
  {\triple{P}{\Phi}{\Set}\ssym{\In(c,X)} \triple{P'}{\Phi'}{\Set'}
  & \triple{P'}{\Phi'}{\Set'}\ssymf{\;\tr\;\;}{i^*} \triple{P''}{\Phi''}{\Set''}}
   \\ [4mm]
\mbox{\sc Out} & 
 \infer[\mbox{with } \ell\in\{i^*;o^*\}]
   {\triple{P}{\Phi}{\Set}\ssymf{\Out(c,w). \tr\;}{\ell} 
   \triple{P''}{\Phi''}{\Set''}}
   {\triple{P}{\Phi}{\Set}\ssym{\Out(c,w)} \triple{P'}{\Phi'}{\Set'}
   & \triple{P'}{\Phi'}{\Set'}\ssymf{\;\;\tr\;\;}{o^*} \triple{P''}{\Phi''}{\Set''}}
 \\ [4mm]
\mbox{\sc Tau} &
 \infer[\mbox{with }\ell\in\{o^*;i^+;i^*\}]
 {\triple{P}{\Phi}{\Set} \ssymf{\;\;\tr\;\;}{\ell} \triple{P''}{\Phi''}{\Set''}}
 {\triple{P}{\Phi}{\Set} \ssym{\tau} \triple{P'}{\Phi'}{\Set'} &
    \triple{P'}{\Phi'}{\Set'} \ssymf{\;\;\tr\;\;}{\ell} \triple{P''}{\Phi''}{\Set''}}
 \\ [4mm]
\mbox{\sc Proper} &
\infer{
   {\triple{0}{\Phi}{\Set}\ssymf{\;\;\epsilon\;\;}{o^*}
     \triple{0}{\Phi}{\Set}}
   }{}
\quad
\infer{
   {\triple{\In(c,x).P}{\Phi}{\Set}\ssymf{\;\;\epsilon\;\;}{o^*}
     \triple{\In(c,x).P}{\Phi}{\Set}}
   }{}\\[4mm]
& \quad
   \infer{
     {\triple{\Out(c,u).P}{\Phi}{\Set}\ssymf{\;\;\epsilon\;\;}{o^*}
     \triple{\Out(c,u).P}{\Phi}{\Set \cup \{u\neqi u\}}}
   }{} 
\\ [7mm]
\mbox{\sc Improper} &
\infer{
   {\triple{0}{\Phi}{\Set}\ssymf{\;\;\epsilon\;\;}{i^*} \triple{\bot}{\Phi}{\Set}}
 }{}
\quad
\infer{
{\triple{\Out(c,u).P}{\Phi}{\Set}\ssymf{\;\;\epsilon\;\;}{i^*}
                      \triple{\bot}{\Phi}{\Set \cup \{u \neqi u\}}}}
{}
\end{array}
\]
\bigskip{}
\[
\begin{array}{lcl}
 \mbox{\sc Block} &&\mbox{\sc Failure}\\
  \infer
  {\triple{\{Q\}\uplus \p}{\Phi}{\Set}\ssymc{\;\;\tr\;\;}
    \triple{\{Q'\}\uplus\p}{\Phi'}{\Set'}}
  {\triple{Q}{\Phi}{\Set}\ssymf{\;\;\tr\;\;}{i^+} \triple{Q'}{\Phi'}{\Set'}
    & Q'\neq\bot} 
&\;\;\;\;\;&
  \infer
  {\triple{\{Q\}\uplus\p}{\Phi}{\Set}\ssymc{\;\;\tr\;\;} 
  \triple{\varnothing}{\Phi'}{\Set'}}
  {\triple{Q}{\Phi}{\Set}\ssymf{\;\;\tr\;\;}{i^+} \triple{Q'}{\Phi'}{\Set'} & Q' = \bot}
  \end{array}
\]
 \caption{Compressed symbolic semantics}
  \label{figure:compressed-symbolic-sem}
\end{figure}

\begin{example}
\label{ex:symbolic-semantics}
\stef{Continuing Example~\ref{ex:private},}
we have that 
$\triple{\{Q_0(\skb,\pk{\ska})\}}{\Phi_0}{\varnothing} \ssymc{\;\tr\;}
\triple{\varnothing}{\Phi}{\Set}$ where:
\begin{itemize}
\item $\tr = \In(\cb,
Y
)\cdot\Out(\cb,w_3)$, and
\item $\C = \cs{\Phi}{\Set}$ is the constraint system defined in
  Example~\ref{ex:cs}.
\end{itemize}
\end{example}

We are now able to define the notion of equivalence associated to
these two semantics, namely symbolic trace equivalence (denoted $\esym{}$) and
symbolic compressed trace equivalence (denoted $\esymc{}$).
For a trace $\tr$, we note $\obse(\tr)$ the trace obtained from $\tr$ by
removing all~$\tau$ actions.

\begin{definition}
\label{def:equiv-symb}
Let~$A = \proc{\p}{\Phi}$
and~$B=\proc{\q}{\Psi}$ be two simple processes. We have that $A \incesym B$
when, for every trace $\tr$ such
  that $\triple{\p}{\Phi}{\varnothing} \ssym{\;\tr}
  \triple{\p'}{\Phi'}{\Set_A}$, for every $\theta \in
  \Sol(\Phi';\Set_A)$, we have that:
\begin{itemize}
\item $\triple{\q}{\Psi}{\varnothing} \ssym{\;\tr'}
  \triple{\q'}{\Psi'}{\Set_B}$ where $\obse(\tr')=\obse(\tr)$
  with $\theta \in \Sol(\Psi';\Set_B)$, and 
\item $\Phi'\lambda^A_\theta \statequiv
  \Psi'\lambda^B_\theta$ where $\lambda^A_\theta$ (resp. $\lambda^B_\theta$) is the substitution
  associated to $\theta$ w.r.t. $\cs{\Phi'}{\Set_A}$
  (resp. $\cs{\Psi'}{\Set_B}$).
\end{itemize}
We have that $A$ and $B$ are in \emph{trace equivalence w.r.t. $\ssym{}$}, 
denoted $A \esym B$, if $A \incesym B$ and $B\incesym A$.
\end{definition}

We derive similarly the notion of trace equivalence induced by
$\ssymc{}$. We do not have to take care of the $\tau$ actions since
they are performed implicitly in
the compressed semantics.

\begin{definition}
\label{def:equiv-comp-symb}
Let~$A = \proc{\p}{\Phi}$
and~$B=\proc{\q}{\Psi}$ be two extended simple processes. We have that $A \incesymc B$
when, for every trace $\tr$ such
  that $\triple{\p}{\Phi}{\varnothing} \ssymc{\;\tr}
  \triple{\p'}{\Phi'}{\Set_A}$, for every $\theta \in
  \Sol(\Phi';\Set_A)$, we have that:
\begin{itemize}
\item $\triple{\q}{\Psi}{\varnothing} \ssymc{\;\tr}
  \triple{\q'}{\Psi'}{\Set_B}$ with $\theta \in
  \Sol(\Psi';\Set_B)$, and 
\item $\Phi'\lambda^A_\theta \statequiv
  \Psi'\lambda^B_\theta$ where $\lambda^A_\theta$ (resp. $\lambda^B_\theta$) is the substitution
  associated to $\theta$ w.r.t. $\cs{\Phi'}{\Set_A}$
  (resp. $\cs{\Psi'}{\Set_B}$).
\end{itemize}
We have that $A$ and $B$ are in \emph{trace equivalence
  w.r.t. $\ssymc{}$}, denoted $A \esymc B$, if $A \incesymc B$ and $B
\incesymc A$.
\end{definition}

\begin{example}
We have that $\proc{\{Q_0(\skb,\pk{\ska})\}}{\Phi_0} \not\incesymc
\proc{\{Q_0(\skb,\pk{\ska'})\}}{\Phi_0}$. Continuing
Example~\ref{ex:symbolic-semantics}, we have seen that:
\begin{itemize}
\item  $\triple{\{Q_0(\skb,\pk{\ska})\}}{\Phi_0}{\varnothing} \ssymc{\;\tr\;}
\triple{\varnothing}{\Phi}{\Set}$ \lucca{(see Example~\ref{ex:symbolic-semantics})}, and 
\item \stef{$\theta = \{Y \mapsto \aenc{\pair{w_1}{w_1}}{w_2}\}$ is a
  solution of $\C = (\Phi;\Set)$}
\stef{(see Example~\ref{ex:solution})}. 
\end{itemize}
The only symbolic process that
is reachable from $\triple{\{Q_0(\skb,\pk{\ska'})\}}{\Phi_0}{\varnothing}$ 
using $\tr$ is $\triple{\varnothing}{\Phi'}{\Set'}$ with:
\begin{itemize}
\item $\Phi' = \Phi_0 \uplus\{w_3 \refer \aenc{\pair{\projr{N}}{\pair{n_b}{\pk{skb}}}}{\pk{ska'}}\}$, and
\item $\Set' = \big\{\dedi{\{w_0,w_1,w_2\}}{Y}{y};\;\; \projr{N} \eqi
  \pk{ska'}\big\}$.
\end{itemize} 
One can check that $\theta$ is not a solution of $\cs{\Phi'}{\Set'}$.
\end{example}

For processes without replication, the symbolic transition system
induced by $\ssym{}$ (resp~$\ssymc{}$) is essentially finite.
Indeed, the choice of fresh names for handles and second-order variables
does not matter, and therefore the relations
$\ssym{}$ and $\ssymc{}$ are essentially finitely branching.
Moreover, the length of traces of a simple process is obviously bounded.
Thus, deciding (symbolic) trace equivalence between
processes boils down to
the problem of deciding a notion of equivalence between sets of
constraint systems. This problem is well-studied and several
procedures already exist~\cite{baudet-ccs2005,chevalier10}, \eg~\apte~\cite{cheval-ccs2011} (see Section~\ref{sec:apte}).


\subsection{Soundness and completeness}
\label{subsec:symbolic-soundness-completeness}

It is well-known that the symbolic semantics $\ssym{}$ is sound and
complete w.r.t. $\sint{}$, and therefore that the two underlying
notions of equivalence, namely $\esym$ and $\eint$, coincide.
This has been proved for instance in~\cite{baudet-ccs2005,CCD-tcs13}.
Using the same approach, we can
show soundness and completeness of our symbolic compressed semantics
w.r.t. our concrete compressed semantics. We have:
\begin{itemize}
\item \emph{Soundness}: each transition in the compressed symbolic semantics represents a set
of transitions that can be done in the concrete compressed semantics.
\item \emph{Completeness}: each transition in the compressed semantics can be
matched by a transition in the compressed symbolic semantics.
\end{itemize}

These results are formally expressed in Proposition~\ref{pro:symb-soundness} and
Proposition~\ref{pro:symb-completeness} below. These propositions are
simple consequences of similar propositions that link the (small-step) symbolic
semantics and the (small-step) standard semantics. Lifting these
results to the compressed semantics is straightforward since both
semantics are built using exactly the same scheme (see Figures~\ref{fig:sintf} and~\ref{figure:compressed-symbolic-sem}).

\begin{proposition}
\label{pro:symb-soundness}
Let $\proc{\p}{\Phi}$ be an extended simple process such that
$\triple{\p}{\Phi}{\varnothing} \ssymc{\tr} \triple{\p'}{\Phi'}{\Set'}$,
and $\theta \in \Sol(\Phi';\Set')$. We have that
$\proc{\p}{\Phi} \sintc{\tr\theta}
\proc{\p'\lambda}{\Phi'\lambda}$ where $\lambda$ is the first-order
solution of $\triple{\p'}{\Phi'}{\Set'}$ associated to $\theta$.
\end{proposition}



\begin{proposition}
\label{pro:symb-completeness}
Let $\proc{\p}{\Phi}$ be an extended simple process such that
$\proc{\p}{\Phi} \sintc{\tr} \proc{\p'}{\Phi'}$. There exists a
symbolic process $\triple{\p_s}{\Phi_s}{\Set}$, a solution
$\theta \in \Sol(\Phi_s;\Set)$, and a sequence $\tr_s$ such
that:
\begin{itemize}
\item $\triple{\p}{\Phi}{\varnothing} \ssymc{\tr_s}
  \triple{\p_s}{\Phi_s}{\Set}$;
\item $\proc{\p'}{\Phi'} = \proc{\p_s\lambda}{\Phi_s\lambda}$;
  and
\item $\tr = \tr_s\theta$
\end{itemize} 
where $\lambda$ is the first-order solution of
$\triple{\p_s}{\Phi_s}{\Set}$ associated to $\theta$.
\end{proposition}


Finally, relying on these two results, we can establish  
that symbolic trace equivalence ($\esymc$) exactly captures
compressed trace equivalence ($\eintc$). 
Actually, both inclusions can be established separately.

\begin{restatable}{theorem}{thrmsintcssymv}
 \label{thrm:sintc-ssymc}
For any extended simple processes $A$ and $B$, we have that:
\[
A \inceintc B
\iff A \incesymc B.
\]
\end{restatable}

As an immediate consequence of Theorem~\ref{theo:comp-soundness-completeness}
and Theorem~\ref{thrm:sintc-ssymc}, we obtain that the relations~$\eint$ and~$\esymc$ coincide.

\begin{corollary}
\label{cor:sintc-ssymc}
For any initial simple processes $A$ and $B$, we have that:
\[
A \eint B
\iff A \esymc B.
\]
\end{corollary}

%% file: diff-stef.tex
\section{Reduction using dependency constraints}
\label{sec:diff}

Unlike compression, 
which is essentially based on the input/output nature of actions, our second
optimisation takes
into account the exchanged messages.
Let us first illustrate one simple instance
of our optimisation and how dependency constraints~\cite{ModersheimVB10}
may be used to incorporate it into symbolic semantics.

\begin{example}
\label{ex:diamond}
Let $P_i = \In(c_i,x_i).\Out(c_i,u_i).P'_i$ with $i \in \{1,2\}$,
and $\Phi_0 = \{w_0 \refer n\}$ be a ground frame.
We consider the simple process $A = \proc{\{P_1,P_2\}}{\Phi_0}$, and
the  two symbolic interleavings depicted in Figure~\ref{figure:ex}.
The two resulting symbolic processes are of the form
$\triple{\{P'_1,P'_2\}}{\Phi}{\Set_i}$ where
$\Phi = \Phi_0 \uplus \{w_1 \refer u_1, w_2 \refer u_2\}$, 
$$\Set_1 = \big\{\dedi{w_0}{X_1}{x_1}; \; \dedi{w_0,
w_1}{X_2}{x_2}\big\}, \mbox{ and }
\Set_2 = \big\{\dedi{w_0}{X_2}{x_2};\; \dedi{w_0,
  w_2}{X_1}{x_1}\big\}.$$
The sets of concrete processes that these two symbolic processes
represent are different, which means that we cannot discard any of
those interleavings. However, these sets have a significant overlap
corresponding to concrete instances of the interleaved blocks that
are actually independent, \ie where the output of one block is not
necessary to obtain the input of the next block.
In order to avoid considering such concrete processes twice,
we may add a \emph{dependency constraint} $\constrd{X_1}{w_2}$ \stef{in $\Set_2$},
whose purpose is to discard all solutions~$\theta$
such that the message $x_1\lambda_\theta$ can be derived without using
$w_2 \refer u_2\lambda_\theta$.
For instance, the concrete trace 
$\In(c_2, w_0)\cdot\Out(c_2,w_2)\cdot \In(c_1,w_0)\cdot \Out(c_1,w_1)$
would be discarded thanks to this new constraint.
\end{example}

\begin{figure}[h]
\begin{tikzpicture}[inner sep=0pt,node distance=0.7cm]


\node (sommet) {$\bullet$};
\node[below of=sommet, left=1.2cm of sommet] (l1) {$\bullet$};
\node[below of=l1] (l2) {$\bullet$};
\node[below of=l2] (l3) {$\bullet$};
\node[below of=l3] (l4) {$\bullet$};
\node[below of=sommet, right=1.2cm of sommet] (r1) {$\bullet$};
\node[below of=r1] (r2) {$\bullet$};
\node[below of=r2] (r3) {$\bullet$};
\node[below of=r3] (r4) {$\bullet$};


\path (sommet) edge  node[left,pos=0.3]{$\In(c_1,X_1)\;\;\;\;$} (l1);
\path (l1) edge node[left]{$\Out(c_1,w_1)\;$} (l2);
\path (l2) edge node[left]{$\In(c_2,X_2)\;$} (l3);
\path (l3) edge node[left]{$\Out(c_2,w_2)\;$} (l4);
\path (sommet) edge  node[right,pos=0.3]{$\;\;\;\;\In(c_2,X_2)$} (r1);
\path (r1) edge node[right]{$\;\Out(c_2,w_2)$} (r2);
\path (r2) edge node[right]{$\;\In(c_1,X_1)$} (r3);
\path (r3) edge node[right]{$\;\Out(c_1,w_1)$} (r4);


\coordinate (l4t) at ($ (l4) - (90:0.1cm) $);
\coordinate (l4l) at ($ (l4t) - (60:0.5cm) $);
\coordinate (l4r) at ($ (l4t) - (120:0.5cm) $);
\filldraw[color=gray!20] (l4t) -- (l4l) -- (l4r) -- (l4t);
\coordinate (r4t) at ($ (r4) - (90:0.1cm) $);
\coordinate (r4l) at ($ (r4t) - (60:0.5cm) $);
\coordinate (r4r) at ($ (r4t) - (120:0.5cm) $);
\filldraw[color=gray!20] (r4t) -- (r4l) -- (r4r) -- (r4t);


\node[right=1.0cm of l4] (cl) {} ;
\node[left=1.0cm of r4] (cr) {} ;
\filldraw[pattern=north east lines] (cl) circle (0.5);
\filldraw[pattern=north west lines] (cr) circle (0.5);
\draw[dashed] (l4) -- ($ (cl) + (120:0.5cm) $);
\draw[dashed] (l4) -- ($ (cl) + (-120:0.5cm) $);
\draw[dashed] (r4) -- ($ (cr) - (120:0.5cm) $);
\draw[dashed] (r4) -- ($ (cr) - (-120:0.5cm) $);

\end{tikzpicture}
\caption{Two symbolic compressed traces (Example~\ref{ex:diamond})}
\label{figure:ex}
\end{figure}

The idea of \cite{ModersheimVB10} is to accumulate dependency constraints
generated whenever such a pattern is detected in an execution, and use
an adapted constraint resolution procedure to narrow and eventually
discard the constrained symbolic states.
We seek to exploit similar ideas for optimising the verification of
trace equivalence rather than reachability.
This requires extra care, since pruning traces as
described above may break completeness when considering trace equivalence. As 
before, the key to obtain a valid optimisation will be to discard traces in a 
similar way on the two processes being compared.
In addition to handling this necessary subtlety, we also propose a new proof
technique for justifying dependency constraints. The generality of that 
technique allows us to add more dependency constraints, taking into account
more patterns than the simple one from the previous example.

\subsection{Reduced semantics}
\label{subsec:reduced-semantics}

We start by introducing \emph{dependency constraints}.

\begin{definition} \label{def:dep-constr}
  A \emph{dependency
  constraint} is a constraint of the form $\constrd{\vect X}{\vect w}$
  where $\vect X$ is a vector of second-order variables in $\X^2$, 
  and $\vect w$ is a vector of handles, \ie variables in $\W$.

Given a constraint system $\C = \cs{\Phi}{\Set}$, a set $\Set_D$ of
dependency constraints, and $\theta \in
\Sol(\C)$. We write $\theta \models_{\cs{\Phi}{\Set}} {\Set_D}$ when
$\theta$ also satisfies the dependency constraints in $\Set_D$, \ie \stef{when
for
each  $\constrd{\vect X}{\vect w} \in \Set_D$ 
there is some $X_i \in \vect X$ 
 such that
  for all recipes $M\in\T(\Sigma, D_\C(X_i))$ satisfying
  $M(\Phi\lambda_\theta) {=_{\E}} (X_i\theta)(\Phi\lambda_\theta)$
  and $\valid(M(\Phi\lambda_\theta))$,
  we have that $\fvp(M) \cap \vect w \neq \varnothing$}
where $\lambda_\theta$ is the substitution associated to $\theta$
w.r.t.\ $\cs{\Phi}{\Set}$.
\end{definition}

Intuitively, a dependency constraint  $\constrd{\vect X}{\vect w}$ is
satisfied as soon as at least one message among those in $({\vect
  X}\theta)(\Phi\lambda_\theta)$ can  only be deduced by using a message stored in $\vect w$.

\begin{example}
Continuing Example~\ref{ex:diamond}, assume that $u_1 = u_2 = n$
and let $\theta = \{X_1 \mapsto w_2; X_2 \mapsto w_0\}$. We have that
$\theta \in  \Sol(\C_2)$ and
the substitution associated to $\theta$ w.r.t. $\C_2$ is
$\lambda^2_\theta = \{x_1 \mapsto n; x_2 \mapsto n\}$. However,
$\theta$ does not satisfy the dependency constraint
$\constrd{X_1}{w_2}$. Indeed, we have that $w_0(\Phi\lambda^2_\theta)
=_\E (X_1\theta)(\Phi\lambda^2_\theta)$ whereas $\{w_0\} \cap
\{w_2\} = \emptyset$. Intuitively, this means that there is no good
reason to postpone the execution of the block on channel~$c_1$ if the
output on $c_2$ is not useful to build the message used in input on $c_1$.
\end{example}

We shall now define formally how dependency constraints will be added to
our constraint systems. For this, we fix an arbitrary
total order $\och$ on channels. Intuitively, this order expresses
which executions should be favored, and which should be allowed
only under dependency constraints.
To simplify the presentation, we use the notation $\io{c}{X}{w}$ as a
shortcut for $\In(c,X_1)\cdot \ldots \cdot \In(c,X_\ell)\cdot \Out(c,w_1)
\cdot \ldots \cdot \Out(c,w_k)$ assuming that $\vect{X} =
(X_1,\ldots,X_\ell)$ and $\vect{w} = (w_1,\ldots,w_k)$. Note that
$\vect{X}$ and/or $\vect{w}$ may be empty.

\begin{definition}[generation of dependency constraints]
  \label{def:alldep}
Let $c$ be a channel, and
 $\tr = \io{c_1}{X_1}{w_1}\cdot\ldots\cdot\io{c_{n}}{X_{n}}{w_{n}}$
 be a trace.
 If there exists a rank $k\leq n$ such that
 $c\och c_k$ and $c_i \och c$ for all $k < i \leq n$, then 
 $\Dep{\tr}{c} =\{\; w ~|~ w \in \vect{w_i} \mbox{ with $k \leq i \leq n$}\}$.
 Otherwise, we have that $\Dep{\tr}{c} = \varnothing$.

Then, given a trace $\tr$, we define $\AllDep{\tr}$ by
$\AllDep{\epsilon} = \emptyset$ and
$$\AllDep{\tr \cdot \io{c}{X}{w} } =
\begin{cases}
 \AllDep{\tr} \cup \{\constrd{\vect X}{\Dep{\tr}{c}}\} & \text{if }\Dep{\tr}{c}\neq\emptyset\\
 \AllDep{\tr} & \text{otherwise}
\end{cases}$$
\end{definition}

Intuitively, $\AllDep{\tr}$ corresponds to the accumulation of
the dependency constraints generated for all prefixes of $\tr$.

\begin{example}
\label{ex:taille-trois}
Let $a$, $b$, and $c$ be channels in $\C$ such that $a \och b \och c$.
The dependency constraints generated during the symbolic execution of a simple process of the form
$\proc{\{\In(a,x_a).\Out(a,u_a), \,\In(b,x_b).\Out(b,u_b),  \,\In(c,x_c).\Out(c,u_c)\}}{\Phi}$
are depicted below.
\begin{center}
  \begin{tikzpicture}[inner sep=0pt]
\draw (4,3) node 		(racine) {$\bullet$};

\draw (1,2) node 		(a) {$\bullet$};
\draw (4,2) node 		(b) {$\bullet$};
\draw (7,2) node 		(c) {$\bullet$};

\draw (0,1) node 		(ab) {$\bullet$};
\draw (2,1) node 		(ac) {$\bullet$};
\draw (3,1) node 		(ba) {$\bullet$};
\draw (5,1) node 		(bc) {$\bullet$};
\draw (6,1) node 		(ca) {$\bullet$};
\draw (8,1) node 		(cb) {$\bullet$};

\draw (0,0) node 		(abc) {$\bullet$};
\draw (2,0) node 		(acb) {$\bullet$};
\draw (3,0) node 		(bac) {$\bullet$};
\draw (5,0) node 		(bca) {$\bullet$};
\draw (6,0) node 		(cab) {$\bullet$};
\draw (8,0) node 		(cba) {$\bullet$};

\path (racine) edge  node[left]{$\mathtt{io}_{a}\;\;\;\;\;$} (a);
\path (racine) edge  node[left]{$\mathtt{io}_{b}\;$} (b);
\path (racine) edge  node[right]{$\;\;\;\;\mathtt{io}_{c}$} (c);

\path (a) edge node[left]{$\mathtt{io}_{b}\;\;$} (ab);
\path (a) edge node[right]{$\;\;\mathtt{io}_{c}$} (ac);
\path (b) edge node[left]{$\mathtt{io}_{a}\;\;$} (ba);
\path (b) edge node[right]{$\;\;\mathtt{io}_{c}$} (bc);
\path (c) edge node[left,above]{$\mathtt{io}_{a}\;\;\;\;\;\;$} (ca);
\path (c) edge node[left,below]{$\mathtt{io}_{b}\;\;\,$} (cb);

\path (ab) edge node[left]{$\mathtt{io}_{c}\;$} (abc);
\path (ac) edge node[left]{$\mathtt{io}_{b}\;$} (acb);
\path (ba) edge node[left]{$\mathtt{io}_{c}\;$} (bac);
\path (bc) edge node[left]{$\mathtt{io}_{a}\;$} (bca);
\path (ca) edge node[left,pos=0.6]{$\mathtt{io}_{b}\;$} (cab);
\path (cb) edge node[left]{$\mathtt{io}_{a}\;$} (cba);

\draw[->,>=latex,color=blue, line width=1pt] (acb) to[bend right] (ac);
\draw[->,>=latex,color=blue, line width=1pt] (ba) to[bend right] (b);
\draw[->,>=latex,color=blue, line width=1pt] (bca) to[bend right] (bc);
\draw[->,>=latex,color=blue, line width=1pt] (ca) to[out=20,in=250] (c);
\draw[->,>=latex,color=blue, line width=1pt] (cb) to[bend right] (c);
\draw[->,>=latex,color=blue, line width=1pt] (cba) to[bend right] (cb);

\draw[-latex,color=red,dashed,line width=1pt] (cab) .. controls +(20:8mm) and +(290:4mm) .. (c);
\draw[-latex,color=red,dashed,line width=1pt] (cab) .. controls +(20:4mm) .. (ca);
\end{tikzpicture}
\end{center}

We use $\mathsf{io}_i$ as a shortcut for $\In(i,X_i)\cdot\Out(i,w_i)$ and
we represent  dependency constraints using arrows.
For instance, on the trace $\mathsf{io}_a\cdot\mathsf{io}_c\cdot\mathsf{io}_b$, a dependency constraint
of the form $\constrd{X_b}{w_c}$ (represented by the left-most arrow)
is generated. 
    Now, on the trace
    $\mathsf{io}_c\cdot\mathsf{io}_a\cdot\mathsf{io}_b$
    we add $\constrd{X_a}{w_c}$ after the second transition,
    and $\constrd{X_b}{\{w_c,w_a\}}$ (represented by the 
    dashed $2$-arrow) after the third transition. Intuitively,
    the latter constraint expresses that
    $\mathsf{io}_b$ is only allowed to come after $\mathsf{io}_c$ if it 
    depends on it, possibly indirectly through $\mathsf{io}_a$.
\end{example}

Dependency constraints give rise to a new notion of trace equivalence,
which further refines the previous ones.

\begin{definition}[reduced trace equivalence]
\label{def:reduced-equivalence}
Let $A = \proc{\p}{\Phi}$ and $B  = \proc{\q}{\Psi}$ be two extended simple
processes. We have that $A\incesymdff B$ when, for every sequence $\tr$
such that $\triple{\p}{\Phi}{\emptyset} \ssymc{\;\tr}
  \triple{\p'}{\Phi'}{\Set_A}$, for every $\theta \in
  \Sol(\Phi';\Set_A)$ such that $\theta \models_{\cs{\Phi'}{\Set_A}}
  \AllDep{\tr}$, we have that:
\begin{itemize}
\item $\triple{\q}{\Psi}{\varnothing} \ssymc{\;\tr}
  \triple{\q'}{\Psi'}{\Set_B}$ with $\theta \in
  \Sol(\Psi';\Set_B)$, and $\theta \models_{\cs{\Psi'}{\Set_B}} \AllDep{\tr}$;
\item $\Phi'\lambda^A_\theta \statequiv
  \Psi'\lambda^B_\theta$ where $\lambda^A_\theta$ (resp. $\lambda^B_\theta$) is the substitution
  associated to $\theta$ w.r.t. $\cs{\Phi'}{\Set_A}$
  (resp. $\cs{\Psi'}{\Set_B}$).
\end{itemize}
We have that $A$ and $B$ are in \emph{reduced trace equivalence}, {denoted $A 
\esymdff B$}, if $A \incesymdff B$ and $B \incesymdff A$.
\end{definition}

\subsection{Soundness and completeness}
In order to establish that $\esymc$ and $\esymdff$ coincide, we shall
study more carefully concrete traces, consisting of proper blocks
possibly followed by a single improper block.
We will then define a precise characterization of executions whose associated
solution satisfies dependency constraints.
We denote by~$\Bio$ the set of blocks $\io{c}{M}{w}$ such that
$c\in\mathcal{C}$, $M_i\in\T(\Sigma, \W)$ for each $M_i\in\vect M$, and $w_j
\in \W$ for each $w_j \in \vect{w}$.
In this section, a concrete trace is seen as a sequence of blocks,
\ie it belongs to $\Bio^*$.

\begin{definition}[independence between blocks]
\label{def:independence}
Two  blocks $b_1 = \io{c_1}{M_1}{w_1}$ and $b_2 = \io{c_2}{M_2}{w_2}$ 
are \emph{independent}, written
$b_1 \mid\mid b_2$, when $c_1 \neq c_2$ and none of the variables of
$\vect{w_2}$ \david{occurs} in $\vect{M_1}$, and none of the variables of
$\vect{w_1}$ \david{occurs} in $\vect{M_2}$. Otherwise the blocks are \emph{dependent}. 
\end{definition}

It is easy to see that
independent blocks that are proper can be permuted in a compressed trace without
affecting the executability and the result of executing that
trace. It is not the case for improper blocks, which
can only be performed at the very end of a compressed execution.

However, this notion of independence based on recipes is too restrictive:
it may introduce spurious dependencies.
Indeed, it is often possible to make two blocks dependent
by slightly modifying recipes without altering the inputted messages.
For instance, $w'$ does not occur in recipe $M = w$ but does in
$M'=\pi_1(\langle w, w' \rangle)$ while $M'$ induces the same message as~$M$.
We thus define a more permissive notion of equivalence over traces, which allows
permutations of independent blocks but also changes of recipes that
preserve messages.
During these permutations, we require that (concrete) traces remain
\emph{plausible}.

\begin{definition}[plausible]
A trace $\tr$ is
\emph{plausible} if for any input $\In(c,M)$ such that $\tr =
\tr_0 \cdot \In(c,M) \cdot \tr_2$, we have $M \in \T(\Sigma,\W_0)$ where $\W_0$ is the set of
handles occurring in~$\tr_0$.
\end{definition}

Given two blocks $b_1 = \io{c_1}{M_1}{w_1}$ and $b_2
=\io{c_2}{M_2}{w_2}$, 
we note  $(b_1 =_\E b_2)\Phi$ when $\vect{M_1}\Phi =_\E
\vect{M_2}\Phi$, $\valid(M_1\Phi)$, $\valid(M_2\Phi)$,
and $\vect{w_1} = \vect{w_2}$. Intuitively, the two blocks only differ
by a change of recipes such that the underlying messages are kept unchanged.
We lift this notion to sequences of blocks,
  \ie $(\tr =_\E \tr')\Phi$, in the natural way.

\begin{definition}
\label{def:equiv-phi}
Given a frame $\Phi$, the relation $\equiv_\Phi$ is the smallest
equivalence over plausible traces (made of blocks) such that:
\begin{enumerate}
\item $\tr \cdot b_1 \cdot b_2 \cdot \tr' \equiv_\Phi \tr \cdot b_2 \cdot b_1 
  \cdot \tr'$ when $b_1 \mid\mid b_2$; and 
\item $\tr \cdot b_1 \cdot \tr' \equiv_\Phi \tr \cdot b_2 \cdot \tr'$ when $(b_1 =_\E
  b_2)\Phi$.
\end{enumerate}
\end{definition}



\begin{lemma}
\label{lem:permute-propre-concret}
Let  $A \sintc{\tr} \proc{\p}{\Phi}$ with $\tr$ be a trace made of proper
blocks. We have that $A \sintc{\tr'} \proc{\p}{\Phi}$ for any $\tr'
\equiv_\Phi \tr$.
\end{lemma}

This result is easily proved, following from the
fact that proper compressed executions are preserved by the two generators of
$\equiv_\Phi$.  The first case is given by
Lemma~\ref{lem:permute}. The second one follows from a simple
observation of the transition rules: only the derived messages
matter, while the recipes that are used to derive them are irrelevant
(as long as validity is ensured).

We established that compressed executions are preserved by changes
of traces within~$\equiv_\Phi$-equivalence classes.
We shall now prove that, by keeping only executions satisfying dependency 
constraints, we actually select exactly one representative in this class.

We lift the ordering on channels to blocks: $\io{c}{M}{w} \prec
\io{c'}{M'}{w'}$ if and only if $c \prec c'$. Finally, we define
$\prec$ on concrete traces as the lexicographic extension of the order on
blocks. 
Given a frame $\Phi$, we say that a plausible trace $\tr$ is
$\Phi$-minimal if it is minimal in its equivalence class modulo $\equiv_\Phi$.

\begin{lemma}
\label{lem:minimal}
Let $A \ssymc{\tr} \triple{\p}{\Phi}{\Set}$
  and $\theta\in\Sol\cs{\Phi}{\Set}$.
  We have that $\tr\theta$ is $\Phi\lambda_\theta$-minimal if, and
  only if, $\theta \models_{\cs{\Phi}{\Set}} \AllDep{\tr}$.
\end{lemma}

\begin{proof}
Let $A$ and $\triple{\p}{\Phi}{\Set}$ be such that $A \ssymc{\tr}
\triple{\p}{\Phi}{\Set}$ and $\theta \in\Sol\cs{\Phi}{\Set}$.
Let $\lambda_\theta$ be the substitution associated to $\theta$
w.r.t. $\cs{\Phi}{\Set}$.

\smallskip{}

\noindent $(\Rightarrow)$ We first show that if $\tr\theta$ is
$\Phi\lambda_\theta$-minimal then $\theta \models_{\cs{\Phi}{\Set}}
\AllDep{\tr}$, by induction on the length of the trace $\tr$.
The base case, \ie $\tr=\epsilon$, is straightforward since
$\AllDep{\tr} = \emptyset$.
Now, assume that $\tr=\tr_0 \cdot b$ for some block $b$ and $A \ssymc{\tr_0}
\triple{\p_0}{\Phi_0}{\Set_0} \ssymc{b}
\triple{\p}{\Phi}{\Set}$. Let~$\theta_0$ be the substitution
$\theta$ restricted to variables occurring in
$\cs{\Phi_0}{\Set_0}$, and $\lambda_{\theta_0}$ be the associated
first-order substitution. We have that $\theta_0 \in
\Sol\cs{\Phi_0}{\Set_0}$, $\lambda_{\theta_0}$ coincides with $\lambda_\theta$
on variables occurring in $\cs{\Phi_0}{\Set_0}$, and
$\Phi_0\lambda_{\theta_0}$ coincides with $\Phi\lambda_\theta$ on the
domain of $\Phi_0\lambda_{\theta_0}$. As a prefix of $\tr\theta$,
we have that $\tr_0\theta$ is $\Phi\lambda_{\theta}$-minimal. 
We can thus apply our induction hypothesis on  $A \ssymc{\tr_0}
\triple{\p_0}{\Phi_0}{\Set_0}$ and $\theta_0
\in\Sol\cs{\Phi_0}{\Set_0}$.
Assume that $b = \io{c}{X}{w}$.
If $\Dep{\tr_0}{c} = \emptyset$, we immediately conclude.
Otherwise, it only remains to show that $\theta \models_{\cs{\Phi}{\Set}}
\constrd{\vect X}{\Dep{\tr_0}{c}}$.
By definition of the generation of dependency constraints,  we know that  
$\tr_0$ is of the form
$\tr'_0 \cdot b_{c_0} \cdot b_{c_1}\cdot \ldots \cdot b_{c_n}$ where:
\begin{itemize}
\item $\forall\,0\le i\le n,\, b_{c_i}=\myio_{c_i}(\vect{X_i},\vect{w_i})$,
\item $c\prec c_0$ and $c_i \prec c$ for all $0 < i \leq n$; and
\item $\Dep{\tr}{c} = \{ w ~|~ w \in \vect{w_i} \mbox{ with } 0 \leq i
    \leq n\}$.
\end{itemize} 
Assume that the dependency constraint is not satisfied, this
means that for some $\vect M$ such that
$(\vect X\theta)(\Phi\lambda_\theta) =_\E (\vect
M)(\Phi\lambda_\theta)$ and $\valid((\vect
M)(\Phi\lambda_\theta))$,
we have that $\fv^1(\vect M) \cap \{w ~|~ w
\in \vect{w_i} \mbox{ with } 0 \leq i \leq n\} = \emptyset$.
Therefore, we have that
\[
\begin{array}{rll}
\tr\theta &=& \tr_0\theta\cdot b\theta \\
& =&\tr'_0\theta\cdot b_{c_0}\theta\cdot b_{c_1}\theta\cdot \ldots\cdot 
  b_{c_n}\theta\cdot \io{c}{X\theta}{w}\\
 &\equiv_{\Phi\lambda_\theta}& \tr'_0\theta\cdot \io{c}{M}{w}\cdot 
  b_{c_0}\theta\cdot b_{c_1}\theta\cdot \ldots\cdot b_{c_n}\theta.
\end{array}
\]
Since $c \prec c_0$, this would contradict the
$\Phi\lambda_\theta$-minimality of $\tr\theta$. Hence the result.

\smallskip{}
\noindent $(\Leftarrow)$
Now, assuming that $\tr\theta$ is not $\Phi\lambda_\theta$-minimal,
we shall establish that there is a dependency constraint $\constrd{\vect
  X}{\vect w}
\in \AllDep{\tr}$ that is not satisfied by $\theta$.
Let $\tr_m$ be a $\Phi\lambda_\theta$-minimal trace of the equivalence
class of $\tr\theta$. We have in particular $\tr_m
\equiv_{\Phi\lambda_\theta} \tr\theta$ and $\tr_m \prec \tr\theta$.

Let $\tr_0$ (resp.\ $\tr_m^0$) be the longest prefix of $\tr$
(resp.\ $\tr_m$) such that $(\tr_0\theta =_\E
\tr_m^0)\Phi\lambda_\theta$. We have that $\tr = \tr_0\cdot b\cdot\tr_1$  and
$\tr_m = \tr_m^0\cdot b_m\cdot \tr_m^1$ with $c_m \prec c$ where $c_m$
(resp.\ $c$) 
is the channel used in block $b_m$ (resp.\ $b$). 
By definition of $\equiv_{\Phi\lambda_\theta}$, block $b_m$ must have
a counterpart in $\tr$ and, more precisely, in $\tr_1$. We thus have a
more precise decomposition of $\tr$: $\tr =
\tr_0\cdot b\cdot\tr_{11}\cdot b'_m\cdot \tr_{12}$ such that $(b'_m\theta =_\E
b_m)\Phi\lambda_\theta$.

Let $b'_m = \io{c_m}{X}{w}$. We now show that
the constraint 
$\constrd{\vect X}{\Dep{\tr_0\cdot b\cdot \tr_{11}}}{c_m}$ is in $\AllDep{\tr}$
and is not satisfied by $\theta$, implying 
$\theta \not\models_{\cs{\Phi}{\Set}} \AllDep{\tr}$. We have seen that 
$(b'_m\theta =_\E b_m)\Phi\lambda_\theta$ and $c_m \prec c$.
Since $c_m \prec c$, by definition of $\Dep{\cdot}{\cdot}$ we deduce that 
$\emptyset\neq\Dep{\tr_0\cdot b\cdot \tr_{11}}{c_m}
\subseteq \{\; w \;|\; \Out(d,w) \mbox{ occurs in } b\cdot \tr_{11}\text{ for some 
}d,w \;\}$.
But, since we also know that $(b'_m\theta =_\E b_m)\Phi\lambda_\theta$
and $\tr_m = \tr^0_m\cdot b_m\cdot \tr_m^1$ is a plausible trace, we have that
$b_m = \io{c_m}{M}{w_m}$ for some recipes $M$ such that
$\vect{M}\Phi\lambda_\theta =_\E (\vect{X}\theta)\Phi\lambda_\theta$, $\valid(\vect{M}\Phi\lambda_\theta)$,
and $\fv^1(M) \cap \Dep{\tr_0\cdot b\cdot \tr_{11}}{c_m} = \emptyset$. This allows
us to conclude that $\theta \not\models_{\cs{\Phi}{\Set}} \AllDep{\tr}$.
\end{proof}

We are now able to show that the notion of trace equivalence
  based on this reduced semantics coincides with the compressed one
  (as well as its symbolic counterpart as given in
  Definition~\ref{def:equiv-comp-symb}).
Even though the reduced
semantics is based on the symbolic compressed semantics, it is more natural to
establish the theorem by going back to the concrete compressed semantics,
because we have to consider a concrete execution to check whether
dependency constraints are satisfied or not in our reduced semantics
anyway.

\begin{theorem}
\label{theo:reduced}
For any extended simple processes $A$ and $B$, we have that:
\begin{center}
$A \inceintc B$ if and only if $A \incesymdff B$.
\end{center}
\end{theorem}

\begin{proof}
Let $A = \proc{\p}{\Phi}$ and $B = \proc{\q}{\Psi}$ be two extended simple processes.

\smallskip{}
\noindent ($\Rightarrow$)
Consider an execution of the form $\triple{\p}{\Phi}{\emptyset}
\ssymc{\tr_s} \triple{\p_s}{\Phi_s}{\Set_A}$ and a substitution $\theta \in
\Sol\cs{\Phi_s}{\Set_A}$ such that $\theta \models_{\cs{\Phi_s}{\Set_A}}
\AllDep{\tr_s}$. Thanks to Proposition~\ref{pro:symb-soundness}, we have
that $\proc{\p}{\Phi} \sintc{\tr_s\theta}
\proc{\p_s\lambda^A_\theta}{\Phi_s\lambda^A_\theta}$ where 
$\lambda^A_\theta$ is the
substitution associated to $\theta$ w.r.t. $\cs{\Phi_s}{\Set_A}$.

Since $A \inceintc B$, we deduce that there exists
$\proc{\q'}{\Psi'}$ such that:
\[
B \stackrel{\mathsf{def}}{=} \proc{\q}{\Psi}\sintc{\tr_s\theta}
\proc{\q'}{\Psi'} \mbox{ and } \Phi_s\lambda^A_\theta
\statequiv \Psi'
\]
Relying on Proposition~\ref{pro:symb-completeness}, we deduce that
there exists $\triple{\q_s}{\Psi_s}{\Set_B}$ such that:
\[
\triple{\q}{\Psi}{\emptyset} \ssymc{\tr_s}
\triple{\q_s}{\Psi_s}{\Set_B}, \theta \in \Sol\cs{\Psi_s}{\Set_B} \mbox{
  and } \proc{\q_s\lambda^B_{\theta}}{\Psi_s\lambda^B_{\theta}} = \proc{\q'}{\Psi'}
\]
where $\lambda^B_{\theta}$ is the
substitution associated to $\theta$ w.r.t. $\cs{\Psi_s}{\Set_B}$.
The fact that we get the same symbolic trace $\tr_s$ and same solution $\theta$ comes from
the third point of Proposition~\ref{pro:symb-completeness} 
and the flexibility of the symbolic semantics $\ssymc{\cdot}$ that allows us to choose
second order variables of our choice 
(as long as they are fresh).

Lemma~\ref{lem:minimal} tells us that $\tr\theta$ is
$\Phi_s\lambda^A_\theta$-minimal. 
Since $\Phi_s\lambda^A_\theta
\statequiv \Psi_s\lambda^B_\theta$, we easily deduce that $\tr_s\theta$
is also $\Psi_s\lambda^B_\theta$-minimal, and thus 
Lemma~\ref{lem:minimal} tells us that $\theta  \models_{\cs{\Psi_s}{\Set_B}}
\AllDep{\tr}$. This allows us to conclude.

\smallskip{}
\noindent ($\Leftarrow$)
Consider an execution of the form $\proc{\p}{\Phi}
  \sintc{\tr} \proc{\p'}{\Phi'}$. We prove the result by induction on
the number of blocks involved in~$\tr$, and
we distinguish two cases depending on whether~$\tr$ ends with an
improper block or not.

\smallskip{}
\noindent \stef{\underline{Case where $\tr$ is made of proper blocks.}}
Let $\tr_m$ be a $\Phi'$-minimal trace in the
equivalence class of $\tr$.
Lemma~\ref{lem:permute-propre-concret} tells us that
$\proc{\p}{\Phi} \sintc{\tr_m} \proc{\p'}{\Phi'}$.
Thanks to Proposition~\ref{pro:symb-completeness}, we know that there
exist $\triple{\p_s}{\Phi_s}{\Set_A}$, $\tr^s_m$, and $\theta$ such
that:
\[
\triple{\p}{\Phi}{\emptyset} \ssymc{\tr^s_m}
\triple{\p_s}{\Phi_s}{\Set_A}, 
\theta \in \Sol\cs{\Phi_s}{\Set_A},
\proc{\p_s\lambda^A_\theta}{\Phi_s\lambda^A_\theta} =
\proc{\p'}{\Phi'}, \mbox{ and }
\tr^s_m\theta =\tr_m.
\]
Using Lemma~\ref{lem:minimal}, we deduce that $\theta
\models_{\cs{\Phi_s}{\Set_A}} \AllDep{\tr^s_m}$.
By hypothesis, $A \incesymdff B$, hence:
\begin{itemize}
\item $\triple{\q}{\Psi}{\varnothing} \ssymc{\;\tr^s_m}
  \triple{\q_s}{\Psi_s}{\Set_B}$ with $\theta \in
  \Sol(\Psi_s;\Set_B)$, and $\theta \models_{\cs{\Psi_s}{\Set_B}} \AllDep{\tr^s_m}$;
\item $\Phi_s\lambda^A_{\theta} \statequiv
  \Psi_s\lambda^B_{\theta}$ where $\lambda^B_{\theta}$ is the substitution
  associated to $\theta$ w.r.t. $\cs{\Phi_s}{\Set_B}$
\end{itemize}
Thanks to Proposition~\ref{pro:symb-soundness}, we deduce that 
\[
\proc{\q}{\Psi} \sintc{\tr^s_m\theta} \proc{\q_s\lambda^B_\theta}{\Psi_s\lambda^B_\theta}.
\]
Moreover, since $\Phi' = \Phi_s\lambda^A_\theta \statequiv \Psi_s\lambda^B_\theta$, we get
$\tr_m \equiv_{\Psi_s\lambda^B_\theta} \tr$ from the fact that $\tr_m
\equiv_{\Phi'} \tr$. Applying Lemma~\ref{lem:permute-propre-concret},
we conclude that
\[
\proc{\q}{\Psi} \sintc{\tr}
\proc{\q_s\lambda^B_\theta}{\Psi_s\lambda^B_\theta} \mbox{ with }
\Phi' \statequiv \Psi_s\lambda^B_\theta.
\]
\smallskip{}

\noindent \stef{\underline{Case where
  $\tr$ is of the form $\tr_0\cdot b$ where $b$ is an improper block.}}
We have that:
\[
\proc{\p}{\Phi}\sintc{\tr_0} \proc{\p'}{\Phi'}
\sintc{b} \proc{\emptyset}{\Phi'}
\]
Let $\tr_m$ be a $\Phi'$-minimal trace in the
equivalence class of $\tr$.
By definition of the relation $\equiv$,  block $b$ must have a
counterpart in $\tr_m$. We thus have that $\tr_m$ is of the form
$\tr_m = \tr_1\cdot b_m\cdot \tr_2$ where $b_m$ is the improper block
corresponding to~$b$. We do not necessarily have that $b = b_m$ but we
know that $(b =_\E b_m)\Phi'$. If $\tr_2$ is an empty trace, \ie $b_m$
is at the end of $\tr_m$, the reasoning from the previous case
applies.

%
Otherwise, we have that $\tr_1\cdot \tr_2 \equiv_{\Phi'} \tr_0$, 
$\tr_1\cdot \tr_2$ is $\Phi'$-minimal, and $\tr_2$ is non
  empty. Thus, thanks to 
Lemma~\ref{lem:permute-propre-concret}, we
have that:
\[
\proc{\p}{\Phi} \sintc{\tr_1} \proc{\p_1}{\Phi_1} \sintc{\tr_2} \proc{\p'}{\Phi'}
\]
Since $\tr_1\cdot \tr_2$ is made of proper blocks, we can apply our
previous reasoning, and conclude that there
exist $\proc{\q_1}{\Psi_1}$ and $\proc{\q'}{\Psi'}$ such that:
\[
\proc{\q}{\Psi} \sintc{\tr_1} \proc{\q_1}{\Psi_1} \sintc{\tr_2}
\proc{\q'}{{\Psi'}}
\;\mbox{ and }\; \Phi' \statequiv \Psi'
\;\mbox{ (and thus } (b =_\E b_m)\Psi' \mbox{).}
\]

Since we know that $\proc{\p'}{\Phi'} \sintc{b}
\proc{\emptyset}{\Phi'}$, and $\tr_0\cdot b \equiv_{\Phi'} \tr_1\cdot b_m\cdot \tr_2$,
we deduce that $\proc{\p_1}{\Phi_1} \sintc{b_m}
\proc{\emptyset}{\Phi_1}$. We have that $\proc{\p}{\Phi}
\sintc{\tr_1\cdot b_m} \proc{\emptyset}{\Phi_1}$, and $\tr_1\cdot b_m$ is a
$\Phi_1$-minimal trace (note that the improper block is at the end).
Thus, applying our induction hypothesis, we have that:
\[
\proc{\q}{\Psi} \sintc{\tr_1} \proc{\q_1}{\Psi_1} \sintc{b_m} \proc{\emptyset}{\Psi_1}.
\]

Since the channel used in $b_m$ does not
occur in $\tr_2$, we deduce that
\[
\proc{\q}{\Psi} \sintc{\tr_1} \proc{\q_1}{\Psi_1} \sintc{\tr_2}
\proc{\q'}{\Psi'} \sintc{b_m} \proc{\emptyset}{\Psi'}
\]
Relying on Lemma~\ref{lem:permute-propre-concret} and the fact that
$(b_m =_\E b)\Psi'$, we deduce that 
\[
\proc{\q}{\Psi} \sintc{\tr_0}
\proc{\q'}{\Psi'} \sintc{b} \proc{\emptyset}{\Psi'}.\tag*{\qEd}
\]

\def\popQED{}
\end{proof}

Putting together Theorem~\ref{theo:comp-soundness-completeness} and Theorem~\ref{theo:reduced}, we are now able to
state our main result: our notion of reduced trace equivalence
actually coincides with the usual notion of trace equivalence.
\lucca{This result is generic and holds for an arbitrary equational
  theory, as well as for
  an arbitrary notion of validity (as defined in
  Section~\ref{subsec:syntax}).}

\begin{corollary}
For any initial simple processes $A$ and $B$, we have that:
\begin{center}
  $A \eint B$ if and only if $A \esymdff B$.
\end{center}
\end{corollary}

%% file: apte.tex
\section{Integration in \apte}
\label{sec:apte}

We validate our approach by integrating our refined semantics
in the \apte tool. As we shall see, the compressed semantics
can easily be used as a replacement for the usual semantics in verification 
algorithms. However, exploiting the reduced semantics is not trivial,
and requires to adapt the constraint resolution procedure. 

\stef{It is beyond the scope of this paper to provide a detailed summary of
how the verification tool \apte actually works. A 50 pages paper
describing solely the
constraint resolution procedure of \apte
is available~\cite{CCD-ic16}. \david{This procedure manipulates matrices
of constraint systems, with additional kinds of constraints necessary for its
inner workings.} Proofs of the soundness,
completeness and termination of the algorithm are available in a long
and technical appendix (more than 100 pages).}

\stef{In order to show how our reduced semantics have been integrated
  in the constraint solving procedure of \apte, we choose to provide a
high-level axiomatic presentation of \apte's algorithm. This
allows us to prove that our integration is correct \david{without having
to enter into complex, unnecessary details}
of \apte's algorithm.
\david{Our axioms} are consequences of
results stated and proved in~\cite{Cheval-phd2012} and have been
written in concertation with Vincent Cheval. However, due to 
some changes in the presentation, proving them will require to adapt
most of the proofs. It is therefore beyond the scope of this paper to
\david{formally prove that our axioms are satisfied by the concrete
procedure.}}


We start this section
with a high-level axiomatic presentation of \apte's algorithm,
following the original procedure \cite{cheval-ccs2011} but
assuming public channels only (sections~\ref{subsec:nutshell}, \ref{subsec:specs}).
\stef{The purpose of this  presentation is to
provide enough details about \apte to explain how our
optimisations have been integrated, leaving out unimportant details.}
\lucca{Next, we show that this axiomatization is sufficient to prove soundness and
completeness of \apte w.r.t.~ trace equivalence (Section~\ref{subsec:proof}).}
Then we explain the simplifications induced by the restriction to simple
processes, and how compressed semantics can be used to enhance the
procedure and prove the correctness of this integration (Section~\ref{subsec:apte-compression}).
We finally describe how our reduction technique can be integrated,
and prove the correctness of this integration  (Section~\ref{subsec:apte-dependency}).
\david{
We present some benchmarks in
Section~\ref{subsec:benchmarks}, showing that our integration
allows to effectively benefit from both of our
partial order reduction techniques.
}


\subsection{\apte in a nutshell}
\label{subsec:nutshell}
\apte has been designed for a fixed equational theory
 $\E_\apte$ (formally defined in Example~\ref{ex:signature-a})
containing standard cryptographic primitives.
It relies on a notion of message which requires that only constructors
are used, and a semantics in which actions are blocked unless they
are performed on such messages.
This fits in our
framework, described in Section~\ref{sec:model},
by taking $\M = \T(\Sigma_c,\N)$.

We now give a high-level description of the algorithm that is
implemented in \apte. 
The main idea is to perform all possible symbolic executions
of the processes, keeping together
the processes that can be reached using the same sequence of symbolic
action. Then, at each step of this symbolic execution, the procedure
checks that for every solution of every process on one side, there is
a corresponding solution for some process on the other side so that
the resulting frames are in static equivalence. This check
for \emph{symbolic equivalence} is not obviously decidable. To achieve
it, \apte's procedure relies on a set of rules for simplifying sets of
constraint systems. These rules are used to put constraint systems in a
\emph{solved form} that enables the efficient verification of symbolic
equivalence.

The symbolic execution used in \apte is the same as described in 
Section~\ref{sec:constraint-solving}. However, \apte's constraint resolution
procedure introduces new kinds of constraints.
Fortunately, we do not need to enter into the details of those constraints
and how they are manipulated.
Instead, we treat them axiomatically.

\begin{definition}[extended constraint system/symbolic process]
\label{def:ext-constraint-system}
An \emph{extended constraint system}
$\C^+ = \csext{\Phi}{\Set}{\Set^+}$ consists of a
constraint system $\C = \cs{\Phi}{\Set}$ together with an additional set
$\Set^+$ of \emph{extended constraints}. We treat this latter set abstractly,
only assuming an associated satisfaction relation,
written $\theta\models\Set^+$, such that $\theta\models\emptyset$
always holds, and $\theta\models\Set^+_1$ implies
$\theta\models\Set^+_2$ when $\Set^+_2 \subseteq \Set^+_1$.
We define the set of solutions of $\C^+$ as
$ \Sol^+(\C^+) =
\{\; \theta\in\Sol(\C) \;|\; \theta\models\Set^+ \;\}$.

\label{def:ext-symb-proc}
An \emph{extended symbolic process }
$\quadr{\p}{\Phi}{\Set}{\Set^+}$ is a symbolic process with
an additional set of extended constraints~$\Set^+$.
\end{definition}

We shall denote extended constraint systems by $\Set^+$, $\Set^+_1$,
etc. Extended symbolic processes will be denoted by $A^+$, $B^+$, etc.
Sets of extended symbolic processes will simply be denoted by
$\as$, $\bs$, etc.
For convenience, we extend $\Sol$ and $\Sol^+$ to symbolic processes
and extended symbolic processes in the natural way:
\[
\Sol\triple{\p}{\Phi}{\Set} = \Sol\cs{\Phi}{\Set}
\quad\mbox{and}\quad
\Sol^+\quadr{\p}{\Phi}{\Set}{\Set^+} =
\Sol^+\csext{\Phi}{\Set}{\Set^+}.
\]
We may also use the following notation to translate back and forth
between symbolic processes and extended symbolic processes:
\[
\extof{\triple{\p}{\Phi}{\Set}}=\quadr{\p}{\Phi}{\Set}{\emptyset}
\quad\mbox{and}\quad
\ofext{\quadr{\p}{\Phi}{\Set}{\Set^+}}=\triple{\p}{\Phi}{\Set}
.
\]

We can now introduce the key notion of symbolic equivalence between
sets of extended symbolic processes, or more precisely between
their underlying extended constraint systems.

\begin{definition}[symbolic equivalence]
  \label{def:symbeqset}
  Given two sets of extended symbolic processes~$\as$ and~$\bs$,
  we have that $\as\symbinclset \bs$
  if for every $A^+ = \quadr{\p_A}{\Phi_A}{\Set_A}{\Set^+_A}\in\as$, for every ${\theta \in\Sol^+(A^+)}$, there
  exists $B^+ = \quadr{\p_B}{\Phi_B}{\Set_B}{\Set^+_B} \in\bs$ such that $\theta \in\Sol^+(B^+)$ and
  $\Phi_A\lambda^A_\theta\statequiv\Phi_B\lambda^B_\theta$
where $\lambda^A_\theta$ (resp.~$\lambda^B_\theta$) is the
substitution associated to $\theta$
w.r.t. $\cs{\Phi_A}{\Set_A}$ (resp.~$\cs{\Phi_B}{\Set_B}$).
  We say that $\as$ and $\bs$ are in \emph{symbolic equivalence},
  denoted by $\as\symbeqset\bs$, if  $\as\symbinclset \bs$
  and $\bs\symbinclset \as$.
\label{def:apte:symbEqu}
\end{definition}

The whole trace equivalence procedure can finally be abstractly described by
means of a transition system $\ssyma{}$ on pairs of sets of extended
symbolic processes, labelled by observable symbolic actions.
Informally, the intent is that a pair of processes is in trace
equivalence iff only symbolically equivalent pairs
may be reached from the initial pair using $\ssyma{}$.

We now define $\ssyma{}$ formally.
A transition $(\as;\bs) \ssyma{}$ can take place iff
$\as$ and $\bs$ are in symbolic equivalence\footnote{
  This definition yields infinite executions for $\ssyma{}$ if no
  inequivalent pair is met. Each such execution eventually
  reaches $(\emptyset;\emptyset)$ while,
  in practice, executions are obviously not explored past empty pairs.
  We chose to introduce this minor gap to make the theory more uniform.
}.

Each transition for some observable action $\alpha$ consists of two steps, \ie
$(\as;\bs)\ssyma{\alpha}(\as'';\bs'')$ iff
$(\as;\bs)\ssymaone{\alpha}(\as';\bs')$ and $(\as';\bs')\ssymatwo(\as'';\bs'')$,
where the latter transitions are described below:
\begin{enumerate}
\item The first part of the transition consists in
  performing an observable symbolic action $\alpha$ (either $\In(c,X)$ or
  $\Out(c,w)$) followed by all available unobservable ($\tau$) actions.
  This is done for each extended symbolic process that occurs in
  the pair of sets, and each possible transition of one such process
  generates a new element in the target set.
  Formally, we have $(\as;\bs) \ssymaone{\alpha} (\as';\bs')$
  if
  \[
    \as' = \bigcup_{(\p;\Phi;\Set;\Set^+)\in\as} \bigl\{\;
      (\p';\Phi';\Set';\Set^+) \;|\; (\p;\Phi;\Set) \ssym{\alpha\cdot \tau^*} 
      (\p';\Phi';\Set') \not\ssym{\tau}
    \;\bigr\},
  \]
  and correspondingly for $\bs'$.
  Note that elements of $(\as;\bs)$ that cannot perform $\alpha$
  are simply discarded, and that the constraint systems
  of individual processes are enriched according to their own
  transitions whereas the extended part of constraint systems are
  left unchanged.
  For a fixed symbolic action $\alpha$, the $\ssymaone{\alpha}$
  transition is deterministic. The choice of names for
  handles and second-order variables does not matter, and therefore the 
  relation $\ssymaone{}$ is also finitely branching.
\item The second part consists in
  simplifying the constraint systems of $(\as';\bs')$
  until reaching \emph{solved forms}.
  This part of the transition is non-deterministic, \ie several different
  $(\as'';\bs'')$ may be reached depending on various choices, \eg
  whether a message is derived by using a function symbol or
  one of the available handles.
  Although branching, this part of the transition is finitely branching.
  Moreover, only extended constraints may change:
  for any $(\p;\Phi;\Set;\Set^+_1)\in\as''$ there must be 
  a $\Set^+_0$ such that $(\p;\Phi;\Set;\Set^+_0)\in\as'$,
  and similarly for $\bs''$.
\end{enumerate}

An important invariant of this construction is that all the processes
occurring in any of the two sets of processes have constraint systems
that share a common structure.
 More precisely the transitions maintain
that for any $(\p_1;\Phi_1;\Set_1;\Set^+_1),(\p_2;\Phi_2;\Set_2;\Set^+_2)\in\as\cup\bs$,
$\fvs(\Set_1)=\fvs(\Set_2)$ and
$\dedi{D}{X}{x}$ occurs in $\Set_1$ iff it occurs in $\Set_2$.

\begin{example}
\label{ex:toy1}
Consider the simple basic processes
 $R_i=\In(c_i,x_i).\testt{x_i=\ok}{\Out(c_i,n_i)}$
for $i\in\mathbb{N}, x_i\in\X, n_i\in\N$, $\ok$ a public constant.
We illustrate the roles of $\ssymaone{}$ and $\ssymatwo{}$ on the pair
$(\{Q_0\};\{Q_0\})$ where $Q_0 =
\quadr{\{R_1,R_2\}}{\emptyset}{\emptyset}{\emptyset}$.
We have that
$$
(\{Q_0\};\{Q_0\})\ssymaone{\In(c_2,X_2)}(\{Q_0^t,Q_0^e\};\{Q_0^t,Q_0^e\})
$$
where $Q_0^t$ and $Q_0^e$ are the two symbolic processes
one may obtain by executing the observable action $\In(c_2,X_2)$,
depending on the conditional after that input. Specifically, we have:
\begin{itemize}
\item 
$Q_0^t=\quadr{\{R_1,\Out(c_2,n_2)\}}
{\emptyset}
{\{\dedi{X_2}{\emptyset}{x_2},x_2\eqi \ok\}}
{\emptyset}$
\item 
$Q_0^e=\quadr{\{R_1\}}
{\emptyset}
{\{\dedi{X_2}{\emptyset}{x_2},x_2\neqi \ok\}}
{\emptyset}$
\end{itemize}
After this first step, $\ssymatwo{}$ is going to non-deterministically 
solve the constraint systems.
From the latter pair, it will produce only two alternatives.
Indeed, if $x_2\eqi \ok$ holds then \apte infers that the only recipe
that it needs to consider is the recipe $R=\ok$. In that case, the only
considered solution is $\{X_2\mapsto \ok\}$.
Otherwise, $x_2\neqi \ok$ holds but, at this point,
no more information is inferred on $X_2$.
Formally,
\[
\begin{array}[]{cl}
(\{Q_0^t,Q_0^e\};\{Q_0^t,Q_0^e\})\ssymatwo{}(\{Q_1^t\};\{Q_1^t\})&\\
(\{Q_0^t,Q_0^e\};\{Q_0^t,Q_0^e\})\ssymatwo{}(\{Q_1^e\};\{Q_1^e\})
\end{array}
\]
where
\begin{itemize}
\item
  $Q_1^t=\quadr{\{R_1,\Out(c_2,n_2)\}} {\emptyset}
  {\{\dedi{X_2}{\emptyset}{x_2},x_2\eqi \ok\}} {\Set_1^t}$
  and $\Sol^+(Q_1^t)=\{\Theta^t_1\}$ 
where
  $\Theta^t_1=\{X_2\mapsto\ok\}$;
\item
  $Q_1^e=\quadr{\{R_1\}} {\emptyset}
  {\{\dedi{X_2}{\emptyset}{x_2},x_2\neqi \ok\}} {\Set_1^e}$.
\end{itemize}
\stef{The content of $\Set_1^t$ and $\Set_1^e$ is not important.}
\stef{Note} that after ${\ssymatwo}$, only one alternative remains (\ie there is only
one extended symbolic process on each side of the resulting pair) because only one of the two processes
$Q_0^t,Q_0^e$ complies with the choices made in each branch.
\end{example}

\begin{definition}[${\eqapte}$] \label{def:eqapte}
Let $A = \proc{\p_A}{\Phi_A}$ and $B = \proc{\p_B}{\Phi_B}$ be two
processes. We say that $A \eqapte B$ when
$\as\symbeqset\bs$ for any pair $(\as;\bs)$ such that 
${(\quadr{\p_A}{\Phi_A}{\emptyset}{\emptyset};\quadr{\p_B}{\Phi_B}{\emptyset}{\emptyset})
  \ssyma{\tr}(\as;\bs)}$.
\end{definition}

As announced above, we expect $\eqapte$ to coincide with trace
equivalence. We shall actually prove it (see Section~\ref{subsec:proof}), 
after having introduced a few axioms (Section~\ref{subsec:specs}).
We note, however, that this can only hold under some minor assumptions
on processes.
In practice, \apte does not need those assumptions but they allow
for a more concise presentation.

\begin{definition}
  A simple process (resp. symbolic process)
  $A$ is said to be {\em quiescent} when
  $A\not\sint{\tau}$ (resp. $A\not\ssym{\tau}$).
  An extended symbolic process $A^+$ is {\em quiescent} when
  $\ofext{A^+}\not\ssym{\tau}$.
\end{definition}

In $\ssymaone{\alpha}$ transitions,
processes must start by executing an observable action
$\alpha$ and possibly some $\tau$ actions after that.
Hence, it does not make sense to consider $\ssyma{}$ transitions 
on processes that can still perform $\tau$ actions.
We shall thus establish that $\eqapte$ and $\esym$ coincide
only on quiescent processes,
which is not a significant restriction since
it is always possible to pre-execute all available $\tau$-actions
before testing equivalences.



\subsection{Specification of the procedure}
\label{subsec:specs}

We now list and comment the specification satisfied by the
exploration performed by \apte. These statements are consequences of
results stated and proved in~\cite{Cheval-phd2012} \stef{but}
it is beyond the scope of this paper to prove them.


\subsubsection*{Soundness and completeness of constraint resolution}
The ${\ssymatwo}$ step, corresponding to \apte's constraint resolution 
procedure, only makes sense under some assumptions on the (common) structure
of the processes that are part of the pairs of sets under consideration. Rather than 
precisely formulating these conditions (which would be at odds with the 
abstract treatment of extended constraint systems) we start by 
defining an \david{under}-approximation of the set of pairs on which we may apply
${\ssymatwo}$ at some point.
We choose this \david{under}-approximation sufficiently large to cover
pairs produced by the compressed semantics, and we then formulate our specifications in that domain.
\stef{More precisely, the \david{under}-approximation has to cover two things:
\begin{enumerate}
\item we have to consider additional disequalities
  of the form $u \neq^? u$ in constraint systems since they are eventually added by our
  compressed symbolic semantics (see Figure~\ref{figure:compressed-symbolic-sem});
\item we have to allow the removal of some extended symbolic process
  from the original sets since they are eventually discarded by our
  compressed (resp. reduced) symbolic semantics.
\end{enumerate}
}

Given an extended symbolic process $A^+ =
\quadr{\p}{\Phi}{\Set}{\Set^+}$, we \stef{denote} $\mathrm{add}(A^{+})$ 
the set of extended symbolic processes obtained
from $A^+$ by adding into $\Set$ a number of disequalities of the form $u\neq^?u$ with
$\fvp(u) \subseteq \fvp(\Set)$. 
This is then extended to sets of
extended symbolic processes as follows: $\mathrm{add}(\{A^+_1,\ldots,A^+_n\}) =
\{ \{ B^+_1, \ldots, B^+_n \} ~ | ~ B^+_i \in \mathrm{add}(A^+_i) \}$.
\begin{definition}[valid and intermediate valid pairs]
  The set of \emph{valid pairs} is the least set such that:
  \begin{itemize}
    \item
      For all quiescent, symbolic processes $A=(\p;\Phi;\emptyset)$ and $B=(\q;\Psi;\emptyset)$,
      $(\{\extof{A}\};\{\extof{B}\})$ is valid.
    \item If $(\as;\bs)$ is valid and $\as \symbeqset \bs$,
      $(\as;\bs)\ssymaone{\alpha}(\as_1;\bs_1)$,
      $\as_2\subseteq\as_1$, $\bs_2\subseteq\bs_1$,
  $\as_3 \in \mathrm{add}(\as_2)$,
  $\bs_3 \in \mathrm{add}(\bs_2)$,
      and
      $(\as_3;\bs_3)\ssymatwo(\as';\bs')$
      then $(\as';\bs')$ is valid.
      In that case, the pair $(\as_3;\bs_3)$ is called
      an \emph{intermediate valid pair}.
  \end{itemize}
\end{definition}

It immediately follows that
$(\{\extof{A}\};\{\extof{B}\})\ssyma{\tr}(\as;\bs)$ implies that
$(\as;\bs)$ is valid and only made of quiescent, extended symbolic processes.
But the notion of validity accomodates more pairs:
it will cover pairs accessible under refinements of $\ssyma{}$
based on subset restrictions of $\ssymaone{}$.
\stef{We may note that these pairs are actually pairs that would have
  been explored by $\apte$ when starting with another pair of
  processes (\eg a process that makes explicit the use of trivial conditionals
  of the form $\test{u=u}{P}{Q}$ ). Therefore, those pairs  do not
  cause any trouble when they have to be handled by $\apte$.}

\begin{axiom}[soundness of constraint resolution] \label{spec:sound}
Let $(\as';\bs')$ be an intermediate valid pair such that
$(\as';\bs')\ssymatwo(\as'';\bs'')$. Then, for all $A''\in\as''$
(resp. $B''\in\bs''$) there exists
some $A'\in\as'$ (resp. $B'\in\bs')$ such that $\ofext{A'}=\ofext{A''}$
(resp. $\ofext{B'}=\ofext{B''}$)
and $\Sol^+(A'') \subseteq \Sol^+(A')$
(resp. $\Sol^+(B'') \subseteq \Sol^+(B')$).
\end{axiom}

\apte treats almost symmetrically the two  components of the pair of sets on which
transitions take place.
This is reflected by the fact that axioms concern both sides and are completely 
symmetric, like \stef{{\bf Axiom~\ref{spec:sound}}}. 
In order to make \stef{the} following specifications more concise and readable, we state 
properties only for one of the two sets and consider \david{the other}
\stef{``symmetrically''} as well.


The completeness specification is in two parts: it first states that no 
first-order solution is lost in the constraint resolution process, and then
that the branching of $\ssymatwo$ corresponds to different second-order solutions.

\begin{axiom}[first-order completeness of constraint resolution]
  \label{spec:complete}
  Let $(\as;\bs)$ be an intermediate valid pair.
  For all $A^+\in \as$ and $\theta\in\Sol(A^+)$
  there exists $(\as;\bs)\ssymatwo(\as_2;\bs_2)$,
  $A_2^+\in\as_2$ and $\theta^+\in\Sol^+(A^+_2)$
  such that
  $\ofext{A^+_2}=\ofext{A^+}$ and
  $\lambda_\theta^{A} =_\E \lambda_{\theta^+}^{A}$,
  where $\lambda^{A}_\theta$
  (resp.~$\lambda^{A}_{\theta^+}$) is the substitution associated
  to $\theta$ (resp.~to $\theta^+$) w.r.t.~$\ofext{A^+}$. 
  Symmetrically for $B^+\in\bs$.
\end{axiom}

\begin{axiom}[second-order consistency of constraint resolution]
  \label{spec:consistent}
  Let $(\as;\bs)$ be an intermediate valid pair such that
  $(\as;\bs)\ssymatwo(\as_2;\bs_2)$,
  {$\theta\in\Sol^+(A^+)$} for some $A^+\in\as$ and
  {$\theta\in\Sol^+(C^+_2)$} for some $C^+_2\in\as_2\cup\bs_2$.
  Then there exists some $A^+_2\in\as_2$ such that $\ofext{A^+}=\ofext{A^+_2}$
  and $\theta\in\Sol^+(A^+_2)$.
  Symmetrically for $B^+\in\bs$.
\end{axiom}

\subsubsection*{Partial solution}
In order to avoid performing some explorations when dependency
constraints of our reduced semantics are not satisfied,
we shall be interested in knowing when all solutions of a given constraint
system assign a given recipe to some variable.
Such information is generally available in the solved forms computed
by \apte, but not always in a complete fashion.
We reflect this by introducing an abstract function that
represents the information that can effectively be inferred by the
procedure.

\begin{definition}[partial solution]
  \label{def:presol}
  We assume a \emph{partial solution}\footnote{We use the
notation $\sigma_1\sqcup \sigma_2$ to emphasize the fact that the
two substitutions do not interact together. They have 
disjoint domain, i.e. \stef{$\dom(\sigma_1) \cap \dom(\sigma_2) =
\emptyset$}, and 
no variable of $\dom(\sigma_i)$ occurs in $\img(\sigma_j)$ with
$\{i,j\} = \{1,2\}$.} function $\ps$
  which maps sets of extended constraints $\Set^+$ to a substitution,
  such that for any
  $\theta \in\Sol^+\quadr{\p}{\Phi}{\Set}{\Set^+}$, there exists~$\theta'$ such that
  $\theta = \ps(\Set^+) \sqcup \theta'$.
  We extend $\ps$ to extended symbolic processes:
  $\ps\quadr{\p}{\Phi}{\Set}{\Set^+} = \ps({\Set^+})$.
\end{definition}

Intuitively, given an extended constraint system, the function
  $\ps$
returns the value of some of its second-order variables (those for
which their instantiation is already completely determined).
Our specification of the partial solution shall postulate that
the partial solution returned by \apte is the same for each
extended symbolic process occurring in a pair $(\as; \bs)$ reached
during the exploration. Moreover,
there is a monotonicity property that ensures that this partial
solution becomes more precise along the exploration.

\begin{axiom}
  \label{spec:presol}
  We assume the following about the partial solution:
  \begin{enumerate}
    \item
      For any valid pair $(\as;\bs)$,
      we have that $\ps(A) = \ps(B)$ for any $A,B \in \as \cup \bs$.
      This allows us to simply write $\ps(\as;\bs)$
      when $\as \cup \bs \neq \emptyset$.
    \item
      For any intermediate valid pair $(\as;\bs)$
      such that $(\as;\bs)\ssymatwo(\as';\bs')$ and
      ${\as' \cup \bs' \neq\emptyset}$, we have
      $\ps(\as';\bs') = \ps(\as;\bs) \sqcup \theta$ for some $\theta$.
  \end{enumerate}
\end{axiom}

\begin{example}
  \label{ex:toy2}
  Continuing Example~\ref{ex:toy1},
  we first note that $(\{Q_0\};\{Q_0\})$ is a valid pair.
  Second,
  the exploration $(\{Q_0\};\{Q_0\})\ssyma{\In(c_2,X_2)}(\{Q_1^t\};\{Q_1^t\})$
  covers all executions of the form 
  $\ofext{Q_0}\ssym{\In(c_2,X_2).\tau}\ofext{Q_1^t}$ going to the
  $\mathtt{then}$
  branch even though the only solution of $Q_1^t$ is $\Theta_1^t$.
  Indeed, if $\Theta\in\Sol(\ofext{Q_1^t})$ then the message computed
  by $X_2\Theta$ should be equal to $\ok$ and thus no first-order solution is lost
  as stated by {\textbf{\normalfont{Axiom 2}}}.
  Moreover, because the value of $X_2$ is already known in $Q_1^t$,
  we may have $\ps(Q_1^t)=\ps(\Set^+_1)=\{X_2\mapsto \ok\}$.
\end{example}


\subsection{Proof of the original procedure}
\label{subsec:proof}
The procedure, axiomatized as above, can be proved correct w.r.t 
the regular symbolic semantics $\ssym{}$ and its induced trace 
equivalence~$\esym$ as defined in Section~\ref{subsec:symbolic-calculus}.
Of course, Axiom~\ref{spec:presol} is unused in this first result.
It will be used later on when implementing our reduced semantics.
We first start by establishing that all the explorations
performed by $\apte$ correspond to symbolic executions.

\stef{This result is not new and has been established from
  scratch (\ie without relying on the axioms stated in the previous
  section) in~\cite{Cheval-phd2012}. Nevertheless, we found it useful to establish that our axioms
  are sufficient to prove correctness of the original \apte procedure.
 The proofs provided in the following sections to establish
  correctness of our optimised procedure follow the same lines
as the ones  presented below.}

\begin{lemma} \label{lem:ssyma-sound}
  Let $(\as;\bs)$ be a valid pair such that
  $(\as;\bs)\ssyma{\tr}(\as';\bs')$. Then, for all
  $A'\in\as'$ there is some $A\in\as$ such that
  $\ofext{A}\ssym{\tr'}\ofext{A'}$ for some $\tr'$ with $\obse(\tr')=\tr$.
  Symmetrically for $B'\in\bs'$.
\end{lemma}
\begin{proof}
  We proceed by induction on $\tr$. 
When
  $\tr$ is empty, we have that $(\as;\bs)=(\as';\bs')$, and the
  result trivially holds.
  Otherwise we have that:
\begin{center}
 $(\as;\bs) \ssymaone{\alpha} (\as_1;\bs_1) 
  \ssymatwo (\as_2;\bs_2) \ssyma{\tr_0} (\as';\bs')$ with $\tr =
  \alpha\cdot \tr_0$.
\end{center}
  Let $A'$ be a process
  of $\as'$. By induction hypothesis we have some $A_2 \in \as_2$ such that
  $\ofext{A_2} \ssym{\tr'_0} \ofext{A'}$ with $\obse(\tr'_0)=\tr_0$. By \refS{spec:sound} there is some
  $A_1 \in \as_1$ such that $\ofext{A_1} = \ofext{A_2}$, and by definition
  of $\ssymaone{}$ we finally find some $A \in \as$ such that
  $\ofext{A} \ssym{\alpha\cdot \tau^k} \ofext{A_1}$. To sum up, we have
  $A \in \as$ such that $\ofext{A} \ssym{\tr} \ofext{A'}$ with
  $\obse(\tr') = \tr$. 
\end{proof}

We now turn to completeness results. Assuming that processes
  under study
  are in equivalence $\eqapte$ (so that \apte will not stop its exploration
  prematurely), we are able to show that any valid symbolic execution
  (\ie a symbolic execution with a solution in its resulting
  constraint system) is captured by an exploration performed by \apte.
 Actually, since \apte discards some second-order solution during its
 exploration, we can only assume that another second-order solution
 with the same associated first-order solution will be found.
 
\begin{lemma} \label{lem:ssyma-complete-first}
  Let $A=\triple{\p}{\Phi}{\emptyset}$,
  $B=\triple{\q}{\Psi}{\emptyset}$ and
   $A'=\triple{\p'}{\Phi'}{\Set'}$ be three quiescent, symbolic processes
  such that $\proc{\p}{\Phi} \eqapte \proc{\q}{\Psi}$,
  $A \ssym{\tr} A'$, and $\theta\in\Sol(A')$.
  Then there exists an \apte exploration
  $(\{\extof{A}\};\{\extof{B}\}) \ssyma{\tr_o} (\as';\bs')$ and
  some $A^+ \in \as'$, $\theta^+ \in \Sol^+(A^+)$ such that
  $\obse(\tr)=\tr_o$,
  $\ofext{A^+} = A'$ and $\lambda_\theta =_\E \lambda_{\theta^+}$,
  where $\lambda_\theta$
  (resp.~$\lambda_{\theta^+}$) is the substitution associated
  to $\theta$ (resp.~to $\theta^+$) with respect to $\cs{\Phi'}{\Set'}$. 
  Symmetrically for $B\ssym{\tr} B'$.
\end{lemma}

\begin{proof}
By hypothesis, we have that $A \ssym{\tr} A'$. We will first
reorganize this derivation to ensure that $\tau$ actions are always
performed as soon as possible.
Then,  we proceed by induction on $\obse(\tr)$.
When $\obse(\tr)$ is empty, we have that {$A'=A$} since $A$ is 
quiescent.
Let   $(\as';\bs') = (\{\extof{A}\};\{\extof{B}\})$, $A^+ = \extof{A}$,
$\theta^+ = \theta$. We have that $\theta \in \Sol(A)$ and therefore
$\theta \in \Sol^+(\extof{A})$, \ie $\theta^+ = \theta \in \Sol^+(A^+)$.
We easily conclude.

  Otherwise, consider $A \ssym{\tr_0} A_1 \ssym{\alpha\cdot \tau^k} A'$
  with $\theta\in\Sol(A')$.
Let $A' = \triple{\p'}{\Phi'}{\Set'}$ and $A_1 =
\triple{\p_1}{\Phi_1}{\Set_1}$. We have that $\Set_1 \subseteq
\Set'$.  Since $\theta\in\Sol(A')$, we also have
$\theta|_V\in\Sol(A_1)$ where $V = \fv^2(\Set_1)$.
Therefore, we apply our induction hypothesis and we obtain that there
exists an \apte exploration
  $
  (\{\extof{A}\};\{\extof{B}\})\ssyma{\tr'_0}(\as_1;\bs_1)
  $
and some $A_1^+\in\as_1$,  $\theta_1^+\in\Sol^+(A_1^+)$ such that
$\obse(\tr_0) = \tr'_0$, $\ofext{A_1^+}=A_1$, and the
first-order substitutions
associated to $\theta|_V$ and $\theta_1^+$ with respect to $(\Phi_1;
\Set_1)$ are identical.
  By hypothesis we have $\proc{\p}{\Phi} \eqapte \proc{\q}{\Psi}$,
  thus $\as_1 \symbeqset \bs_1$. Hence
  a $\ssymaone{}$ transition can take place on that pair.
  By definition of $\ssymaone{}$ and since $\ofext{A_1^+} = A_1 
  \ssym{\alpha\cdot \tau^k} A'$ with $A'$ quiescent, there must be some
  $(\as_1;\bs_1) \ssymaone{\alpha} (\as_2;\bs_2)$
  with $A_2^+\in\as_2$, $\ofext{A_2^+}=A'$.
  Thus $\theta\in\Sol(A^+_2)$ and we can apply
  \refS{spec:complete}. There exists
  $(\as_2;\bs_2)\ssymatwo(\as';\bs')$,  $A^+\in\as'$,
  $\ofext{A^+}=\ofext{A_2^+}$ and
  $\theta^+\in\Sol^+(A^+)$ such that $\ofext{A^+_2} =
  \ofext{A^+}$, and the substitutions associated to $\theta$
  (resp. $\theta^+$) w.r.t. $(\Phi';\Set')$ coincide.
To sum up, the exploration 
\begin{center}
$(\{\extof{A}\};\{\extof{B}\})\ssyma{\tr'_0}(\as_1;\bs_1)
\ssymaone{\alpha} (\as_2;\bs_2) \ssymatwo(\as';\bs')$
\end{center}
together with $A^+ \in \as'$, and $\theta^+ \in \Sol^+(A^+)$ satisfy
all the hypotheses.
\end{proof}

\begin{lemma} \label{lem:ssyma-complete-second}
  Let $A,B,A'$ be quiescent symbolic processes such that
  $A \ssym{\tr} A' = \triple{\p'}{\Phi'}{\Set'}$, $\theta \in \Sol(A')$
  and $(\{\extof{A}\};\{\extof{B}\})\ssyma{\tr_o}(\as';\bs')$
  with $\obse(\tr)=\tr_o$ and $\theta\in\Sol^+(C)$ for some $C \in \as'\cup\bs'$.
  Then there exists some $A^+ \in \as'$ such that
  $\ofext{A^+} = A'$ and $\theta\in\Sol^+(A^+)$.
  Symmetrically for $B\ssym{\tr}B'$.
\end{lemma}

\begin{proof}
  We proceed by induction on $\tr_o$.
  When $\tr_o$ is empty, we have that $A' = A$ (because~$A$ 
  is quiescent), 
  $\as' = \{\extof{A}\}$, and $\bs' = \{\extof{B}\}$. Let $A^+$ be  $\extof{A} = \extof{A'}$. We deduce that
    $\theta \in \Sol^+(A^+)$ from the fact that $\theta \in \Sol(A)$
    and $A^+ = \extof{A}$.

  We consider now the case of a non-empty execution:
  \[
  (\{\extof{A}\};\{\extof{B}\})\ssyma{\tr_o}(\as_1;\bs_1)
                               \ssymaone{\alpha}(\as_2;\bs_2)
                               \ssymatwo(\as_3;\bs_3)
  \;\mbox{ and }\;
  A \ssym{\tr} A_1 \ssym{\alpha\cdot \tau^k} A_3.
  \]

  Note that, by reordering $\tau$ actions, we can assume $A_1$ to be quiescent.
  By assumption we have $\theta\in\Sol(A_3)$, $\obse(\tr)=\tr_o$ and
  $\theta\in\Sol^+(C_3)$ for some $C_3 \in \as_3\cup\bs_3$.
  By \refS{spec:sound}, there exists some
  $C_2 \in \as_2\cup\bs_2$ such that $\theta\in\Sol^+(C_2)$.
  By definition of $\ssymaone{}$ we obtain $C_1\in\as_1\cup\bs_1$ such that
  $\ofext{C_1} \ssym{\alpha\cdot \tau^k} \ofext{C_2}$ and
  $\Set^+(C_1) = \Set^+(C_2)$ (\ie the sets of extended constraints of $C_1$ and $C_2$
  coincide).
  The first fact implies $\theta|_V\in\Sol(C_1)$ by monotonicity
  (where $V = \fv^2(\Set(C_1))$, \ie second-order variables that occur
  in the set of non-extended constraints of $C_1$),
  and the second allows us to conclude more strongly that
  $\theta|_V\in\Sol^+(C_1)$.
  Since we also have $\theta|_V\in\Sol(A_1)$ by monotonicity,
  the induction hypothesis applies and we obtain some
  $A^+_1 \in \as_1$ with $\ofext{A^+_1} = A_1$ and $\theta|_V\in\Sol^+(A^+_1)$.

  By definition of $\ssymaone{}$, and since
  $\ofext{A^+_1} \ssym{\alpha\cdot \tau^k} A_3\not\ssym{\tau}$
  ($A_3$ is quiescent by hypothesis),
  we have $A^+_2\in\as_2$
  such that $\ofext{A^+_2} = A_3$ and
  $\Set^+(A^+_1)=\Set^+(A^+_2)$. Therefore, we have that $\theta \in
  \Sol(\ofext{A^+_2})$, and 
  the fact that  $\Set^+(A^+_1)=\Set^+(A^+_2)$ allows us to say that $\theta\in\Sol^+(A^+_2)$
We can finally apply \refS{spec:consistent}
  to obtain some $A^+_3$ such that $\ofext{A^+_3} = \ofext{A^+_2} = A_3$
  and $\theta\in\Sol^+(A^+_3)$.
\end{proof}


\begin{theorem} \label{thm:eqapte-esym}
For any quiescent
    extended simple processes, we have that:
    \[
 \text{$A \esym B$ if, and only if, $A\eqapte B$}.
    \]
\end{theorem}

\begin{proof}
  Let $A_0 = \proc{\p}{\Phi}$, $B_0 = \proc{\p'}{\Phi'}$,
  $\as_0 = \{(\p;\Phi;\emptyset;\emptyset)\}$ and
  $\bs_0 = \{(\p';\Phi';\emptyset;\emptyset)\}$. We prove the two
  directions separately.

  \smallskip\noindent ($\Rightarrow$) 
  Assume $A_0 \esym B_0$ and consider some exploration
  $(\as_0;\bs_0)\ssyma{\tr_o}(\as;\bs)$. We shall establish
  that $\as \symbinclset \bs$.
  Let $A = \triple{\p_A}{\Phi_A}{\Set_A}$ be in $\as$ and $\theta\in\Sol^+(A)$.
  By Lemma~\ref{lem:ssyma-sound},
  we have $\triple{\p}{\Phi}{\emptyset} \ssym{\tr} \ofext{A}$
  such that $\obse(\tr)=\tr_o$.
  By hypothesis, there exists $B = \triple{\p_B}{\Phi_B}{\Set_B}$ such
  that 
  $\triple{\p'}{\Phi'}{\emptyset} \ssym{\tr'} B$,
  $\obse(\tr')=\obse(\tr)=\tr_o$,
  $\theta\in\Sol(B)$ and $\Phi_B\lambda_\theta^B \statequiv
  \Phi_A\lambda_\theta^A$.
  We can finally apply Lemma~\ref{lem:ssyma-complete-second},
  which tells us that
  there must be some $B^+ \in \bs$ such that $\ofext{B^+} = B$ and
  $\theta\in\Sol^+(B^+)$.

  \smallskip\noindent ($\Leftarrow$) 
  We now establish $A_0 \incesym B_0$
  assuming $A_0 \eqapte B_0$. Consider $\triple{\p}{\Phi}{\emptyset} \ssym{\tr} A$
  and $\theta\in\Sol(A)$.
  If $A$ is not quiescent, it is easy to complete the latter execution
  into $\triple{\p}{\Phi}{\emptyset} \ssym{\tr\cdot \tau^k} A' = \triple{\p_A}{\Phi_A}{\Set_A}$ and $\theta\in\Sol(A')$ such
  that $A'$ is quiescent.
  By Lemma~\ref{lem:ssyma-complete-first} we know that
  $(\as_0;\bs_0)\ssyma{\tr_o}(\as;\bs)$ with $\obse(\tr)=\tr_o$,
  $A^+ \in \as$,
  $\theta^+\in\Sol^+(A^+)$ with $A' = \ofext{A^+}$ and
  $\lambda_\theta =_\E \lambda_{\theta^+}$ where $\lambda_\theta$
  (resp. $\lambda_{\theta^+}$) is the substitution  associated to
  $\theta$ (resp. $\theta^+$) w.r.t. $(\Phi_A;\Set_A)$.
  By assumption we have $\as\symbinclset\bs$ and thus there
  exists some $B = (\p_B; \Phi_B; \Set_B; \Set^+_B) \in\bs$ with $\theta^+\in\Sol^+(B)$, 
  and $\Phi_{B}\lambda_{\theta^+}^{B} \statequiv
  \Phi_A\lambda_{\theta^+}$ where $\lambda_{\theta^+}^B$  is the
  substitution associated to $\theta^+$ w.r.t. $(\Phi_B;\Set_B)$. 
  %
  By Lemma~\ref{lem:ssyma-sound} we have
  $\triple{\p'}{\Phi'}{\emptyset} \ssym{\tr'} \ofext{B}$ with $\obse(\tr')=\tr_o=\obse(\tr)$.
  To conclude the proof, it remains to show that
  $\theta\in\Sol(\ofext{B})$  and that $\Phi_A\lambda_\theta \statequiv
  \Phi_B\lambda^B_\theta$ where 
  $\lambda^B_\theta$ is the substitution associated to $\theta$
  w.r.t. $(\Phi_B; \Set_B)$.

For any $X \in \fvs(\Set_{B}) = \fvs(\Set_A)$, we have 
$\valid((X\theta)(\Phi_{A}\lambda_{\theta^+}))$,
$\valid((X\theta^+)(\Phi_{A}\lambda_{\theta^+}))$, and
\[
(X\theta)(\Phi_{A}\lambda_{\theta^+})
  =_\E (X\theta)(\Phi_{A}\lambda_\theta)
  =_\E x_A\lambda_\theta
  =_\E x_A\lambda_{\theta^+}
  =_\E (X\theta^+)(\Phi_{A}\lambda_{\theta^+})
\]
where $x_A$ is the first-order variable associated to $X$ in $\Set_A$.
Since 
  $\Phi_{A}\lambda_{\theta^+} \sim
   \Phi_{B}\lambda_{\theta^+}^{B}$, we deduce that 
  $(X\theta)(\Phi_{B}\lambda_{\theta^+}^{B}) =_\E 
   (X\theta^+)(\Phi_{B}\lambda_{\theta^+}^{B})$,
   $\valid((X\theta)(\Phi_B\lambda_{\theta^+}^B))$ and therefore
   $\theta \in \Sol(\ofext{B})$, and its associated substitution $\lambda_\theta^B$
   w.r.t. $(\Phi_B;\Set_B)$
   coincides with $\lambda_{\theta^+}^B$, and therefore $\Phi_A\lambda_\theta \statequiv
  \Phi_B\lambda^B_\theta$ is a direct consequence of  $\Phi_{B}\lambda_{\theta^+}^{B} \statequiv
  \Phi_A\lambda_{\theta^+}$ and $\lambda_\theta =_\E \lambda_{\theta^+}$.
\end{proof}

%
%
%


\subsection{Integrating compression} 
\label{subsec:apte-compression}
We now discuss the integration of the compressed semantics of 
Section~\ref{sec:constraint-solving} as a replacement for the regular symbolic 
semantics in \apte. 

Although our compressed semantics $\ssymc{}$ has been defined as executing 
blocks rather than elementary actions, we allow ourselves to view 
it in a slightly different way in this section: we shall assume that the 
symbolic compressed semantics deals with elementary actions and enforces
that those actions, when put together,
form a prefix of a
sequence of blocks that can actually be executed (for the process under
consideration) in the compressed semantics of Section~\ref{sec:constraint-solving}. 
This can easily be obtained by means of extra annotations at the level
of processes, and we will not detail that modification. This slight change
makes it simpler to integrate compression into \apte, both in the theory
presented here and in the implementation.


\begin{definition}
 Given two sets of extended symbolic processes $\as,\bs$,
 and an observable action $\alpha$, we write
$(\as;\bs)\ssymaonec{\alpha}(\as';\bs')$  when
  \[
    \as' = \;\;\bigcup_{\mathclap{(\p;\Phi;\Set;\Set^+)\in\as}}\;\; \bigl\{\;
      (\p';\Phi';\Set';\Set^+) \;|\; (\p;\Phi;\Set) \ssymc{\alpha} 
      (\p';\Phi';\Set')\not\ssym{\tau}
    \;\bigr\},
  \]
  and similarly for $\bs'$.
  We say that $(\as;\bs)\ssymac{\alpha}(\as'';\bs'')$
  when
  $(\as;\bs)\ssymaonec{\alpha}(\as';\bs')$
  and $(\as';\bs')\ssymatwo(\as'';\bs'')$.
 
 Finally, given two simple extended processes $A = (\p_A;\Phi_A)$ and
  $B = (\p_B; \Phi_B)$, we say
  that $A \eqaptec B$ when $\as \approx^+ \bs$ for any 
  $(\{\extof{(\p_A;\Phi_A;\emptyset)}\};
    \{\extof{(\p_B;\Phi_B;\emptyset)}\})
    \ssymac{\tr} (\as; \bs)$.
\end{definition}

As expected, $\ssymaonec{}$ allows to \stef{consider} much fewer
explorations than with the original $\ssymaone{}$. It inherits the features of compression,
prioritizing outputs, not considering interleavings of outputs, executing
inputs only under focus, and preventing executions beyond improper blocks.
These constraints apply to individual processes in $\as\cup\bs$,
but we remark that they also have a global effect in $\ssymaonec{}$,
\eg all processes of $\as\cup\bs$ must start a new block simultaneously:
recall that the beginning of a block corresponds to an input after some
outputs, and such inputs can only be executed if no more outputs are
available.

\begin{example}
  \label{ex:toy3}
  Continuing Example~\ref{ex:toy2},
  there is only one non-trivial\footnote{
    We dismiss here the (infinitely many) transitions
    obtained for infeasible actions, which yield
    $(\emptyset;\emptyset)$.
  } compressed exploration of one action from the valid pair
  $(\{Q_1^t\};\{Q_1^t\})$.
  It corresponds to the output on channel $c_2$:
  $(\{Q_1^t\};\{Q_1^t\})\ssymac{\Out(c_2,w_2)}(\{Q_2\},\{Q_2\})$
    for $Q_2=\quadr{\{R_1\}} {\{w_2\refer n_2\}}
  {\{\dedi{X_2}{\emptyset}{x_2},x_2\eqi \ok\}} {\Set_2^+}$.
  In particular, for any $i \in \{1,2\}$,
  we have
  $(\{Q_1^t\};\{Q_1^t\})\ssymac{\In(c_i,X_i)}(\emptyset;\emptyset)$.
\end{example}

Observe that,
because $\ssymac{}$ is obtained from $\ssyma{}$ by a subset restriction in 
$\ssymaone{}$ up to some disequalities, 
we have that $(\as';\bs')$ is a valid pair when
$(\{\extof{A}\};\{\extof{B}\})\ssymac{\tr}(\as';\bs')$ for some quiescent,
symbolic processes $A,B$ having empty sets of constraints.
Following the same reasoning as the one performed in
Section~\ref{subsec:proof}, 
we can establish that $\esymc{}$ coincides with
$\eqaptec$.
The main difference is that $\ssymc{}$ already
ignores $\tau$-actions, and therefore we do not need to apply the 
$\obse(\cdot)$ operator.

\begin{lemma} \label{lem:ssymac-sound}
  Let $(\as;\bs)$ be a valid pair such that
  $(\as;\bs)\ssymac{\tr}(\as';\bs')$. Then, for all
  $A'\in\as'$ there is some $A\in\as$ such that
  $\ofext{A}\ssymc{\tr}\ofext{A'}$.
  Symmetrically for $B'\in\bs'$.
\end{lemma}

\begin{lemma} \label{lem:ssymac-complete-first}
  Let $A=\triple{\p}{\Phi}{\emptyset}$, $B=\triple{\q}{\Psi}{\emptyset}$, and
  $A'=\triple{\p'}{\Phi'}{\Set'}$ be three quiescent, symbolic processes
  such that $\proc{\p}{\Phi} \eqaptec \proc{\q}{\Psi}$,
  $A \ssymc{\tr} A'$ and $\theta\in\Sol(A')$.
  Then there exists an exploration
  $(\{\extof{A}\};\{\extof{B}\}) \ssymac{\tr} (\as';\bs')$ and
  some $A^+ \in \as'$, $\theta^+ \in \Sol^+(A^+)$ such that
  $\ofext{A^+} = A'$ and $\lambda_\theta =_\E \lambda_{\theta^+}$,
  where $\lambda_\theta$
  (resp.~$\lambda_{\theta^+}$) is the substitution associated
  to $\theta$ (resp.~to $\theta^+$) with respect to $\cs{\Phi'}{\Set'}$. 
  Symmetrically for $B\ssymc{\tr}B'$.
\end{lemma}

\begin{lemma} \label{lem:ssymac-complete-second}
  Let $A,B$ and $A'$ be quiescent, simple symbolic processes such that
  $A \ssymc{\tr} A' = \triple{\p'}{\Phi'}{\Set'}$, $\theta \in \Sol(A')$,
  and $(\{\extof{A}\};\{\extof{B}\})\ssymac{\tr}(\as';\bs')$
  with $\theta\in\Sol^+(C)$ for some $C \in \as'\cup\bs'$.
  Then there exists some $A^+ \in \as'$ such that
  $\ofext{A^+} = A'$ and $\theta\in\Sol^+(A^+)$.
  Symmetrically for $B\ssymc{\tr}B'$.
\end{lemma}

\begin{theorem} \label{thm:eqaptec-esymc}
 For any quiescent extended simple processes, we have that:
\begin{equation*}
\text{$A \esymc B$ if, and only if, $A \eqaptec B$.}
\end{equation*}
\end{theorem}


\subsection{Integrating dependency constraints in \apte} 
\label{subsec:apte-dependency}

We now define a final variant of \apte explorations, which integrates
the ideas of Section~\ref{sec:diff} to further reduce redundant explorations.
We can obviously generate dependency constraints in \apte, just like we
did in Section~\ref{sec:diff}, but the real difficulty is to exploit
them in constraint resolution to prune some branches of the exploration
performed by \apte.
Roughly, we shall simply stop the exploration when reaching
a state for which we know that all of its solutions
violate dependency constraints.
To do that, we rely on the notion of partial solution
introduced in Section~\ref{subsec:nutshell}.
In other words, we do not modify \apte's constraint resolution, but simply
rely on information that it already provides to know when dependency
constraints become unsatisfiable. As we shall see, this simple strategy is
very satisfying in practice.

\begin{definition} \label{def:ssymad}
 We define $\ssymad{}$ as the greatest relation contained in $\ssymac{}$
  and such that, for any symbolic processes $A$ and $B$
  with empty constraint sets,
  $(\{\extof{A}\};\{\extof{B}\})\ssymad{\tr}(\as';\bs')$ implies that
  there is no
  $\constrd{\vect X}{\vect w}\in \AllDep{\tr}$
  such that
  for all $X_i\in\vect X$ we have
  $X_i\in \dom(\ps(\as';\bs'))$,
  and $\vect w\cap\fvp(X_i\ps(\as';\bs'))=\varnothing$.

Finally, given two simple extended process $A = (\p_A; \Phi_A)$ and $B
= (\p_B; \Phi_B)$, we say that $A \eqapted B$ when $\as \symbeqset \bs$
for any pair $(\as; \bs)$ such that $((\p_A; \Phi_A; \emptyset;\emptyset);(\p_B;\Phi_B;\emptyset;\emptyset)) \ssymad{\tr} (\as;\bs)$.
\end{definition}

\begin{example}
  \label{ex:toy4}
Continuing Example~\ref{ex:toy3},
consider the following compressed exploration, where $Q_3$ contains 
the constraints $\dedi{X_2}{\emptyset}{x_2}$,
$\dedi{X_1}{\{w_2\refer n_2\}}{x_1}$,
$x_2\eqi \ok$ and $x_1\eqi\ok$:
$$(\{Q_0\};\{Q_0\})\ssymac{\In(c_2,X_2)}\ssymac{\Out(c_2,w_2)}
(\{Q_2\};\{Q_2\})\ssymac{\In(c_1,X_1)}\ldots\ssymac{\Out(c_1,w_1)}
(\{Q_3\};\{Q_3\}).
$$
Assuming that $\ps(Q_3)=\{X_2\mapsto\ok,X_1\mapsto\ok\}$ (which
is the case in the actual \apte procedure)
this compressed exploration is not explored by $\ssymad{}$ because
$$\constrd{X_1}{w_2}\in\AllDep{\iossvect{c_2}{X_2}{w_2}\cdot\iossvect{c_1}{X_1}{w_1}},\;
X_1\ps(Q_3) = \ok
\;\mbox{ and }\;
\{w_2\}\cap\fvp(\ok)=\emptyset.$$
\end{example}

\begin{lemma} \label{lem:ssymad-complete-first}
  Let $A=\triple{\p}{\Phi}{\emptyset}$, $B=\triple{\q}{\Psi}{\emptyset}$
  and $A'=\triple{\p'}{\Phi'}{\Set'}$ be quiescent, simple symbolic processes
  such that $\proc{\p}{\Phi} \eqapted \proc{\q}{\Psi}$,
  $A \ssymc{\tr} A'$, $\theta\in\Sol(A')$ and
  $\theta\models_{(\Phi';\Set')}\AllDep{\tr}$.
  Then there exists an exploration
  $(\{\extof{A}\};\{\extof{B}\}) \ssymad{\tr} (\as';\bs')$ and
  some $A^+ \in \as'$, $\theta^+ \in \Sol^+(A^+)$ such that
  $\ofext{A^+} = A'$ and $\lambda_\theta =_\E \lambda_{\theta^+}$,
  where $\lambda_\theta$
  (resp.~$\lambda_{\theta^+}$) is the substitution associated
  to $\theta$ (resp.~to $\theta^+$) with respect to $\cs{\Phi'}{\Set'}$. 
  Symmetrically for $B\ssymc{\tr}B'$.
\end{lemma}

\begin{proof}
  We proceed by induction on $\tr$. The empty case is easy.
  Otherwise, consider $A \ssymc{\tr} A_1 \ssymc{\alpha} A_3 = \triple{\p_3}{\Phi_3}{\Set_3}$ with 
  $\theta\in\Sol(A_3)$, $A_1,A_3$ quiescent, and
  $\theta\models_{(\Phi_3; \Set_3)}\AllDep{\tr\cdot \alpha}$.
Let $A_1 = \triple{\p_1}{\Phi_1}{\Set_1}$ and $V_1 = \fvs(\Set_1)$.
  We also have $\theta|_{V_1}\in\Sol(A_1)$ and
  $\theta|_{V_1}\models_{(\Phi_1; \Set_1)}\AllDep{\tr}$,
  so the induction hypothesis applies and we obtain
  $(\{\extof{A}\};\{\extof{B}\})\ssymad{\tr}(\as_1;\bs_1)$
  with $A_1^+\in\as_1$, $\ofext{A_1^+}=A_1$ and $\theta_1^+\in\Sol^+(A_1^+)$
 such that the first-order substitutions associated to
  $\theta|_{V_1}$ and $\theta_1^+$ w.r.t.\ $(\Phi_1;\Set_1)$
  coincide.
 
  By hypothesis we have $A \eqapted B$, thus $\as_1 \symbeqset \bs_1$.
  Hence a $\ssymaonec{}$ transition can take place on that pair.
  By definition of $\ssymaonec{}$ and since $\ofext{A_1^+} = A_1 
  \ssymc{\alpha} A_3$, there must be some
  $(\as_1;\bs_1) \ssymaonec{\alpha} (\as_2;\bs_2)$
  with $A_2^+\in\as_2$, $\ofext{A_2^+}=A_3$.
  Thus $\theta\in\Sol(A^+_2)$ and we can apply \refS{spec:complete}
  to obtain $(\as_2;\bs_2)\ssymatwo(\as_3;\bs_3)$ with $A^+_3\in\as_3$,
  $\ofext{A^+_3}=\ofext{A^+_2}$ and
  $\theta^+_3\in\Sol^+(A^+_3)$ such that the subsitutions associated to
  $\theta$ and $\theta^+_3$ w.r.t.\ $(\Phi_3;\Set_3)$ coincide.

  It only remains to show that this extra execution step in $\ssymac{}$
  is also present in $\ssymad{}$, \ie
  that $\ps(\as_3;\bs_3)$ does not violate
  $\AllDep{\tr\cdot \alpha}$ in the sense of Definition~\ref{def:ssymad}.
  This is because, by definition of the partial solution, we have that
  $\theta_3^+ = \ps(\as_3;\bs_3) \sqcup \tau$ for some $\tau$,
  so that if $\ps(\as_3;\bs_3)$ violated $\AllDep{\tr\cdot \alpha}$ then
  we would have $\theta_3^+\not\models_{(\Phi_3;\Set_3)}\AllDep{\tr\cdot \alpha}$.
  Since $\theta_3^+$ and $\theta$
  induce the same first-order substitutions with
  respect to $(\Phi_3;\Set_3)$,
  we would finally have
  $\theta\not\models_{(\Phi_3;\Set_3)}\AllDep{\tr\cdot \alpha}$,
  contradicting the hypothesis on $\theta$.
\end{proof}

\begin{theorem}
\label{thm:cor-com-apte-red}
For any quiescent initial simple processes $A$ and $B$,
we have that:
\begin{equation*}
\text{$A \eint B$ if, and only if, $A \eqapted B$.}
\end{equation*}
\end{theorem}

\begin{proof}
  Let $A=\proc{\p}{\Phi}$ and $B=\proc{\q}{\Psi}$ be two quiescent,
  initial simple processes. Thanks to
our previous results, we have that $A \eint B$ implies $A \eqaptec B$.
  Then, we obviously have $A \eqapted B$:
  for any $(\{\extof{\triple{\p}{\Phi}{\emptyset}}\};\{\extof{\triple{\q}{\Psi}{\emptyset}}\})\ssymad{\tr}(\as';\bs')$
  we have $(\{\extof{\triple{\p}{\Phi}{\emptyset}}\};\{\extof{\triple{\q}{\Psi}{\emptyset}}\})\ssymac{\tr}(\as';\bs')$
  by definition of $\ssymad{}$, and thus $\as' \symbeqset \bs'$
  by hypothesis.

  For the other direction, it suffices to show that
  $A \eqapted B$ implies $A \incesymdff B$.
  Let $\triple{\p}{\Phi}{\emptyset} \ssymc{\tr} A' = (\p';\Phi';\Set')$ with $\theta\in\Sol(A')$ and
  $\theta\models_{(\Phi';\Set')}\AllDep{\tr}$.
  By Lemma~\ref{lem:ssymad-complete-first}
  we have 
  $(\{\extof{\triple{\p}{\Phi}{\emptyset}}\};\{\extof{\triple{\q}{\Psi}{\emptyset}}\})\ssymad{\tr}(\as';\bs')$
  with $A^+\in\as'$, $\theta^+\in\Sol^+(A^+)$ such that $\ofext{A^+}=A'$
  and $\lambda_{\theta}^{A'} =_\E \lambda_{\theta^+}^{A'}$ where
  $\lambda_{\theta}^{A'}$ (resp. $\lambda_{\theta^+}^{A'}$) is the
  substitution associated to $\theta$ (resp. $\theta^+$) w.r.t. $(\Phi';\Set')$.

  Since $A \eqapted B$, we have $\as'\symbeqset\bs'$:
  there must be some $B^+ = (\p_{B'};\Phi_{B'};\Set_{B'};\Set_B^+) \in\bs'$
  such that $\theta^+\in\Sol^+(B^+)$ 
  and $\Phi'\lambda_{\theta^+}^{A'} \statequiv 
  \Phi_{B'}\lambda_{\theta^+}^{B'}$ where $\lambda_{\theta^+}^{B'}$ is
  the substitution associated to $\theta^+$ w.r.t. $(\Phi_{B'};\Set_{B'})$.
  By Lemma~\ref{lem:ssymac-sound} we have
  $\triple{\q}{\Psi}{\emptyset} \ssymc{\tr} \ofext{B^+}$.
  Furthermore, we can show as before (see the end of the proof of
  Theorem~\ref{thm:eqapte-esym}) that $\theta\in\Sol(B^+)$
  and
  $\Phi'\lambda_\theta^{A'} \statequiv \Phi_{B'}\lambda_\theta^{B'}$,
  where $\lambda_\theta^{B'}$ is the substitution associated
  to $\theta$ w.r.t.\ $(\Phi_{B'};\Set_{B'})$.
  Finally, by
  $\theta \models_{(\Phi';\Set')}\AllDep{\tr}$,
  $D_{(\Phi';\Set')}=D_{(\Phi_{B'};\Set_{B'})}$ (\ie sets of handles
     that second-order variables may use coincide),
  and
  $\Phi'\lambda_\theta^{A'} \statequiv \Phi_{B'}\lambda_\theta^{B'}$,
  we obtain that
  $\theta\models_{(\Phi_{B'};\Set_{B'})}\AllDep{\tr}$.
\end{proof}

\subsection{Benchmarks}
\label{subsec:benchmarks}


\input{bench}

%% file: bench.tex
The optimisations developed in the present paper have been implemented,
following the above approach, in the official version of
\apte~\cite{apte-github}.

In practice, many processes enjoy a nice property that allows one to
  ensure that non-blocking outputs will never occur: it is often the case
  that enough tests have been performed before outputting a term 
  to
  ensure its validity.

\begin{example}
\label{ex:non-blocking}
Consider the following process, where $k'$ is assumed to be valid
  (\eg because it is a pure constructor term):
\[
  \In(c,x).\testt{\dec{x}{k} = 
  \mathsf{hash}(u)}{\Out(c,\enc{\dec{x}{k}}{k'})}
\]
The term outputted during an execution is necessarily valid thanks to the test
that is performed just before this output.
\end{example}

We exploit this property in order to avoid
  adding additional disequalities when integrating compression
  in~\apte.
  Therefore, in this section,
  we will restrict ourselves to simple
  processes that are \emph{non-blocking} as defined below.
\begin{definition}
Let $(\p;\Phi)$ be a simple process. We say that $(\p;\Phi)$ is
  \emph{non-blocking}
if  $u$ is valid for
  any $\tr$, $c$, $u$, $Q'$, $\q$, $\Psi$ such that  $(\p;\Phi)\sint{\tr} (\{\Out(c,u).Q'\}\cup
  \q;\Psi)$.

\end{definition}
 This condition may be hard to check in general, but
 it is actually quite
 easy to see that it is satisfied on all of our
examples. Roughly, enough tests are performed before any output action,
and this ensures the validity of the term when the output action 
becomes reachable, as in Example~\ref{ex:non-blocking}.
\medskip{}

\david{
Our modified version of \apte can verify our optimised equivalences
in addition to the original trace equivalence. It has been integrated
into the main development line of the tool.%
}
The modifications of the code ($\approx$ 2kloc)
are summarized at
\begin{center}
  \url{https://github.com/lutcheti/APTE/compare/ref...APTE:POR}
\end{center}
\david{
For reference, the version of \apte that we are using in the benchmarks
below is available at 
\url{https://github.com/APTE/APTE/releases/tag/bench-POR-LMCS} together
with all benchmark files, in subdirectory \nolinkurl{bench/protocols}.
More details, including}
instructions for reproducing our benchmarks are available
\david{at \url{http://www.lsv.fr/~hirschi/apte_por}}.

We ran the tool (compiled with OCaml 3.12.1) on a single 2.67GHz Xeon core
(memory is not relevant) and compared three different versions:
\begin{itemize}
\item {\em reference:} the reference version without our optimisations (\ie $\eqapte$);
\item {\em compression:} using only the compression optimisation (\ie $\eqaptec{}$);
\item {\em reduction:} using both compression and reduction (\ie $\eqapted{}$).
\end{itemize}

We \david{first} show examples in which equivalence holds.
\david{They are the most significant}, because the time
spent on inequivalent processes is too sensitive to the order in which
the (depth-first) exploration is performed.



\paragraph{\textit{Toy example.}} 
We  consider a parallel composition of $n$ roles $R_i$ as defined
in Example~\ref{ex:toy1}:
$P_n := \Pi_{i=1}^n R_i$.
When executed in the regular symbolic semantics $\ssym{}$,
the $2 n$ actions of $P_n$ may be interleaved in $(2n)!/2^n$ ways in
a trace containing all actions.
In the compressed symbolic semantics $\ssymc{}$,
the actions of individual $R_i$ processes must be bundled in blocks, so there
are only $n!$ interleavings containing all actions.
In the reduced symbolic semantics $\ssymd{}$,
only one interleaving of that length remains:
the trace cannot deviate from the priority order, since the only way
to satisfy a dependency constraint would be to feed an input with a message
that cannot be derived without some previously output nonce $n_i$,
but in that case the message will not be $\ok$
and the trace won't be explored further.
Note that there is still an exponential number of 
symbolic traces in the reduced semantics
when one takes into account traces with less than $2n$ actions.

\david{We show in Figure~\ref{fig:graph-toy} the time needed
to verify $P_n \approx P_n$ for $n = 1$ to $22$ in the three versions
of \apte described above: \emph{reference}, \emph{compression} and 
\emph{reduction}.
The results, in logarithmic scale,
show that each of our optimisations}
brings an exponential speedup,
as predicted by our theoretical analysis.
Similar \david{improvements} are observed if one compares the numbers of explored
pairs rather than execution times.

\begin{figure}[htpb]
  \begin{center}
    \includegraphics[scale=1.5]{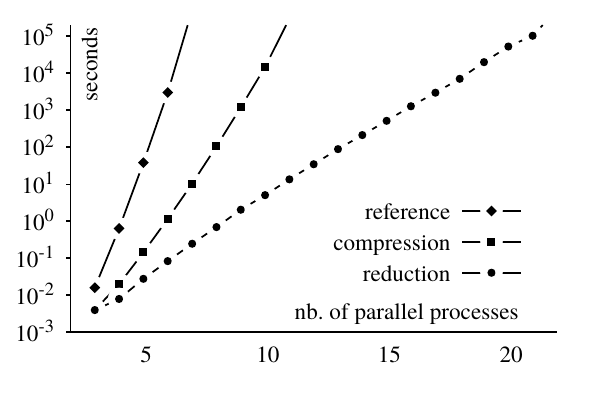}
  \end{center}
  \caption{Impact of optimisations on verification time on toy example.}
  \label{fig:graph-toy}
\end{figure}

\paragraph{\textit{Denning-Sacco protocol.}} We ran a similar benchmark,
checking that Denning-Sacco ensures strong secrecy in various scenarios.
The protocol
has three roles and we added processes playing those roles in turn,
starting with three processes in parallel.
\david{
Srong secrecy is expressed by considering, after one of the roles B,
the output of a message encrypted with the established key
on one side of the equivalence, and with a fresh key on the other side.
}
\stef{The results are plotted in
Figure~\ref{fig:graph-DS}.}
The fact that we add one role out of three at each step
explains the irregular growth in verification time. 
We still observe an exponential speedup for each optimisation.

\begin{figure}[htpb]
  \begin{center}
    \includegraphics[scale=1.5]{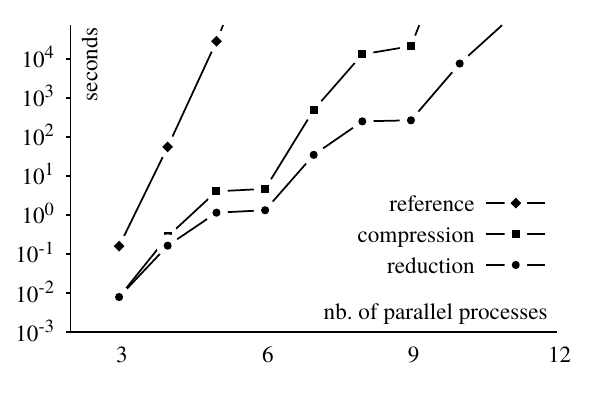}
  \end{center}
  \caption{Impact of optimisations on verification time on Denning-Sacco.}
  \label{fig:graph-DS}
\end{figure}

\paragraph{\textit{Practical impact.}}
Finally, we illustrate how our optimisations make \apte much more useful
in practice for investigating interesting scenarios.
Verifying a single session of a protocol brings little assurance
into its security. In order to detect replay attacks and to allow the
attacker to compare messages that are exchanged, at least two sessions
should be considered.
This means having at least four parallel processes for two-party protocols,
and six when a trusted third party is involved.
This is actually beyond 
what the unoptimised \apte can handle in a reasonable amount
of time.
We show \david{in Figure~\ref{fig:sessions}}
how many parallel processes could be handled in $20$ hours
by \apte on various use cases of protocols,
\david{for the same three variants of \apte as before, \ie
\emph{reference}, \emph{compression} and \emph{reduction}.}
\david{We verify an anonymity property for the Passive Authentication
protocol of e-passports. For other protocols,}
\lucca{we analyse strong secrecy of established keys:
\david{
for one of the roles we add, on one side of the equivalence,
an output encrypted by the established key and, on the other side,
an output encrypted by a fresh key.}}


\begin{figure}[htpb]
\centering
\begin{tabular}{|l|>{\centering\arraybackslash}p{2.5cm}|>{\centering\arraybackslash}p{2.5cm}|>{\centering\arraybackslash}p{2.5cm}|}\hline
 Protocol & reference & compression & reduction \\ \hline
      Needham Schroeder (3-party)& 4 & 6  & 7  \\
  Private Authentication (2-party)& 4 & 7 & 7 \\ 
Yahalom (3-party) & 4 & 5 &5 \\
E-Passport PA (2-party) & 4 & 7 & 9 \\  
Denning-Sacco (3-party) & 5 & 9  & 10\\
  Wide Mouth Frog (3-party)& 6 & 12 & 13\\
  \hline
\end{tabular}
\caption{Maximum number of parallel processes verifiable in 20 hours.}
\label{fig:sessions}
\end{figure}

\david{We finally present the benefits of our optimisations for discovering 
  attacks.}
\lucca{We performed some experiments on flawed variants of protocols,
  \david{
    shown in Figure~\ref{fig:time-attacks}, corresponding
    to example files in subdirectory
      \nolinkurl{bench/protocols/attacks/} of the above mentioned release.
    The scenario \emph{Denning-Sacco A} expresses strong secrecy of the
    (3-party) Denning-Sacco protocol,
    but this time on two instances of roles at the same time (instead of one as in Figure~\ref{fig:sessions}).
    In \emph{Denning-Sacco B}, we consider again a form of strong secrecy
    expressed by outputting encrypted messages 
    but this time at the end of role B.
    The \emph{Needham-Schroeder pub} scenario corresponds to strong secrecy of
    the public-key Needham-Schroder protocol.
    The \emph{E-Passport PA exposed} experiments show that anonymity is 
    (obviously) lost with the Passive Authentication protocol
    when the secret key is made public.
    Similarly, the \emph{Yahalom exposed} experiment shows that strong secrecy
    of Yahalom is lost when secrets keys are revealed.
  }
  Since \apte stops its exploration as soon as
  an attack is found, the time needed for \apte to
  find the attack highly depends on the order in which the depth-first 
  exploration is performed.
  However,  as shown in Figure~\ref{fig:time-attacks}, we always observe 
  in practice dramatic improvements brought by our optimisations
  compared to the reference version of \apte. 
\david{In some cases,} our 
  optimisations are even mandatory for \apte \david{to find} the attack
  \david{using reasonable resources}.
}

\begin{figure}[htpb]
\centering
\begin{tabular}{|l|>{\centering\arraybackslash}p{2.5cm}|>{\centering\arraybackslash}p{2.5cm}|>{\centering\arraybackslash}p{2.5cm}|}\hline
 Protocol & reference & compression & reduction \\ \hline
Denning-Sacco A (6 par. proc.) & OoM & 0.07s  & 0.02s\\
Denning-Sacco B (6 par. proc.) & 5.83s & 0.04s  & 0.04s\\
Needham-Shroeder pub (7 par. proc.) & TO & 0.77s & 0.67s \\
Needham-Shroeder pub (\david{5} par. proc.) & 0.79s & 0.21s & 0.13s \\
E-Passport PA exposed (8 par. proc.) & TO & 0.02s & 0.02s \\  
E-Passport PA exposed (6 par. proc.) & 4.37s & 0.03s & 0.02s \\  
Yahalom exposed (4 par. proc.) & 7.24s & 0.02s & 0.02s \\
 \hline
\end{tabular}
\caption{Impact of optimisations for finding attacks (OoM denotes a consumption of $>$32Go of RAM and TO denotes a running time of $>$20 hours).}
\label{fig:time-attacks}
\end{figure}

%% file: related-work.tex
\section{Related Work}
\label{sec:relWork}

The techniques we have presented borrow from standard ideas
from concurrency theory and trace theory.
Blending all these ingredients, and adapting them to the
demanding framework of security protocols, we have come up with
partial order reduction techniques that can effectively be used in
symbolic verification algorithms for equivalence properties
of security protocols.
We now discuss related work, and there is a lot of it given the
huge success of POR techniques in various application areas.
We shall focus on the novel aspects of our approach,
and explain why such techniques have not been needed outside
of security protocol analysis. These observations are not new:
as pointed out by Baier and Katoen~\cite{ModelCheckingBook},
``[POR] is mainly appropriate to control-intensive applications and less
suited for data-intensive applications'';
Clarke \emph{et al.}~\cite{clarke2000partial} also remark that
``In the domain of model checking of
reactive systems, there are numerous techniques for reducing
the state space of the system. One such technique
is partial-order reduction. This technique does not directly
apply to [security protocol analysis] 
because we explicitly keep
track of knowledge of various agents, and our logic can
refer to this knowledge in a meaningful way.''
We first compare our work with classical POR techniques, and 
then comment
on previous work in the domain of security protocol analysis.

\subsection{Classical POR}

Partial order reduction techniques
have proved very useful in the domain of model checking concurrent programs.
Given a Labelled Transition System (LTS) and some property to check
(\eg a Linear Temporal Logic formula), 
the basic idea of POR~\cite{Peled98,godefroid1996partial,ModelCheckingBook}
is to only consider a reduced version of the given LTS whose transitions of some states might be not exhaustive
but are such that this transformation does not affect the property.
POR techniques can be categorized in two groups~\cite{godefroid1996partial}.
First, the {\em persistent set} techniques
(\eg \emph{stubborn sets, ample sets}) where only a
sufficiently representative subset of available transitions is explored.
Second, {\em sleep set} techniques memoize past exploration
and use this information along with available transitions to
disable some provably redundant transitions.
Note that these two kinds of techniques are compatible, and are indeed
often combined to obtain better reductions.
Theoretical POR techniques apply to transition systems which may not be
explicitly available in practice, or whose explicit computation may be
too costly. In such cases, POR is often applied to an approximation of
the LTS that is obtained through static analysis. Another, more recent
approach is to use
\emph{dynamic 
POR}~\cite{flanagan2005dynamic,tasharofi2012transdpor,abdulla2014optimal}
where the POR arguments are applied
based on information that is obtained during the execution of the
system.

Clearly, classical POR techniques would apply
to our concrete LTS, but that would not be practically useful since this
LTS is wildly infinite, taking into account all recipes that the attacker
could build.
Applying most classical POR techniques to the LTS from which
data would have been abstracted away would be ineffective:
any input would be dependent \stef{on} any output (since the attacker's
knowledge, increased by the output, may enable new input messages).
Our compression technique lies between these two extremes. It exploits
a semi-commutation property: outputs can be permuted before inputs,
but not the converse in general.
Further, it exploits the fact that inputs
do not increase the attacker's knowledge, and can thus be executed
in a chained fashion, under focus.
The semi-commutation is reminiscent of the asymmetrical dependency
analysis enabled by the \emph{conditional} stubborn set
technique~\cite{godefroid1996partial},
and the execution of inputs under focus may be explained
by means of sleep sets.
While it may be possible to formally derive our compressed semantics
by instantiating abstract POR techniques to our setting,
we have not explored this possibility in detail\footnote{
  Although this would be an interesting question, we do not expect
  that any improvement of compression would come out of it.
  Indeed, compression can be argued to be maximal in terms of
  eliminating redundant traces without analysing data:
  for any compressed trace there is a way to choose messages
  and modify tests to obtain a concrete execution
  which does not belong to the equivalence class
  of any other compressed trace.
}.
Concerning
our reduced semantics, it may be seen as an application of the sleep set
technique (or even as a reformulation of Anisimov's and Knuth's
characterization of lexicographic normal forms) but the real contribution
with this technique is to have formulated it in such a way 
(see Definition~\ref{def:alldep}) that it can
be implemented without requiring an \emph{a priori} knowledge of
data dependencies: it allows us to eliminate redundant traces on-the-fly
as data (in)dependency is discovered by the constraint resolution procedure
(as explained in Section~\ref{subsec:apte-dependency})---in this sense, it may be viewed as a case of dynamic POR.

Narrowing the discussion a bit more, we now focus on the fact that
our techniques are designed for the verification of equivalence properties.
This requirement turns several seemingly trivial observations into
subtle technical problems. For instance,
ideas akin to compression are often applied without justification
(\eg in~\cite{sen2006automated,tasharofi2012transdpor,ModersheimVB10})
because they are obvious when one does
reachability rather than equivalence checking.
To understand this,
it is important to distinguish between two very
different ways of applying POR to equivalence checking
(independently of the precise equivalence under consideration).
The first approach is to reduce a system such that the reduced system
and the original systems are equivalent.
In the second approach, one only requires that two reduced
systems are equivalent iff the original systems are equivalent.
The first approach seems to be more common in the POR 
\stef{literature} (where one finds, \eg reductions that preserve
LTL-satisfiability~\cite{ModelCheckingBook}
or bisimilarity~\cite{huhn98fsttcs})
though there are instances of the second approach
(\eg for Petri nets~\cite{godefroid1991using}).
In the present work, we follow the second approach:
neither of our two reduction techniques preserves trace equivalence.
This allows stronger reductions but requires extra care:
one has to ensure that the independencies used in the reduction of one
process are also meaningful for the other processes; in other words,
reduction has to be symmetrical.
We come back to these two different approaches later, when discussing
specific POR techniques for security.

\subsection{Security applications}

The idea of applying POR to the verification of security protocols dates
back, at least, to the work of
Clarke \emph{et al.}~\cite{clarke2000partial,ClarkeJM03}.
In this work, the authors remark that traditional POR techniques
cannot be directly applied to security mainly because
``[they] must keep track of knowledge of various agents''
and ``[their] logic can refer to this knowledge in a meaningful way''.
This led them to define a notion of {\em semi-invisible actions}
(output actions, that cannot be swapped after inputs but only before them)
and design a reduction that prioritizes outputs
and performs them in a fixed order.
Compared to our work, this reduction is much weaker
(even weaker than compression only),
only handles a finite set of messages, and
only focuses on reachability properties checking.

In \cite{escobar14ic}, the authors develop ``state space reduction''
techniques for the Maude-NRL  Protocol Analyzer (Maude-NPA). This tool proceeds
by backwards reachability analysis and treats at the same level
the exploration of protocol executions and attacker's deductions.
Several reductions techniques are specific to this setting,
and most are unrelated to partial order reduction in general,
and to our work in particular.
We note that the lazy intruder techniques from \cite{escobar14ic}
should be compared to what is done in constraint resolution procedures
(\eg the one used in \apte) rather than to our work.
A simple POR technique used in Maude-NPA is based on the observation
that inputs can be executed in priority in the backwards exploration,
which corresponds to the fact that we can execute outputs first in
forward explorations. We note again that this is only one aspect
of the focused strategy, and that it is not trivial to lift this observation
from reachability to trace equivalence.
Finally, a ``transition subsumption'' technique is described for Maude-NPA.
While highly non-trivial due to the technicalities of the model,
this is essentially a tabling technique rather than a partial order
reduction. Though it does yield a significant state space reduction
(as shown in the experiments \cite{escobar14ic}) it falls short of
exploiting independencies fully, and has a potentially high computational
cost (which is not evaluated in the benchmarks of 
\cite{escobar14ic}).

In \cite{fokkink2010partial},
Fokkink {\it et al}.~model security protocols as labeled transition
systems whose states contain the control points of different agents
as well as previously outputted messages.
They devise some POR technique for these transition systems,
where output actions are prioritized and performed in a fixed order.
In their work, the original and reduced systems
are trace equivalent {\em modulo} outputs
(the same traces can be found after removing output actions).
The justification for their reduction would fail in our setting,
where we consider standard trace equivalence with observable outputs.
More importantly, their requirement that a reduced system should be
equivalent to the original one makes it impossible to swap input actions,
and thus reductions such as the execution under focus of our
compressed semantics cannot be used.
The authors leave as future work the problem of combining
their algorithm with symbolic executions, in order to be able
to lift the restriction to a finite number of messages.

Cremers and Mauw proposed~\cite{cremers2005checking} a reduction technique
for checking secrecy in security protocols. Their method allows to perform
outputs eagerly, as in our compressed semantics. It also uses a form of
\emph{sleep set} technique to avoid redundant interleavings of input
actions.
%
%
In addition to being applicable only for reachability property,
the algorithm of \cite{cremers2005checking} works under the assumption that for each
input only finitely \stef{many} input messages need to be considered.
The authors identify as important future work the
need to lift their method to the symbolic setting.

\stef{Earlier work by Mödersheim {\it et al.}
has shown how to combine POR techniques with symbolic semantics~\cite{ModersheimVB10}
in the context of reachability properties for security protocols.
This has led to high efficiency gains in the OFMC tool of the
AVISPA platform~\cite{sysdesc-CAV05}.
While their reduction is very
limited, it brings some key insight on how POR may be combined
with symbolic execution.
For instance, their reduction imposes a dependency constraint
\david{(called \emph{differentiation} constraint in their work)}
on the interleavings of 
$$\{\In(c,x).\Out(c,m),\,\In(d,y).\Out(d,m')\}.$$ Assuming that priority
is given to the process working on channel $c$,
this constraint enforces that any symbolic interleaving of the form
${\In(d,M').\Out(d,w').\In(c,M).\Out(c,w)}$
would only be explored for instances of $M$ that depend on $w'$.
Our reduced semantics constrains patterns of arbitrary size (instead of
just size~2 diamond patterns as above) by means of dependency
constraints. Going back to Example~\ref{ex:taille-trois}, their technique will 
only be able (at most) to exploit the dependencies depicted in plain blue arrows,
and they will not consider the one represented by the dashed 2-arrow.
\david{
Moreover, while we generate dependency constraints on the fly, they
implement their technique by looking for such a pattern afterwards.
}
This causes \lucca{a tradeoff between reduction and the cost of redundancy 
detection}: \lucca{their technique fails} to detect all
patterns of this kind. 
\david{Besides these differences, we note that Mödersheim et al.\ use
dependency constraints to guide a dedicated constraint resolution procedure, 
while we chose to treat constraint resolution (almost) as a black box,
and leave it unchanged.}
\david{Finally, we recall that} our POR technique has \stef{been} designed to 
be sound and complete for trace equivalence checking as well \david{as reachability 
checking}.
}

Finally, in \cite{BDH-concur15}, the authors of the present paper
extend some of the results presented here. Instead of considering
the syntactic fragment of simple processes, we work under the
more general semantical assumption of \emph{action-determinism}.
We show that compression and reduction can be extended to that case,
preserving the main result: the induced equivalences coincide.
However, that work is completely carried out in concrete rather than
symbolic semantics. Thus, this development should be viewed as being
orthogonal to the one carried out in the present paper. The main
ideas behind the integration in symbolic semantics and \apte
would apply to the action-deterministic case as well.
The line of research followed in \cite{BDH-concur15},
that consists in extending the supported fragment
for our POR techniques, is still open: it would be interesting to
support processes that are not action-deterministic,
which are commonplace when analysing anonymity or unlinkability
scenarios.

%% file: conclu.tex
\section{Conclusion}
\label{sec:conclu}

We have developed two POR techniques that are adequate for verifying
trace equivalence properties between simple
processes. 
The first refinement groups actions in blocks, while
the second one uses dependency constraints to restrict to minimal 
interleavings among a class of permutations. In both cases, the
refined semantics has less traces, yet we show that the associated
trace equivalence coincides with the standard one. 
We have effectively implemented these refinements in \apte,
and shown that they yield the expected, significant benefit.

We claim that our POR techniques – at least compression – and the significant optimisations
they allow are generic enough to be applicable to other verification methods as long as they
perform forward symbolic executions. In addition to the integration in Apte we have extensively
discussed, we also have successfully done so in Spec~\cite{SPEC-red}. Furthermore, parts of our POR techniques
have been independently integrated and implemented in the distributed version of
\akiss\footnote{See \url{https://github.com/akiss/akiss}.}.

We are considering several directions for future work.
Regarding the theoretical results presented here, it is actually possible to
slightly relax the syntactic condition we imposed on processes by an
action-determinism hypothesis and apply our reduction techniques on
replicated processes~\cite{BDH-concur15}. The question of
whether the action-determinism condition can be removed without degrading the
reductions too much is left open.
Another interesting direction would be to adapt our techniques for verification
methods based on {\em backward \stef{search}} instead of {\em forward \stef{search}} as 
is the case in this paper.
We also believe that stronger reductions can be achieved:
for instance, exploiting symmetries should be very useful for
dealing with multiple sessions.
Regarding the practical application of our results, we can certainly
go further. We could investigate the role of the particular choice of
the order $\prec$, to determine heuristics for maximising the
practical impact of reduction.

%% file: notations.tex
\section{Notations}
\label{sec:not}

\begin{center}
\begin{tabular}{|c|l|l|}
\hline
Symbol & Description & Reference
\\
\hline\hline
$\sint{}$ &
transition for concrete processes &
\hyperref[fig:concrete-semantics]{Figure~\ref{fig:concrete-semantics}}
\\
$\sintw{}$ &
$\sint{}$ up-to non-observable actions &
\hyperref[subsec:semantics]{Section~\ref{subsec:semantics}}
\\
$\statequiv$ &
static equivalence &
\hyperref[def:statequiv]{Definition~\ref{def:statequiv}}
\\
$\inceint$ &
trace inclusion for concrete processes &
\hyperref[def:concrete-equivalence]{Definition~\ref{def:concrete-equivalence}}
\\
$\eint$ &
trace equivalence for concrete processes &
\hyperref[def:concrete-equivalence]{Definition~\ref{def:concrete-equivalence}}
\\
\hline
$\sintf{}{}$ &
focused semantics &
\hyperref[fig:sintf]{Figure~\ref{fig:sintf}}
\\
$\sintc{}$ &
compressed semantics &
\hyperref[fig:sintc]{Figure~\ref{fig:sintc}}
 \\
$\inceintc$ &
trace inclusion induced by $\sintc{}$  &
\hyperref[def:compressed-equivalence]{Definition~\ref{def:compressed-equivalence}}
\\
$\eintc$ &
trace equivalence induced by $\sintc{}$ &
\hyperref[def:compressed-equivalence]{Definition~\ref{def:compressed-equivalence}}
\\
\hline\hline
$\ssym{}$ &
symbolic semantics &
\hyperref[figure:symbolic-sem]{Figure~\ref{figure:symbolic-sem}}
\\
$\ssymf{}{}$ &
focused symbolic semantics &
\hyperref[figure:compressed-symbolic-sem]{Figure~\ref{figure:compressed-symbolic-sem}}
\\
$\ssymc{}$ &
compressed symbolic semantics &
\hyperref[figure:compressed-symbolic-sem]{Figure~\ref{figure:compressed-symbolic-sem}}
\\
$\incesym$ &
trace inclusion induced by $\ssym{}$ &
\hyperref[def:equiv-symb]{Definition~\ref{def:equiv-symb}}
\\
$\incesymc$ &
trace inclusion induced by $\ssymc{}$ &
\hyperref[def:equiv-comp-symb]{Definition~\ref{def:equiv-comp-symb}}
\\
$\esym$ &
trace equivalence induced by $\ssym{}$ &
\hyperref[def:equiv-symb]{Definition~\ref{def:equiv-symb}}
\\
$\esymc$ &
trace equivalence induced by $\ssymc{}$ &
\hyperref[def:equiv-comp-symb]{Definition~\ref{def:equiv-comp-symb}}
\\
\hline
$\constrd{\vect X}{\vect w}$ &
dependency constraint &
\hyperref[def:dep-constr]{Definition~\ref{def:dep-constr}}
\\
$\AllDep{\tr}$ &
dependency constraints induced by a trace &
\hyperref[def:alldep]{Definition~\ref{def:alldep}}
\\
$\incesymdff$ &
trace inclusion up-to dependency constraints &
\hyperref[def:reduced-equivalence]{Definition~\ref{def:reduced-equivalence}}
\\
$\esymdff$ &
trace equivalence up-to dependency constraints  &
\hyperref[def:reduced-equivalence]{Definition~\ref{def:reduced-equivalence}}
\\
$\mid\mid$ &
independence of blocks &
\hyperref[def:independence]{Definition~\ref{def:independence}}
\\
$\equiv_\Phi$ &
equivalence of two traces &
\hyperref[def:equiv-phi]{Definition~\ref{def:equiv-phi}}
\\
\hline\hline
$\extof{\triple{\p}{\Phi}{\Set}}$ & associated extented symbolic
                         process
&
\hyperref[subsec:nutshell]{Section~\ref{subsec:nutshell}}
\\
$\ofext{\quadr{\p}{\Phi}{\Set}{\Set^+}}$ & associated symbolic process
&
\hyperref[subsec:nutshell]{Section~\ref{subsec:nutshell}}
\\
$\symbinclset$ &
symbolic inclusion of sets of extended symbolic processes &
\hyperref[def:symbeqset]{Definition~\ref{def:symbeqset}}
\\
$\symbeqset$ &
symbolic equivalence of sets of extended symbolic processes &
\hyperref[def:symbeqset]{Definition~\ref{def:symbeqset}}
\\
$\ssyma{}$ &
\apte exploration step &
\hyperref[subsec:nutshell]{Section~\ref{subsec:nutshell}}
\\
$\ssymaone{}$ &
first part of \apte exploration step &
\hyperref[subsec:nutshell]{Section~\ref{subsec:nutshell}}
\\
$\ssymatwo{}$ &
second part of \apte exploration step &
\hyperref[subsec:nutshell]{Section~\ref{subsec:nutshell}}
\\
$\ssymac{}$ &
compressed version of $\ssyma{}$ &
\hyperref[subsec:apte-compression]{Section~\ref{subsec:apte-compression}}
\\
$\ssymad{}$ &
reduced version of $\ssymac{}$ &
\hyperref[subsec:apte-dependency]{Section~\ref{subsec:apte-dependency}}
\\
$\eqapte$ &
equivalence induced by $\ssyma{}$ &
\hyperref[def:eqapte]{Definition~\ref{def:eqapte}}
\\
$\eqaptec$ &
equivalence induced by $\ssymac{}$ &
\hyperref[def:subsec:apte-compression]{Section~\ref{subsec:apte-compression}}
\\
$\eqapted$ &
equivalence induced by $\ssymad{}$ &
\hyperref[subsec:apte-dependency]{Section~\ref{subsec:apte-dependency}}
\\
\hline
\end{tabular}
\end{center}

%% file: app-compression.tex
\section{Proofs of Section~\ref{sec:compression}}
\label{app:compression}

\renewcommand{\LRstep}[1]{\xRightarrow{#1}}


\lempermute*
\begin{proof}
  It suffices to establish that
  $A \LRstep{\alpha\cdot\alpha'\;} A'$ implies
  $A \LRstep{\alpha'\cdot\alpha\;} A'$
  for any $\alpha\;\I_{a}\;\alpha'$.
\begin{itemize}
\item
  \newcommand{\acti}{\Out(c_i,w_i)}
  \newcommand{\actj}{\Out(c_j,w_j)}
  Assume that we have
  $A \LRstep{\acti} A_i \LRstep{\actj} A'$ with $c_i \neq c_j$.
  Because we are considering simple processes,
  the two actions must be concurrent.
  More specifically, our process $A$ must be of the form
  $\proc{\{P_i, P_j\}\uplus\p_r}{\Phi}$ with $P_i$ (resp.~$P_j$)
  being a basic process on channel $c_i$ (resp.~$c_j$).
  We assume that in our sequence of reductions, $\tau$ actions
  pertaining to $P_i$ are all executed before reaching $A_i$,
  and that $\tau$ actions pertaining to $\p_r$ are executed last.
  This is without loss of generality, because a $\tau$ action
  on a given basic process can easily be permuted with actions
  taking place on another basic process, since it does not depend
  on the context and has no effect on the frame.
  Thus we have that
  $A_i = \proc{\{P'_i,P_j\}\uplus\p_r}{\Phi\uplus\{w_i\refer m_i\}}$,
  $A'  = \proc{\{P'_i;P'_j\}\uplus\p'_r}{\Phi\uplus\{w_i\refer m_i, w_j\refer 
  m_j\}}$. Since the $\tau$ actions taking place on $\p'_r$ rely neither
  on the frame nor on the first two basic processes, we easily obtain
  the permuted execution:
  \[ \begin{array}{rcl}
    A
    & \LRstep{\actj} &
    \proc{\{P_i,P'_j\}\uplus\p_r}{\Phi\uplus\{w_j\refer m_j\}} \\
    & \LRstep{\acti} &
    \proc{\{P'_i,P'_j\}\uplus\p'_r}{\Phi\uplus\{w_i\refer m_i,w_j \refer 
    m_j\}}
  \end{array} \]

\item
  The permutation of two input actions on distinct channels is very
  similar. In this case, the frame does not change at all, and the
  order in which messages are derived from the frame does not matter.
  Moreover, the instantiation of the input variable on one basic
  process has no impact on the other ones.

\item
  \renewcommand{\acti}{\Out(c_i,w_i)}
  \renewcommand{\actj}{\In(c_j,M)}
  Assume that we have
  $A \LRstep{\acti} A_i \LRstep{\actj} A'$ with $c_i \neq c_j$ and 
  $w_i\not\in \fv(M)$.
  Again, the two actions are concurrent, and we can assume that
  $\tau$ actions are organized conveniently so that
  $A$ is of the form $\proc{\{P_i,P_j\}\uplus\p_r}{\Phi}$
  with $P_i$ (resp. $P_j$) a basic process on $c_i$ (resp. $c_j$);
  $A_i$ is of the form $\proc{\{P'_i,P_j\}\uplus\p_r}{\Phi\uplus\{w_i\refer 
  m_i\}}$;
  and $A'$ is of the form
  $\proc{\{P'_i,P'_j\}\uplus\p'_r}{\Phi\uplus\{w_i\refer m_i\}}$.
  As before, the $\tau$ actions from $\p_r$ to $\p'_r$ are easily moved 
  around. \stef{Additionally}, $w_i \not\in \fv(M)$ implies
  $\fv(M)\subseteq\dom(\Phi)$ and thus we have:
  \[
    \proc{\{P_i,P_j\}\uplus\p_r}{\Phi}
    \LRstep{\actj} \proc{\{P_i,P'_j\}\uplus\p_r}{\Phi} \]
  The next step is trivial:
  \[ \proc{\{P_i,P'_j\}\uplus\p_r}{\Phi}
     \LRstep{\acti}
   \proc{\{P'_i,P'_j\}\uplus\p'_r}{\Phi\uplus\{w_i\refer m_i\}} \]

\item
  We also have to perform the reverse permutation, but we shall
  not detail it; this time we are delaying the derivation of $M$ from the 
  frame, and it only gets easier.\qedhere
\end{itemize}
\end{proof}



\procompcompleteness*

\begin{proof}
  We first observe that
  $A \LRstep{\tr} A'$ implies $A \sintf{\tr}{o^*} A'$
  if $A'$ is initial and $\tr$ is a (possibly empty)
  sequence of output actions on the same channel.
  We prove this by induction on the sequence of actions.
  If it is empty, we can conclude using
  one of the \textsc{Proper} rules because $A = A'$ is initial.
  Otherwise, we have:
  $$A \LRstep{\Out(c,w)} A'' \LRstep{\tr\;} A'.$$
  We obtain $A'' \sintf{\tr}{o^*} A'$ by induction hypothesis,
  and conclude using rules \textsc{Tau} and \textsc{Out}.

  The next step is to show that
  $A \LRstep{\tr\;} A'$ implies $A \sintf{\tr}{i^{*}} A'$,
  if $A'$ is initial and $\tr$ is the
  concatenation of a (possibly empty) sequence of inputs
  and a non-empty sequence of outputs,
  all on the same channel.
  This is easily shown by induction on the number of input actions.
  If there \stef{are} none we use the previous result, otherwise we conclude
  by induction hypothesis and using rules \textsc{Tau} and \textsc{In}.
  Otherwise, the first output action allows us to conclude from the
  previous result and rules \textsc{Tau} and \textsc{Out}.

  We can now show that
  $A \LRstep{\tr\;} A'$ implies $A \sintf{\tr}{i^{+}} A'$
  if $A'$ is initial and $\tr$ is a proper block.
  Indeed, we must have $$A \LRstep{\In(c,M)} A'' \LRstep{\tr'\;} A'$$
  which allows us to conclude using the previous result and rules
  \textsc{Tau} and \textsc{In}.

  We finally obtain that $A \LRstep{\tr} A'$ implies
  $A \sintc{\tr} A'$ when $A$ and $A'$ are initial \emph{simple} processes
  and $\tr$ is a sequence of proper blocks. This is done by induction
  on the number of blocks. The base case is trivial.
  Because $A$ is initial, the execution of its basic processes can only
  start with observable actions, thus only one basic process is involved in
  the execution of the first block.
  Moreover, we can assume without loss of generality that the execution of
  this first block results in another initial process: indeed the basic process
  resulting from that execution is either in the final process $A'$, which is
  initial, or it will perform another block, \ie it can perform
  $\tau$ actions followed by an input, in which case we can force
  those $\tau$ actions to take place as early as possible. Thus we have
  \[ A \LRstep{\mathsf{b}} A'' \LRstep{\tr'} A' \]
  where $\mathsf{b}$ is a proper block, and we conclude using the previous
  result and the induction hypothesis.
\end{proof}

\stef{Before proving Theorem~\ref{theo:comp-soundness-completeness}, we establish the following result.}

\begin{proposition} \label{prop:focus-trace}
  Let $\tr$ be a trace of observable actions such that,
  for any channel $c$ occurring in the trace, it appears first
  in an input action.
  There exists a sequence of proper blocks $\tr_{io}$ and
  a sequence of improper blocks
  $\tr_i$
  such that $\tr =_{\I_a} \tr_{io}\cdot\tr_i$.
\end{proposition}

\begin{proof}
  We proceed by induction on the length of $\tr$, and distinguish
  two cases:
\begin{itemize}
\item If $\tr$ has no output action then, by swapping input actions
  on distinct channels,
  we reorder $\tr$ so as to obtain
  $\tr_i = \tr^{c_1}\cdot\ldots\cdot\tr^{c_n}=_{\I_a}\tr$
  where the $c_i$'s are pairwise distinct and
  $\tr^{c_i}$ is an improper block on channel $c_i$.
\item Otherwise,
  there must be a decomposition $\tr = \tr_1\cdot\out{c}{w}\cdot\tr_2$
  such that
  $\tr_1$ does not contain any output. We can perform swaps involving
  input actions of $\tr_1$ on all channel $c'\neq c$, so that they
  are delayed after the first output on $c$.
  We obtain $\tr =_{\I_a}
  \inp{c}{M_1}\cdot\ldots\cdot\inp{c}{M_n}\cdot\out{c}{w}
  \cdot \tr_1'\cdot \tr_2$
  with $n\geq 1$.
  Next,
  we swap output actions on channel $c$ from $\tr_1'\cdot\tr_2$ that
  are not preceded by another input on $c$, so as to obtain
  \[ \tr =_{\I_a}
  \inp{c}{M_1}\dots\inp{c}{M_n}\cdot\out{c}{w}
  \cdot\out{c}{w_1}\dots\out{c}{w_m}\cdot\tr_2' \] such that
  either $\tr_2'$ does not contain any action on channel $c$
  or the first one is an input action.
  We have thus isolated a first proper block,
  and we can conclude by induction hypothesis on $\tr_2'$.\qedhere
\end{itemize}
\end{proof}

Note that the above result does not exploit all the richness of $\I_a$.
In particular, it never relies on the possibility to swap an input action
before an output when the input message does not use the output handled.
Indeed, the idea behind compression does not rely on messages.
This is no longer the case in Section~\ref{sec:diff} where we use ${\I_a}$
more fully.

\smallskip


We finally prove the main result about the compressed semantics
\stef{relying on Proposition~\ref{prop:focus-trace} stated and proved above.}
Given two simple process $A = \proc{\p}{\Phi}$ and $A' =
\proc{\p'}{\Phi'}$, we shall write $\Phi(A) \statequiv \Phi(A')$ (or
even $A \statequiv A'$) instead of $\Phi
  \statequiv \Phi'$.

\theocompsoundnesscompleteness*

\begin{proof} We prove the two directions separately.

\noindent $(\Rightarrow)$ 
Let $A$ be an initial simple process such that $A \approx B$
and ${A \sintc{\tr} A'}$.
One can easily see that the trace $\tr$ must be of the form
$\tr_{io} \cdot \tr_i$ where $\tr_{io}$ is made of proper blocks
and $\tr_i$ is a (possibly empty) sequence of inputs
on the same channel $c_j$.
We have: $$A \sintc{\tr_{io}} A'' \sintc{\tr_i} A'$$

Using Proposition~\ref{pro:comp-soundness}, we obtain that
$A \LRstep{\tr_{io}} A''$.
We also claim that $A''\LRstep{\tr_i} A^{+}$
for some $A^{+}$ having the same frame as $A'$ which is itself equal to the one of $A''$.
This is obvious when $\tr_i$ is empty---in that case we can simply choose $A^{+} = A' = A''$.
Otherwise, the execution of the improper block $\tr_i$ results from the 
application of rule \textsc{Improper}. Except for the fact that this rule
``kills'' the resulting process, its subderivation simply packages a
sequence of inputs, and so we have a suitable $A^+$.
We thus have:
\[
A\LRstep{\tr_{io}\;}
A''\LRstep{\tr_i\;}
A^{+} \]

By hypothesis, it implies that
$B \LRstep{\tr_{io}\;} B''$
and $B\LRstep{\tr_{io}\cdot\tr_i\;} B^+$
with $A'' \statequiv B''$ and ${A^+ \statequiv B^+}$.
Relying on the fact that $B$ is a simple process, we have:
\[ B \LRstep{\tr_{io}} B'' \LRstep{\tr_i} B^+ \]

It remains to establish that $B \sintc{\tr} B'$
such that $B'\statequiv A'$.
We can assume that $B''$ does not have any basic process
starting with a test, without loss of generality since forcing
$\tau$ actions cannot break static equivalence.
Further, we observe that $B''$ is initial. Otherwise, it
would mean that a basic process of $B$ is not initial (absurd)
or that one of the blocks of $\tr_{io}$, which are maximal for
$A$, is not maximal for $B$ (absurd again, because it contradicts
$A \approx B$).
This allows us to apply Proposition~\ref{pro:comp-completeness}
to obtain $$B \sintc{\tr_{io}} B''.$$
This concludes when $\tr_i$ is empty, because $B' = B'' \statequiv A'' = A'$.
Otherwise, we note that $A^+$ cannot perform any action on channel $c_j$,
because the execution of $\tr_i$ in the compressed semantics must be maximal.
Since $A \approx B$, it
must be that $B^+$ cannot perform any observable action on the channel $c_j$
either. Thus $B''$ can complete an improper step:
\[ B''\sintc{\tr_i}B' \mbox{ where } B' = \proc{\varnothing}{\Phi(B^+)}. \]
We can finally conclude that
$B \sintc{\tr} B'$ with $\Phi(B') = \Phi(B^+) \statequiv \Phi(A^+) = \Phi(A')$.

\bigskip{}

\noindent $(\Leftarrow)$
Let $A$ be an initial simple process such that $A \eintc B$
and $A \LRstep{\tr} A'$. We ``complete'' this execution as follows:
\begin{itemize}
\item We force $\tau$ actions whenever possible.
\item If the last action on $c$ in $\tr$ is an input, we trigger
  available inputs on $c$ using a valid public constant as a
  recipe. 
\item We trigger all the outputs that are available and not blocked.
\end{itemize}

We obtain a trace of the form $\tr\cdot\tr^{+}$.
Let $A^{+}$ be the process obtained from this trace:
\[ A \LRstep{\tr} A' \LRstep{\tr^{+}} A^{+} \]
We observe that $A^{+}$ is initial. Indeed, for each basic process
that performs actions in $\tr\cdot\tr^{+}$, we have that: 
\begin{itemize}
\item either the last action on its
channel is an output and the basic process is of the form
$\In(c,\_).P$ or $\Out(c,u).P$ with $\neg \valid(u)$, 
\item or the last action is an input and the basic
process is reduced to $0$ and disappears, or it is an output which is
blocked.
\end{itemize}

Next, we apply Proposition~\ref{prop:focus-trace} to obtain traces
$\tr_{sio}$ (resp.~$\tr_i$) made of proper (resp.~improper) blocks,
such that $\tr\cdot\tr^{+} =_{\I_{a}} \tr_{io}\cdot\tr_i$.
By Lemma~\ref{lem:permute} we know that this permuted trace
can also lead to $A^{+}$:
\[
  A \LRstep{\tr_{io}} A_{io} \LRstep{\tr_i} A^{+}
\]
As before, we can assume that $A_{io}$ cannot perform any $\tau$
action. Under this condition, since~$A^{+}$ is initial,
$A_{io}$ must also be initial.

By Proposition~\ref{pro:comp-completeness} we have that
$A \sintc{\tr_{io}} A_{io}$,
and $A \eintc B$ implies that:
$$B \sintc{\tr_{io}} B_{io} \mbox{ with } \Phi(A_{io}) \statequiv \Phi(B_{io}). $$
A simple inspection of the \textsc{Proper} rules shows that a basic
process resulting from the execution of a proper block must be
initial. Thus, since the whole simple process $B$ is initial,
{$B_{io}$ is initial too}.

Thanks to Proposition~\ref{pro:comp-soundness}, we have that
$B \LRstep{\tr_{io}} B_{io}$.
Our goal is now to prove that we can complete this execution with $\tr_i$.
This trace is of the form $\tr^{c_1}\cdot\tr^{c_2}\dots\tr^{c_n}$ where $\tr^{c_i}$ contains
only inputs on channel $c_i$ and the $c_i$ are pairwise disjoint.
Now, we easily see that for each $i$,
$$A_{io} \LRstep{\tr^{c_i}} A_i$$
and $A_i$ has no more atomic process on channel $c_i$.
Thus we have $A_{io} \sintc{\tr^{c_i}} A_i^0$ with
$A_i^0 = \proc{\varnothing}{\Phi(A_i)}$. Since
$A \approx_c B$, we must have some $B_i^0$ such that:
$$B \sintc{\tr_{io}} B_{io} \sintc{\tr^{c_i}} B_i^0 $$
We can translate this back to the regular semantics, obtaining
$
B \LRstep{\tr_{io}} B_{io} \LRstep{\tr^{c_i}} B_i
$.
We can now execute all these
inputs to obtain an execution of $\tr_{io}\cdot\tr_i$ towards some process
$B^{+}$:
$$B \LRstep{\tr_{io}} B_{io} \LRstep{\tr_i} B^{+}$$
Permuting those actions, we obtain thanks to Lemma~\ref{lem:permute}:
\[ B \LRstep{\tr} B' \LRstep{\tr^{+}} B^{+} \]
We observe that
$\Phi(B^{+}) = \Phi(B_{io}) \statequiv \Phi(A_{io}) \statequiv \Phi(A^{+})$,
and it 
 follows that ${A' \statequiv B'}$ because those frames
have the same domain, which is a subset of that of ${\Phi(A^{+}) \statequiv \Phi(B^{+})}$.
\end{proof}